\documentclass[12pt]{article}
\usepackage{authblk}
\usepackage{hyperref}
\usepackage{url}
\usepackage{array}
\usepackage{makecell}
\usepackage{amsmath,amsfonts,amssymb,amsthm}
\usepackage{graphicx}
\usepackage{totcount}
\usepackage[round]{natbib}
\usepackage{lineno}
\usepackage{algorithm}
\usepackage{algpseudocode}
\newenvironment{varalgorithm}[1]
  {\algorithm\renewcommand{\thealgorithm}{#1}}
  {\endalgorithm}

\usepackage{booktabs} 
\usepackage{subcaption}
\usepackage{bbm}
\usepackage{bm}
\usepackage{mwe}
\usepackage{paralist} 
\usepackage{indentfirst} 
\usepackage{xcolor}

\usepackage[font=footnotesize,labelfont=bf]{caption}

\usepackage[margin=1in]{geometry}

\DeclareMathOperator*{\argmax}{arg\,max}
\DeclareMathOperator*{\argmin}{arg\,min}

\title{Prediction-Powered Inference\\ with Imputed Covariates and Nonuniform Sampling}

\author[1]{\small Dan M. Kluger \thanks{Corresponding author: dkluger@mit.edu}}
\author[2]{\small Kerri Lu}
\author[3]{\small Tijana Zrnic}
\author[1]{\small Sherrie Wang}
\author[2]{\small Stephen Bates}
\affil[1]{\scriptsize Institute for Data, Systems, and Society, Massachusetts Institute of Technology, Cambridge MA, 02139}
\affil[2]{\scriptsize Department of Electrical Engineering and Computer Science, Massachusetts Institute of Technology, Cambridge MA, 02139}
\affil[3]{\scriptsize Stanford Data Science and Department of Statistics, Stanford University, Stanford CA, 94305}
\date{}

\newcommand*{\horzbar}{\rule[.5ex]{2.5ex}{0.5pt}}

\newcommand{\rd}{\,\mathrm{d}}

\newcommand{\e}{\mathbb{E}}
\newcommand{\var}{\mathrm{Var}}
\newcommand{\cov}{\mathrm{Cov}}
\newcommand{\corr}{\mathrm{Corr}}

\newcommand{\tran}{\mathsf{T}}
\newcommand{\indep}{\raisebox{0.05em}{\rotatebox[origin=c]{90}{$\models$}}}

\newcommand{\ci}{\mathcal{I}}

\newcommand{\cx}{\mathcal{X}}

\newcommand{\xobs}{X^{\textnormal{obs}}}
\newcommand{\xna}{X^{\textnormal{miss}}}
\newcommand{\txna}{\tilde{X}^{\textnormal{miss}}}

\newcommand{\ntot}{N}
\newcommand{\nm}{N-n}
\newcommand{\nc}{n}

\newcommand{\htc}{\hat{\theta}^{\bullet}}
\newcommand{\hgc}{\hat{\gamma}^{\bullet}}
\newcommand{\hgm}{\hat{\gamma}^{\circ}}
\newcommand{\hga}{\hat{\gamma}^{\textnormal{all}}}

\newcommand{\htcarg}[1]{\hat{\theta}^{\bullet {#1}}}
\newcommand{\hgcarg}[1]{\hat{\gamma}^{\bullet {#1}}}
\newcommand{\hgmarg}[1]{\hat{\gamma}^{\circ {#1}}}

\newcommand{\bparen}{(b)}

\newcommand{\htPP}{\hat{\theta}^{\textnormal{PTD}}}
\newcommand{\htPPIplus}{\hat{\theta}^{\textnormal{PPI++}}}

\newcommand{\htPPom}{\hat{\theta}^{\textnormal{PTD},\Omega}}
\newcommand{\htPPhom}{\hat{\theta}^{\textnormal{PTD},\hat{\Omega}}}
\newcommand{\htPPhomCB}[1]{\hat{\theta}^{\textnormal{PTD},\hat{\Omega},\textnormal{CB},{#1}}}
\newcommand{\htPPhomOpt}{\hat{\theta}^{\textnormal{PTD},\hat{\Omega}_{\textnormal{opt}}}}

\newcommand{\htPPomStar}{\hat{\theta}^{\textnormal{PTD},\Omega,*}}
\newcommand{\htPPhomStar}{\hat{\theta}^{\textnormal{PTD},\hat{\Omega},*}}
\newcommand{\htPPhomB}{\hat{\theta}^{\textnormal{PTD},\hat{\Omega},\bparen}}
\newcommand{\htPPhomArg}[1]{\hat{\theta}^{\textnormal{PTD},\hat{\Omega},{#1}}}

\newcommand{\cwmain}{\mathcal{W}^{\circ}}
\newcommand{\cwcalib}{\mathcal{W}^{\bullet}}

\newcommand{\cwmainStar}{\mathcal{W}^{\circ,*}}
\newcommand{\cwcalibStar}{\mathcal{W}^{\bullet,*}}

\newcommand{\SigGM}{\Sigma_{\gamma}^{\circ}}
\newcommand{\SigGC}{\Sigma_{\gamma}^{\bullet}}
\newcommand{\SigTC}{\Sigma_{\theta}^{\bullet}}
\newcommand{\SigTGC}{\Sigma_{\theta,\gamma}^{\bullet}}

\newcommand{\hSigGM}{\hat{\Sigma}_{\gamma}^{\circ}}
\newcommand{\hSigGC}{\hat{\Sigma}_{\gamma}^{\bullet}}
\newcommand{\hSigTC}{\hat{\Sigma}_{\theta}^{\bullet}}
\newcommand{\hSigTGC}{\hat{\Sigma}_{\theta,\gamma}^{\bullet}}

\newcommand{\hSMatGM}{\hat{S}_{\gamma}^{\circ}}

\newcommand{\SigPTD}{\Sigma_{\textnormal{PTD}}}

\newcommand{\indc}{\mathcal{I}^{\bullet}}
\newcommand{\indm}{\mathcal{I}^{\circ}}
\newcommand{\indcStar}{\mathcal{K}^{\bullet,*}}
\newcommand{\indmStar}{\mathcal{K}^{\circ,*}}

\newcommand{\calA}{\mathcal{A}}

\makeatletter
\newcommand{\skipitems}[1]{\addtocounter{\@enumctr}{#1}}
\makeatother

\newcommand{\giv}{\!\mid\!} 

\theoremstyle{plain}  
\newtheorem{theorem}{Theorem}[section]
\newtheorem{lemma}{Lemma}[section]
\newtheorem{corollary}{Corollary}[section]
\newtheorem{proposition}{Proposition}[section]

\theoremstyle{remark}

\newtheorem{remark}{Remark}
\regtotcounter{assumption} 

\newtheorem{assumption}{Assumption}

\begin{document}

\maketitle

\vspace{-2.5em}
\begin{abstract}

Machine learning models are increasingly used to produce predictions that serve as input data in subsequent statistical analyses. 
For example, computer vision predictions of economic and environmental indicators based on satellite imagery are used in downstream regressions; similarly, language models are widely used to approximate human ratings and opinions in social science research.
However, failure to properly account for errors in the machine learning predictions renders standard statistical procedures invalid.
Prior work uses what we call the \emph{Predict-Then-Debias} estimator to give valid confidence intervals when machine learning algorithms impute missing variables, assuming a small \emph{complete} sample from the population of interest.
We expand the scope by introducing bootstrap confidence intervals that apply when the complete data is a nonuniform (i.e.,  weighted, stratified, or clustered) sample and to settings where an arbitrary subset of features is imputed. Importantly, the method can be applied to many settings without requiring additional calculations.
We prove that these confidence intervals are valid under no assumptions on the quality of the machine learning model and are no wider than the intervals obtained by methods that do not use machine learning predictions.

\end{abstract}


\section{Introduction} \label{sec:Intro}

With increasingly rich data collection and a surge in open science initiatives, investigators now frequently rely on machine learning algorithms to predict quantities of interest based on large collections of related but indirect measurements. For example, in remote sensing studies, land cover is predicted from satellite images with computer vision algorithms. Similarly, millions of protein structures have been predicted from their amino acid sequences, but such structures are rarely measured directly because it is costly and difficult. This situation is pervasive---investigators are now often faced with large data sets that partially consist of machine learning outputs.
However, using error-prone predictions from machine learning models as input data for statistical analysis (e.g., calculating regression coefficients) can lead to considerable biases and misleading conclusions. This has led to a flurry of recent works that develop methods to resolve the conundrum of how to use machine-learning predictions in a statistical analysis. In particular, when the analyst has access to data in which the predictions and the ground truth measurements are jointly available, it is possible to correct the bias of the machine learning model without assumptions on the quality of the predictions. In such settings, we develop methods to construct confidence intervals that account for errors in the machine learning predictions.

\subsection{Problem setup}


We first introduce our setting and two simple baselines. We assume a distribution over data points $X \equiv (\xobs, \xna) \in \mathbb{R}^{p}$. The features $\xobs$ are always observed, but measuring $\xna$ is costly (e.g.,  it requires expert human labeling). As such, few complete data points are available. In addition, we also have access to a much larger data set of machine-learning predictions of $\xna$, which we denote by $\txna$. The reader should interpret $\txna$ as a reasonably accurate but imperfect proxy of $\xna$. Formally, we have access to a small \textit{complete} ($\bullet$) sample $(\xobs_i, \xna_i,\txna_i)_{i \in \indc}$ and a much larger \textit{incomplete} ($\circ$) sample $(\xobs_i, \txna_i)_{i \in \indm}$, adding up to $N$ data points total. Here, $\indc$ and $\indm$ are disjoint sets of indices such that $\indc \cup \indm = \{1,\dots,\ntot\}$.
We denote by $\nc$ the size of the complete sample, $\nc = |\indc|$.

Our goal is to estimate some quantity $\theta \in \mathbb R^d$ describing the distribution of $X$. As our primary example in this work, we may wish to estimate the coefficients in a generalized linear model (GLM), obtained by regressing one component of $X$ on the remaining components. Note that the population regression coefficient vector for a GLM is still a well-defined estimand of interest even when a GLM does not accurately describe the distribution of $X$. Typically an investigator will have access to a function (implemented in software) $\calA(\cdot)$ which takes a sample of data as input and returns an estimate of $\theta$ as an output. They can, therefore, readily deploy two natural and simple approaches to estimating $\theta$:
\begin{enumerate}
	\item \textit{The naive approach}: Act as if $\xna=\txna$ and apply $\calA(\cdot)$ to all available samples imputed with machine learning predictions: $\hga = \calA \big( (\xobs_i, \txna_i)_{i=1}^{\ntot} \big)$.
    \item \textit{The classical approach}: Ignore the machine learning predictions and apply $\calA(\cdot)$ to the small complete sample only: $\htc = \calA \big((\xobs_i, \xna_i)_{i \in \indc} \big)$.
\end{enumerate}
Both approaches have both advantages and limitations. In particular, $\htc$ targets the correct parameter $\theta$, but it does not leverage the potential power of machine-learning predictions. Meanwhile, the naive approach uses the abundant machine-learning predictions, but it targets the wrong parameter; even with infinite data, $\hga$ will be biased, unless the proxy $\txna$ comes from the exact same distribution as $\xna$ (i.e., unless the machine learning model makes no errors). Consequently, confidence intervals based on the naive approach are invalid. We observe in our real-data experiments (Section~\ref{sec:RealDataExperiments}) that the naive estimator is sometimes biased by a factor of $1/3$ or more.

Given the limitations of the classical and naive approaches and the growing prevalence of machine-learning-imputed datasets, there is a growing interest in developing methods that leverage all available data, aiming for the best of both worlds. In this paper, we build on an easy-to-implement estimator proposed in \cite{ChenAndChen2000} and recently promoted in \cite{ChenAndChenLikeMedicalJournal}, \cite{YangAndDing2020}, \cite{kremers2021generalsimplerobustmethod}, \cite{PPBootNote}, \cite{MiaoLuNeurIPS}, and \cite{gronsbell2024PromotesCC}, which we will call the \emph{Predict-Then-Debias} (PTD) estimator. 

\begin{figure}[t]
    \centering 
    \includegraphics[width=0.95 \hsize]{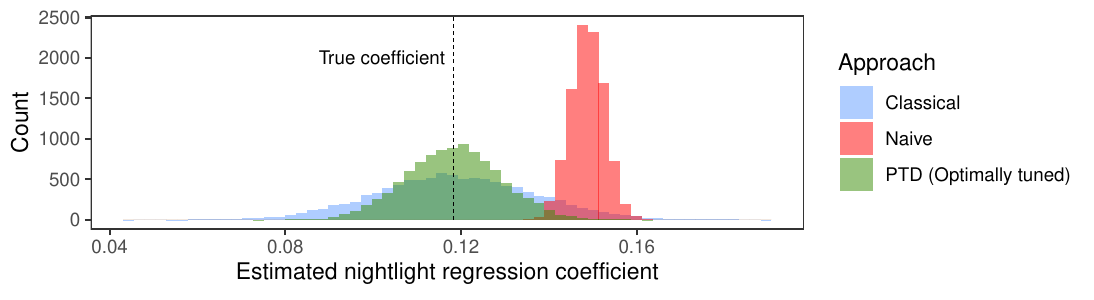}
    \caption{Histograms of nightlight coefficient estimator across $10{,}000$ simulations for three different estimation strategies. Each simulation used a random sample of size $40{,}000$ from the dataset from \cite{MOSAIKSPaper}, with $\nc=1{,}500$ samples randomly assigned to the complete sample. 
    The dashed vertical line gives the ``true coefficient" for nightlights based on fitting a regression using the gold standard data from all available samples. See Section \ref{sec:HousingPriceExample} for more details on the dataset and the regression setup.}
    \label{fig:TeaserFig}
\end{figure}

\paragraph{Preview example.} Figure \ref{fig:TeaserFig} previews a real-data example, in which the goal is to estimate the regression coefficient for nightlights in a regression of housing price on income, nightlights, and road density. 
The complete sample had accurate measurements of all variables, while the incomplete sample had accurate measurements of housing price and income, but used predictions of nightlights and road density based on daytime satellite imagery to impute the missing covariates.
The figure shows histograms for the three estimators across $10{,}000$ simulations. In each simulation, a random sample of size $40{,}000$ was drawn (with replacement) from the dataset of \cite{MOSAIKSPaper}, and $\nc=1{,}500$ of these samples were randomly assigned to the complete sample. 
The key takeaway is that the PTD estimator is approximately \emph{unbiased}, unlike the naive estimator, and has a \emph{lower variance} than the classical estimator. The naive approach was biased and consistently overestimated the nightlight regression coefficient. Meanwhile, the classical approach, which ignored the $38{,}500$ samples with proxy nightlight and road density measurements, had large variance---so much so that it sometimes produced estimates that were even less accurate than those of the naive approach.


\subsection{A flexible solution: Predict-Then-Debias (PTD) estimator}

The PTD estimator takes the naive estimator based on the incomplete sample---which is biased---and then adds a bias correction term based on the complete sample. In particular, let $\hgm = \calA \big( (\xobs_i,\txna_i)_{i \in \indm} \big)$, $\htc =\calA \big( (\xobs_i,\xna_i)_{i \in \indc} \big)$, and $\hgc=\calA \big( (\xobs_i,\txna_i)_{i \in \indc} \big)$; the first estimator is similar to the previously defined naive estimator and the second estimator is the classical estimator. The basic Predict-Then-Debias estimator is then defined as $$\htPP = \underbrace{\hgm}_{\text{biased, low-variance estimator}} + \underbrace{(\htc -\hgc)}_{\text{bias correction}}.$$

To study properties of $\htPP$, we presume standard regularity conditions, namely that $\htc$ will converge to $\theta$ as the sample size grows, and that $\hgc$ and $\hgm$ converge to some limit that we will denote by $\gamma$. 
Under such assumptions, $\htPP$ is \emph{consistent}: it converges to $\theta$ as $\nc  ,\ntot \to \infty$. Meanwhile, the variance of this estimator has the potential to be quite small when $\txna \approx \xna$. In particular, when the samples are independent, the variance of the $j$th 
coordinate of $\htPP$ is given by $$\var(\htPP_j)=\var(\hgm_j)+\var(\htc_j-\hgc_j).$$ The first term should be small, or even negligible, because $\hgm$ is estimated from a large sample. If $\txna \approx \xna$, then $\htc_j-\hgc_j$ will tend to be near zero, implying the second variance is also small.



We can reduce the variance further by considering a matrix-valued tuning parameter---for any (possibly data-dependent) tuning matrix $\hat{\Omega} \in \mathbb{R}^{d \times d}$, define
$$\htPPhom = \hat{\Omega} \hgm + (\htc -\hat{\Omega} \hgc).$$ If $\hat{\Omega}$ converges to some $\Omega$, $\htPPhom$ is again consistent, meaning it will converge to $\Omega\gamma + (\theta - \Omega \gamma)=\theta$ in the large-sample limit. 
One can check that $\var(\htPPhom_j)$ is minimized for each $j \in \{1,\dots,d\}$ when $\hat{\Omega}=\cov(\htc,\hgc) [\var(\hgc)+\var(\hgm)]^{-1}$.  This optimal tuning matrix is unknown \emph{a priori} but can be estimated with data.
We refer to $\htPPhom$ as the \emph{tuned} PTD estimator (or sometimes just the PTD estimator for short). 


While the PTD estimator is straightforward to compute and has desirable consistency and efficiency properties, 
streamlined approaches for constructing confidence intervals for the PTD estimator are less well-developed. One approach is to derive a central limit theorem for $\htPPhom$ as well as a formula for its asymptotic variance, and then use the normal approximation to construct a confidence interval. In the context of GLM coefficient estimators, this approach was taken by \cite{ChenAndChen2000}, among others, but to generalize to other estimation problems requires additional mathematical calculations. 

An alternative approach to constructing confidence intervals for $\htPPhom$
is to use the bootstrap. In particular, one can fix the tuning matrix $\hat{\Omega}$, resample the data with replacement $B$ times, recalculate $\htPPhom$ on each resampled version of the dataset, and return the interval given by the $\frac{\alpha}{2}B$-th and $(1-\frac{\alpha}{2})B$-th largest obtained estimates. This easy-to-generalize approach is presented explicitly in Algorithm~\ref{alg:FullPercentileBootstrapWarmup} when the complete and incomplete datasets are IID. Notably, implementation of this algorithm merely requires software to evaluate $\calA(\cdot)$ on various resamplings of the input data. 
We highlight that these confidence intervals are asymptotically valid no matter how accurate the machine-learning predictions $\txna$ are, making them an appealing choice for use with complex machine-learning models. Moreover, when estimating regression coefficients, the confidence intervals are valid even when the regression is misspecified. 
This manuscript discusses the validity of this algorithm and its extensions to cases with weighted, cluster, and stratified sampling.

\begin{varalgorithm}{1}
\caption{Predict-Then-Debias Bootstrap (uniform sampling version)}\label{alg:FullPercentileBootstrapWarmup}
\begin{algorithmic}[1]
\For{$b= 1,\dots,B$}
\State Sample $i_1,i_2,\dots,i_{\ntot} \stackrel{\text{iid}}{\sim} \text{Unif}(\{1,\dots,\ntot\})$
\State Let $\indcStar = \{ k \in \{1,\dots,\ntot\}: i_k \in \indc\}$
\State Let $\indmStar = \{ k \in \{1,\dots,\ntot\}: i_k \in \indm\}$
\State $\htcarg{,\bparen} \gets \calA \big((\xobs_{i_k},\xna_{i_k})_{k \in \indcStar} \big)$
\State $\hgcarg{,\bparen} \gets \calA \big((\xobs_{i_k},\txna_{i_k})_{k \in \indcStar} \big)$
\State $\hgmarg{,\bparen} \gets \calA \big((\xobs_{i_k},\txna_{i_k})_{k \in \indmStar} \big)$ \label{line:CalcHgmInBoot}
\EndFor
\State Select tuning matrix $\hat{\Omega}$ (e.g., using \ref{alg:FullBootTuningSubroutine})
\State $\htPPhomB \gets \hat{\Omega} \hgmarg{,\bparen} + (\htcarg{,\bparen}-\hat{\Omega} \hgcarg{,\bparen} )$ for $b=1,\dots, B$
\State \Return  $\mathcal{C}_j^{1-\alpha} \gets \big( \text{Quantile}_{\alpha/2}( \{\htPPhomB_j\}_{b=1}^B), \text{Quantile}_{1-\alpha/2}( \{\htPPhomB_j\}_{b=1}^B ) \big) \quad \forall_{j \in \{1,\dots,d\}}$
\end{algorithmic}
\end{varalgorithm}


\subsection{Our contribution}

We develop bootstrap-based approaches to construct confidence intervals for the PTD estimator, generalizing to new settings and providing precise technical conditions under which the intervals have theoretical guarantees of validity.
Building off Algorithm \ref{alg:FullPercentileBootstrapWarmup}, the contributions of our paper are as follows:
\begin{itemize}
    \item In Algorithm \ref{alg:FullPercentileBootstrap}, we use inverse probability weighting to generalize Algorithm \ref{alg:FullPercentileBootstrapWarmup} to handle weighted two-phase sampling designs, where the complete sample is a weighted random subsample of all data points. We provide sufficient conditions under which this algorithm returns asymptotically valid confidence intervals.
    \item In Algorithm \ref{alg:QuickConvolutionBootstrap}, we develop a convolution-based speedup to Algorithm \ref{alg:FullPercentileBootstrap} that can be used when the sample size $\ntot$ is large, rendering calculation of $\hgm$ on $B$ resamplings of the data slow. We prove the validity of this faster bootstrap algorithm.
    \item In Section \ref{sec:ClusterAndStratifiedBootstrapMoreInDepth} we present a cluster bootstrap \citep[e.g.,  ][]{Davison_Hinkley_ChapterWithClusterBootstrap} modification of Algorithms \ref{alg:FullPercentileBootstrap} and \ref{alg:QuickConvolutionBootstrap}. In many applications of interest it is of practical importance to select entire geographical clusters within which all units $X$ will be measured.
    Similarly, we present modifications of Algorithm \ref{alg:FullPercentileBootstrap} that can be used when the data is obtained via stratified sampling; such settings are common when the data comes from statistical surveys.

\end{itemize}
In addition, we discuss the efficiency of the PTD estimator and compare it to the PPI++ estimator from \cite{PPI++}. We show a situation where the latter is more efficient than the former, a first in the literature.

Finally, in Section \ref{sec:RealDataExperiments} we show with real-data experiments that our bootstrap-based approaches achieve the desired coverage, while providing confidence intervals that are narrower than simply using the complete sample. Notably, these simulations include cases where the covariates---not the response variables---are imputed via machine learning, a setting that has received little attention in the literature on prediction-powered inference.

\subsection{Related work}

Because the literature on methods for using proxies in a statistical analysis is vast, we restrict our attention to discussing methods that use a small complete subsample, sometimes called \emph{internal validation} data, where all variables of interest, including the proxies, are jointly measured.
This has historically been investigated in the literature on measurement error~\citep{CarrolStefanski06}, missing data~\citep{LittleAndRubin}, and survey sampling~\citep{sarndal2003model}.

\paragraph{Measurement error and missing data:} Measurement error techniques such as regression calibration and SIMEX~\citep[e.g., ][]{CarrolStefanski06} are generally designed for cases where internal validation data is unavailable, so they rely on assumptions about the prediction error. 
Such methods are sometimes used even when there is an internal validation sample \citep{EfficientRegressionCalibration,FongAndTyler2021}, but they still rely critically upon the nondifferential measurement error assumption, unlike the present work.
Multiple imputation methods~\citep[e.g., ][]{LittleAndRubin} are often used when internal validation data is available~\citep{CarrolStefanski06, GuoAndLittle,ProctorPaper}. Multiple imputation does not require the nondifferential measurement error assumption, but it instead requires Bayesian modelling assumptions about the distribution of the missing data (in our case the ground truth values) given the observed data (in our case the predictions and other widely available variables). As a result, confidence intervals for parameters of interest given by multiple imputation-based approaches do not have frequentist guarantees. This is a critical limitation in practice---data-based experiments in \cite{ProctorPaper} show that multiple imputation confidence intervals with nominal level 95\% have coverage of less than 80\% in some cases.

\paragraph{Semiparametric and semisupervised inference:} 
Approaches from the semiparametrics literature \citep{RobinsRotnitzkyZhao94,TsiatisMissingDataSemiparametricChapter}, give a broad class of estimators for the setting of internal validation with asymptotically valid confidence intervals. As discussed in \cite{ChenAndChen2000}, the PTD estimator is asymptotically equivalent to a special case of the semiparametric estimator from \cite{RobinsRotnitzkyZhao94}. Therefore, the semiparametric estimator, with an optimal estimate of the nuisance function, can be more efficient than the PTD estimator.
A recent body of work on semi-supervised inference considers methods for leveraging a large unlabeled dataset and a small labeled dataset and focuses on efficiency in high-dimensional and semiparametric regimes for specific types of estimators. In particular, recent works studied efficient estimation of means \citep{ZhangBrownCaiSemiSupervisedInference,ZhangSemisupervised1,ZhangSemisupervised2}, linear regression parameters \citep{ChakraborttyCaiSemiSupervisedInf}, quantiles \citep{ChakraborttyDaiCarrollSemisupervised2}, and quantile and average treatment effect estimates \citep{ChakraborttyDaiSemisupervised3}.
While their efficiency is appealing, the complexity of semiparametric methods (which require estimating a nuisance function) makes them less accessible to users who do not have extensive statistical training. 

\paragraph{Loss-debiasing approaches:} More recently, prediction-powered inference (PPI) \citep{OriginalPPI}, and other works that build upon or extend it \citep{PPI++,MiaoPSPASomeErrorInX,PPIVariantAustralianJournal,ActivePPI,CrossPPI}, propose more streamlined methods for leveraging internal validation data for convex M-estimation tasks that use few tuning parameters. The estimators from these works have some conceptual similarities to the PTD estimator but have important technical differences: instead of debiasing the naive estimator directly they debias the empirical loss based on predictions. 
Such methods lead to valid confidence intervals, however they do not easily generalize to many use cases of interest. For example, the debiased empirical loss in PPI++ \citep[e.g., ][]{PPI++} can be nonconvex when using predicted covariates in a regression, violating the conditions under which valid coverage holds. See \citep{OptimalPPI} for a further discussion of the connections between the PTD and PPI estimators, and their semiparametrically efficient variants.

Several works in this space have considered nonuniform sampling \citep[e.g., ][]{StratifiedPPI, ActivePPI}, however they do not rely on the PTD estimator or bootstrap confidence intervals, and thus require case-by-case calculations to form valid confidence intervals. Lastly, we note that there are methods for the same setting as PPI that use modeling assumptions~\citep{wang2020methods, mccaw2023leveraging}---see~\citet{motwani2023revisiting} for a comparison with the PPI approach. 



\paragraph{Predict-Then-Debias approaches:} A simple alternative to approaches that use the validation sample to debias the empirical loss is to use the PTD approach, which readily generalizes to a broad class of estimators. PTD-type estimators have been studied for GLMs \citep{ChenAndChen2000,ChenAndChenLikeMedicalJournal,kremers2021generalsimplerobustmethod},  Cox regression models \citep{kremers2021generalsimplerobustmethod}, treatment effect estimators \citep{YangAndDing2020}, and more \citep{MiaoLuNeurIPS,gronsbell2024PromotesCC,PPBootNote}. However, some of these works \citep{ChenAndChen2000,ChenAndChenLikeMedicalJournal} use asymptotic variance calculations to construct confidence intervals, but generalizing to estimators beyond GLMs requires bespoke asymptotic variance calculations. 

Resampling approaches offer a different path forward that may bypass such calculations. \cite{kremers2021generalsimplerobustmethod} and \cite{MiaoLuNeurIPS} propose jackknife-based and bootstrap-based approaches, respectively, for estimating the covariance of the debiased estimator. However, neither work proves the consistency of the covariance estimators. \cite{YangAndDing2020} do provide theoretical guarantees for a bootstrapped-based approach, but 
that work proposes bootstrapping the asymptotic linear expansion terms rather than the estimator itself. This requires the calculation of the linear expansion for the estimator at hand. By contrast, our approach does not require such calculation, and we prove that the confidence intervals are valid under weak conditions. 

Lastly, a critical practical challenge is that internal validation samples are often weighted, stratified, or clustered random samples. To our knowledge, only \cite{YangAndDing2020} (and a brief remark in \cite{kremers2021generalsimplerobustmethod}) address the case where the validation sample is weighted and existing work does not address the cases where it is a clustered or stratified sample.

\section{Properties of the Predict-Then-Debias estimator}\label{sec:TheoreticalSectionPointEst}

We formalize the key properties of the tuned PTD estimator $\htPPhom$. Moreover, we present a couple of strategic ways to choose the tuning matrix $\hat{\Omega}$ and discuss the efficiency of the optimally tuned $\htPPhom$ relative to the classical estimator $\htc$, the PPI++ estimator \citep{PPI++}, and a commonly considered variant of the PTD estimator. Readers not interested in the theoretical development may prefer to skim the next two sections other than the algorithm definitions and then turn to the experiments in Section \ref{sec:RealDataExperiments}. 

\subsection{Formal setting and assumptions}
We start by introducing the formal setting, notation, and some assumptions.
Let $X=(X^{(1)},\dots,X^{(p)}) \in \mathbb{R}^p$ be a random vector drawn from a distribution $\mathbb{P}_X$ denoting the actual data and $\tilde{X}=(\tilde{X}^{(1)},\dots,\tilde{X}^{(p)})  \in \mathbb{R}^p$ be a random vector drawn from a distribution $\mathbb{P}_{\tilde{X}}$
 denoting the machine-learning-imputed proxy for $X$. The reader should think of the case where only a subset of the components of $X$ are imputed with predictions---$X=(\xobs,\xna)$ and $\tilde{X}=(\xobs,\txna)$. The investigator has $\ntot$ proxy samples $\big(  \tilde{X}_i \big)_{i=1}^\ntot$, and for a subset $\indc \subset\{1,\dots,N\}$ of these samples the corresponding vector of actual data, $X_i$, is also available. 

The goal is to use the small number of samples of $(X,\tilde{X})$ and the larger number of samples of $\tilde{X}$ to estimate some parameter $\theta=\phi(\mathbb{P}_X) \in \mathbb{R}^d$, where $\phi$ is a function that maps probability distributions to a summary statistic of interest.  
Note that $\theta$ can be any quantity describing the joint distribution of the entries of $X$. For example, $\theta$ can be the population mean or quantile of a component of $X$, the population correlation between two components of $X$, the population regression coefficient in a generalized linear model (GLM) relating one component of $X$ to other components of $X$. We denote by $\gamma = \phi(\mathbb{P}_{\tilde{X}}) \in \mathbb{R}^d$ the corresponding parameter of the joint distribution of the proxy $\tilde{X}$.


We now introduce assumptions about the sampling of $\indc$ that are analogous to the missing-at-random and positivity assumptions commonly seen in the missing data and sample survey literatures, respectively. Let $I_i \in \{0,1\}$ be the indicator taking on the value $1$ if $X_i$ is observed and $0$ otherwise. We assume the following:
\begin{assumption}[Sampling and missingness assumption]\label{assump:SamplingLabelling} \
    \begin{enumerate}
    [(i)] \vspace{-0.6em}
    \setlength{\itemsep}{1pt}
    \setlength{\parskip}{1pt}
    \item \textit{IID assumption:} $(I_i, X_i,\tilde{X}_i)_{i=1}^{\ntot} \stackrel{\text{iid}}{\sim} \mathbb{P}$;
    \item \textit{Missing at random assumption:} $I \indep X \giv \tilde{X}$; 
    \item \textit{Known sampling probability:} $\pi(\tilde{X}) \equiv \mathbb{P}(I=1 \giv \tilde{X})$ is known; 
    \item \textit{Overlap assumption:} For some $a,b \in (0,1)$,  $a \leq \pi(\tilde{X}) \leq b$ almost surely.
\end{enumerate}
\end{assumption}
For two-phase labeling designs, where an investigator first starts with $(\tilde{X}_i)_{i=1}^\ntot$ and subsequently chooses the subset of samples where $X_i$ is measured (i.e where $I_i=1$), the investigator can easily ensure by design that parts (ii)--(iv) of Assumption \ref{assump:SamplingLabelling} hold. Note that when $X = (\xobs,\txna)$ and $\tilde{X}=(\xobs,\txna)$, this means that $\pi(\tilde X)$ is allowed to depend on the observable features $\xobs$.



We consider settings where $\theta$ and $\gamma$ can be estimated using standard statistical software and an appropriate choice of weights. Specifically, let $\tilde{\cx} \in \mathbb{R}^{\ntot \times p}$ be the data matrix of the proxies whose $i$'th row is $\tilde{X}_i$ and let $\cx \in \mathbb{R}^{\ntot \times p}$ be the data matrix of the actual data whose $i$'th row is $X_i$ (which is observed only when $I_i=1$). Further, define the weights $$W_i \equiv I_i/ \pi(\tilde{X}_i) \quad \text{and} \quad \bar{W}_i \equiv (1-I_i)/(1-\pi(\tilde{X}_i)),$$ typically used for inverse probability weighted estimators, and let $\cwcalib \equiv (W_1,\dots,W_{\ntot}) \in \mathbb{R}^\ntot \ \text{and} \  \cwmain \equiv (\bar{W}_1,\dots,\bar{W}_{\ntot}) \in \mathbb{R}^\ntot.$ Let $\calA : \mathbb{R}^{\ntot \times p} \times \mathbb{R}^{\ntot} \to \mathbb{R}^d$ be a function that takes a data matrix and a weights vector and estimates $\theta$. For many choices of targets $\theta$, there are statistical software packages that implement such $\calA(\cdot ; \cdot)$, allowing a data matrix and a sample weight vector as inputs. Finally, we use the following shorthand notation 
$$\htc \equiv \calA(\cx; \cwcalib),\quad \hgc \equiv \calA(\tilde{\cx}; \cwcalib), \quad \text{and } \hgm \equiv \calA(\tilde{\cx}; \cwmain),$$ to describe the weighted estimator for $\theta$ based on the true data from the complete sample, the weighted estimator for $\gamma$ based on the proxy data from the complete sample, and the weighted estimator for $\gamma$ based on the proxy data from the incomplete sample, respectively. We remark that, even though many rows of $\cx$ are unobserved, $\calA(\cx; \cwcalib)$ can still be evaluated because the missing rows of $\cx$ are given zero weight according to the vector $\cwcalib$.

We study the tuned PTD estimator:
\begin{equation}\label{eq:hatThetaPP}
    \htPPhom \equiv \hat{\Omega} \hgm + ( \htc - \hat{\Omega} \hgc ).
\end{equation}
where $\hat{\Omega} \in \mathbb{R}^{d \times d}$ is a tuning matrix.
Recall its key property: it converges to $\theta$ regardless of the quality of the proxy samples. Formally, provided $\htc \xrightarrow{p} \theta$, $\hgc \xrightarrow{p} \gamma$, $\hgm \xrightarrow{p} \gamma$, and $\hat{\Omega} \xrightarrow{p} \Omega $ for some fixed $\Omega \in \mathbb{R}^{d \times d}$, then
$$\htPPhom \xrightarrow{p} \Omega \gamma + (\theta-\Omega \gamma)=\theta \quad \text{as } n, \ntot \to \infty.$$
In words, this means the PTD estimator is targeting the parameter $\theta$, which is the value that the algorithm would output with an infinitely large complete sample.


\subsection{Asymptotic normality of the PTD estimator}

We now show that the tuned PTD estimator can be well approximated by a normal distribution under certain regularity conditions.  
We begin with an assumption that holds for a large class of estimator functions $\calA(\cdot;\cdot)$.

\begin{assumption}[Asymptotic weighted linearity of $\htc$, $\hgc$, and $\hgm$]\label{assump:AsymptoticLinearity} There exist functions $\Psi : \mathbb{R}^p \to \mathbb{R}^d$ and $\tilde{\Psi}:  \mathbb{R}^p \to \mathbb{R}^d$ such that each component of the random vector $\big(\Psi(X) ,\tilde{\Psi}(\tilde{X}) \big)$ has mean $0$ and finite variance, and such that, as $\ntot \to \infty$,
    \begin{enumerate}[(i)] \vspace{-0.6em}
    \setlength{\itemsep}{1.2pt}
    \setlength{\parskip}{1.2pt}
    \item $\sqrt{\ntot} \big( \htc -\theta - \frac{1}{\ntot} \sum_{i=1}^\ntot W_i \Psi(X_i) \big) \xrightarrow{p} 0$, 
    \item  $\sqrt{\ntot} \big( \hgc -\gamma - \frac{1}{\ntot} \sum_{i=1}^\ntot W_i \tilde{\Psi}(\tilde{X}_i) \big) \xrightarrow{p} 0$,
    \item $\sqrt{\ntot} \big( \hgm -\gamma - \frac{1}{\ntot} \sum_{i=1}^\ntot \bar{W}_i \tilde{\Psi}(\tilde{X}_i) \big) \xrightarrow{p} 0$.
\end{enumerate}
\end{assumption}

While Assumption \ref{assump:AsymptoticLinearity} is stated in an abstract form, it holds for a broad class of common statistical estimators and there is often a mathematical formula that can be used to derive the functions $\Psi$ and $\tilde{\Psi}$. For example, under certain regularity conditions, $\Psi(\cdot)$ and $\tilde{\Psi}(\cdot)$ can be found by calculating the influence function of $\phi$ with respect to the distributions $\mathbb{P}_X$ and $\mathbb{P}_{\tilde{X}}$, respectively \citep{InfluenceFunctionEarlyPaper,VanderVaartTextbook}. 
Assumption \ref{assump:AsymptoticLinearity} also holds for M-estimators under fairly mild regularity conditions. M-estimators include a broad class of estimators of interest such as sample means, sample quantiles, linear and logistic regression coefficients, and regression coefficients in quantile regression or robust regression, among others. See Appendix \ref{sec:RegularityMestimation}, particularly Proposition \ref{prop:MestimatorsAsymptoticallyLinear} and Assumptions \ref{assump:SmoothEnoughForAsymptoticLineariaty} and \ref{assump:SufficientConditionsForConsistency}, for further details justifying Assumption \ref{assump:AsymptoticLinearity} for M-estimation.

Under Assumptions \ref{assump:SamplingLabelling} and \ref{assump:AsymptoticLinearity}, the following proposition shows that $\htPPhom$ is asymptotically normally distributed. 
Let $(W,\bar{W},X,\tilde{X})$ denote a random vector with the same distribution as $(W_i,\bar{W}_i,X_i,\tilde{X}_i)$ and let $\Psi(\cdot)$ and $\tilde{\Psi}(\cdot)$ denote the functions guaranteed by Assumptions \ref{assump:AsymptoticLinearity}. We let $\SigTC \equiv \var \big( W \Psi(X) \big)$, $\SigGC \equiv \var \big( W \tilde{\Psi}(\tilde{X}) \big)$, $\SigGM \equiv \var \big( \bar{W} \tilde{\Psi}(\tilde{X}) \big),$ and $\SigTGC \equiv \cov \big( W \Psi(X),W \tilde{\Psi}(\tilde{X}) \big)$ denote covariance matrices of interest, with $\SigTC$, $\SigGC$, and $\SigGM$ being the asymptotic covariance matrices of $\htc$, $\hgc$, and $\hgm$, respectively. 


\begin{proposition}\label{prop:PPEstimatorCLT}
    Under Assumptions \ref{assump:SamplingLabelling} and \ref{assump:AsymptoticLinearity}, if $\hat{\Omega} \xrightarrow{p} \Omega$ for some $\Omega \in \mathbb{R}^{d \times d}$, then as $\ntot \to \infty$, $\sqrt{\ntot} ( \htPPhom -\theta ) \xrightarrow{d} \mathcal{N} \big(0, \SigPTD (\Omega) \big),$ where \begin{equation}\label{eq:PPOmAsympVar}
\SigPTD (\Omega) \equiv \SigTC-\SigTGC \Omega^\tran -\Omega [ \SigTGC]^\tran +\Omega ( \SigGC + \SigGM) \Omega^\tran. \end{equation} 
\end{proposition}
We defer all proofs of the results from the main text to the appendix.



\subsection{Optimal tuning matrix}\label{sec:TuningMatrixOptions}

 Ideally, the tuning matrix $\hat{\Omega}$ is chosen to minimize the (asymptotic) variance of each component of $\htPPhom$. 
In Equation \eqref{eq:PPOmAsympVar}, note that for each $j \in \{1,\dots,d\}$, $[\SigPTD (\Omega)]_{jj}$ only depends on the $j$th row of $\Omega$, which we denote by $\Omega_{j \cdot} \in \mathbb{R}^d$. Further, $[\SigPTD (\Omega)]_{jj}$ is a quadratic function in $\Omega_{j \cdot}$ that is minimized when $\Omega_{j \cdot}=(\SigGC + \SigGM)^{-1} [ \SigTGC]^\tran e_j$, where $e_j$ is the $j$th canonical vector. Hence, setting $\Omega=\Omega_{\text{opt}}$, where \begin{equation}\label{eq:OptOmegaAsymp}
    \Omega_{\text{opt}} \equiv \SigTGC (\SigGC + \SigGM)^{-1},
\end{equation}  simultaneously minimizes each diagonal entry of $\SigPTD (\Omega)$. If one uses the tuning matrix \begin{equation}\label{eq:HatOmOpt}
    \hat{\Omega}_{\text{opt}} =\hSigTGC (\hSigGC + \hSigGM)^{-1}, 
\end{equation} and $\hSigTGC \xrightarrow{p} \SigTGC$, $\hSigGC \xrightarrow{p} \SigGC$, and $\hSigGM \xrightarrow{p} \SigGM$, then  $\hat{\Omega}_{\text{opt}} \xrightarrow{p} \Omega_{\text{opt}}$ and, by Proposition \ref{prop:PPEstimatorCLT}, we have a tuned PTD estimator $\htPPhomOpt$ with minimal asymptotic variance.

Because $\Omega_{\text{opt}}$ has $d^2$ tuning parameters, estimating $\htPPhomOpt$ can be unstable in small sample sizes. To address this, we also consider the optimal tuning matrix among the class of diagonal tuning matrices, reducing the number of tuning parameters from $d^2$ to $d$. Note that, by minimizing $d$ univariate quadratic equations, the asymptotic variance in Equation \eqref{eq:PPOmAsympVar} is minimized across all diagonal choices of $\Omega$ when letting $\Omega = \Omega_{\text{opt}}^{\text{(diag)}}$, where $\Omega_{\text{opt}}^{\text{(diag)}}$ is the diagonal matrix with
\begin{equation}\label{eq:OmegaOptDiag}
    [\Omega_{\text{opt}}^{\text{(diag)}}]_{jj} = [\SigTGC]_{jj}/ [\SigGC + \SigGM]_{jj}
\end{equation} 
Therefore, selecting $\hat{\Omega}$ such that $\hat{\Omega} \xrightarrow{p} \Omega_{\text{opt}}^{\text{(diag)}}$ minimizes the asymptotic variance of $\htPPhom$ among all possible diagonal choices of $\hat{\Omega}$.

\subsection{Efficiency of the tuned PTD estimator}\label{sec:EfficiencyCalculations}

We next state that $\htPPhomOpt$ is more efficient than the classical estimator $\htc$: for each coordinate $j$, the asymptotic variance of $\htPPhomOpt_j$ is smaller than that of $\htc_j$. This result has been established in similar settings (e.g., \cite{ChenAndChen2000,gronsbell2024PromotesCC}). 
\begin{proposition}
\label{prop:MoreEfficientThanclassical}
   Under Assumptions \ref{assump:SamplingLabelling} and \ref{assump:AsymptoticLinearity}, if  $\hat{\Omega}_{\textnormal{opt}} \xrightarrow{p} \Omega_{\textnormal{opt}}$, then $\sqrt{\ntot}(\htc-\theta) \xrightarrow{d} \mathcal{N}(0,\SigTC)$ and $
       \sqrt{\ntot}(\htPPhomOpt-\theta) \xrightarrow{d} \mathcal{N}(0,\Sigma_{\textnormal{TPTD}})$, where $$\Sigma_{\textnormal{TPTD}} \equiv \SigTC-\SigTGC (\SigGC+\SigGM)^{-1}[ \SigTGC]^\tran \preceq \SigTC.$$
\end{proposition}
To build intuition about the asymptotic efficiency of $\htPPhomOpt$ and how it compares to that of other estimators, we consider the case where $\htPPhomOpt$ is univariate (i.e., $d=1$) and where the complete sample is a uniform random subsample (i.e., for some $\pi_L \in (0,1)$, $\pi(\tilde{X})=\mathbb{P}(I=1 \giv \tilde X)=\pi_L$ always).  
In this case, the following holds
\begin{equation*} 
\frac{\sigma_{\textnormal{TPTD}}^2}{\sigma_{\textnormal{classical}}^2}=1-(1-\pi_L) \cdot \corr^2(\Psi(X),\tilde{\Psi}(\tilde{X})), 
\end{equation*}
This reveals two natural facts. First, $\htPPhomOpt$ has a larger gain over the classical approach when $\pi_L$, the probability of observing the real data, is small. This is intuitive: if most real samples are observed, then there is little to be gained from machine learning imputations and $\htc$ is a good estimator. Second, the greater the correlation between $\Psi(X)$ and $\tilde{\Psi}(\tilde{X})$, the greater the efficiency of $\htPPhomOpt$ relative to the classical approach. If $\tilde{X} \approx X$, we expect $\Psi(X) \approx \tilde{\Psi}(\tilde{X})$ and that $\corr(\Psi(X),\tilde{\Psi}(\tilde{X}))$ would be large.

In in the same setting, we also compare the variance of the optimally tuned PTD estimator to that of the optimally tuned PPI++ estimator \citep{PPI++}, which we call $\sigma_{\text{PPI++}}^2$. We prove in Appendix \ref{sec:PPI++ComparisonSection} that
\begin{equation*} 
\frac{\sigma_{\text{TPTD}}^2}{\sigma_{\text{PPI++}}^2} = \frac{1-(1-\pi_L) \cdot \corr^2(\Psi(X),\tilde{\Psi}(\tilde{X}))}{1-(1-\pi_L) \cdot \corr^2(\Psi(X),\Psi(\tilde{X}))}.  
\end{equation*}
This is of interest because it shows that recent proposals in the prediction-powered inference literature that involve debiasing the loss function (e.g., \cite{PPI++}) are not necessarily less efficient than the PTD approach. Indeed, in Appendix \ref{sec:PPI++ComparisonSection} we present examples where $\sigma_{\text{PPI++}}^2>\sigma_{\text{TPTD}}^2$ and other examples where $\sigma_{\text{PPI++}}^2<\sigma_{\text{TPTD}}^2$. The latter is the first such case exhibited in the literature.

 
Lastly, the reader might wonder about the efficiency of a PTD-type estimator that uses all samples (rather than just incomplete samples) to calculate $\hgm$. It turns out this variant has the same asymptotic variance as the PTD estimator herein, provided the complete sample is a uniform random subsample and the optimal tuning matrices are used; see Appendix \ref{sec:PTDusingGammaHatAll}.


\section{Bootstrap confidence intervals}\label{sec:BootstrapCIOverallSection}

We now develop bootstrap algorithms that construct confidence intervals based on $\htPPhom$. Algorithm~\ref{alg:FullPercentileBootstrap} generalizes Algorithm~\ref{alg:FullPercentileBootstrapWarmup} to non-uniformly weighted settings. We then introduce a computational speedup in Algorithm~\ref{alg:QuickConvolutionBootstrap}. We prove that both algorithms provide asymptotically valid confidence intervals under suitable assumptions. Finally, we discuss generalizations of these bootstrap approaches to clustered and stratified sampling settings. 
 
The bootstrap approaches presented in this section are more flexible than CLT-based approaches for constructing confidence intervals in the sense that they do not require the mathematical calculation of the asymptotic variance terms. Nonetheless, for completeness, we present a CLT-based approach in Appendix \ref{sec:CLTBasedCIsAlgorithmAndTheory}.

\subsection{Main bootstrap algorithm and its validity}
\label{sec:BootstrapBasedCIs}

We first introduce the necessary notation. Let $V_i=(W_i,\bar{W}_i,X_i,\tilde{X}_i)$ for each $i$. Further, let $\hat{\mathbb{P}}_{\ntot}=\frac{1}{\ntot} \sum_{i=1}^\ntot \delta_{V_i}$ be the empirical distribution of $V_i$ from the $\ntot$ samples, where $\delta_{v}$ assigns a point mass of $1$ at $v$ and $0$ elsewhere. Next, define a single bootstrap draw of $\htc$, $\hgc$, $\hgm$, and $\htPPhom$ to be the version of that quantity with a starred superscript in the following procedure: \begin{enumerate}
\vspace{-0.6em}
\setlength{\itemsep}{1.2pt}
    \setlength{\parskip}{1.2pt}
    \item Draw $V_1^*,\dots,V_{\ntot}^* \stackrel{\text{iid}}{\sim} \hat{\mathbb{P}}_{\ntot}$ and set $(W_i^*,\bar{W}_i^*,X_i^*,\tilde{X}_i^*)=V_i^*$ for each $i \in \{1,\dots,\ntot\}$.
    \item Set $\cwcalibStar=(W_1^*,\dots,W_{\ntot}^*)$, $\cwmainStar=(\bar{W}_1^*,\dots,\bar{W}_{\ntot}^*)$ and $\cx^*,\tilde{\cx}^* \in \mathbb{R}^{\ntot \times p}$ such that the $i$'th rows of $\cx^*$ and $\tilde{\cx}^*$ are $X_i^*$ and $\tilde{X}_i^*$, respectively.
    \item Evaluate $\htcarg{,*} = \calA(\cx^*; \cwcalibStar)$, $\hgcarg{,*} = \calA(\tilde{\cx}^*; \cwcalibStar)$, and  $\hgmarg{,*} = \calA(\tilde{\cx}^*; \cwmainStar)$.
    \item Set $\htPPhomStar=\hat{\Omega} \hgmarg{,*} + (\htcarg{,*}- \hat{\Omega} \hgcarg{,*})$. 
\end{enumerate} 
With this in hand, Algorithm \ref{alg:FullPercentileBootstrap} computes confidence intervals at level $\alpha$ for each component of $\htPPhom$ by first taking $B$ independent draws from the bootstrap distribution and returning the $\alpha/2$ and $1-\alpha/2$ empirical quantiles of each coordinate. Note that $\calA(\cx^*; \cwcalibStar)$ can be evaluated despite unobserved rows of $\cx^*$ because $\cwcalibStar$ assigns zero weight to such rows.


\begin{varalgorithm}{2}
\caption{Predict-Then-Debias Bootstrap}\label{alg:FullPercentileBootstrap}
\begin{algorithmic}[1]
\For{$b= 1,\dots,B$}
\State Let $\ci$ be a vector of length $\ntot$ generated by sampling $\{1,2,\dots,\ntot\}$ with replacement
\State $(\cwcalibStar,\cwmainStar, \cx^*,\tilde{\cx}^*) \gets (\cwcalib_{\ci},\cwmain_{\ci},\cx_{\ci \cdot},\tilde{\cx}_{\ci \cdot})$
\State $\htcarg{,\bparen} \gets \calA(\cx^*; \cwcalibStar)$
\State $\hgcarg{,\bparen} \gets \calA(\tilde{\cx}^*; \cwcalibStar)$
\State $\hgmarg{,\bparen} \gets \calA(\tilde{\cx}^*; \cwmainStar)$
\EndFor
\State Select tuning matrix $\hat{\Omega}$ (e.g., using \ref{alg:FullBootTuningSubroutine})
\State $\htPPhomB \gets \hat{\Omega} \hgmarg{,\bparen} + (\htcarg{,\bparen}-\hat{\Omega} \hgcarg{,\bparen} )$ for $b=1,\dots, B$
\State \Return  $\mathcal{C}_j^{1-\alpha} \gets \big( \text{Quantile}_{\alpha/2}( \{\htPPhomB_j\}_{b=1}^B), \text{Quantile}_{1-\alpha/2}( \{\htPPhomB_j\}_{b=1}^B ) \big) \quad \forall_{j \in \{1,\dots,d\}}$

\end{algorithmic}
\end{varalgorithm}

 To show that Algorithm \ref{alg:FullPercentileBootstrap} leads to asymptotically valid confidence intervals, we must first introduce some further notation and an assumption. Let \begin{equation}\label{eq:ZetaDefNoCov}
     \zeta \equiv (\theta,\gamma,\gamma), \quad \hat{\zeta} \equiv (\htc,\hgc,\hgm), \quad \text{and} \quad \hat{\zeta}^*\equiv(\htcarg{,*}, \hgcarg{,*}, \hgmarg{,*}).
 \end{equation}
Under a fixed realization of the data (i.e., for a fixed $\hat{\mathbb{P}}_{\ntot}$), we call the distribution of $\hat{\zeta}^*$ generated by the above 4-step empirical bootstrap procedure the bootstrap distribution of $\hat{\zeta}$, and we use $\mathbb{P}_*( A| \hat{\mathbb{P}}_{\ntot})$ to denote the probability that an event $A$ occurs under the bootstrap distribution of $\hat{\zeta}$. Below we introduce a bootstrap consistency assumption for $\hat{\zeta}$ that heuristically says the random bootstrap distribution of $\hat{\zeta}^*-\hat{\zeta}$ uniformly converges to the distribution of $\hat{\zeta}-\zeta$. The below consistency assumption is not trivial to check; however, because much literature has been devoted to proving the bootstrap consistency for a large variety of settings and estimators of interest \citep{ShaoAndTuTextbook,KosorokEmpiricalProcessTextbook,VDVAndWellnerTextbook}, we state bootstrap consistency of $\hat{\zeta}$ as an assumption and then give precise technical conditions where the assumption will be met for example use cases of interest.
 

\begin{assumption}[Bootstrap consistency and limiting distribution of $\hat{\zeta}$]\label{assump:ZetaBootstrapConsistency} For each fixed $v \in \mathbb{R}^{3d}$, \begin{enumerate}[(i)]
    \item $\sup\limits_{x \in \mathbb{R}} \vert \mathbb{P}_*( \sqrt{\ntot} v^\tran (\hat{\zeta}^*-\hat{\zeta})  \leq x \giv \hat{\mathbb{P}}_{\ntot})- \mathbb{P}(\sqrt{\ntot} v^\tran (\hat{\zeta}-\zeta) \leq x) \vert \xrightarrow{p} 0,$ and
    \item $\sqrt{\ntot} v^\tran (\hat{\zeta}-\zeta)$ converges in distribution to some random variable with a symmetric distribution and a continuous, strictly increasing CDF.
\end{enumerate}
    
\end{assumption}

This assumption holds in a variety of use cases of interest. For many common statistical estimands of interest, such as Z-estimators (including linear or logistic regression coefficients), quantiles, and L-statistics, among others, Assumption \ref{assump:ZetaBootstrapConsistency} will hold, provided that Assumption \ref{assump:SamplingLabelling} and mild regularity conditions (that are specific to the estimand) are met.


\begin{remark}[Bootstrap Consistency for Z-estimators] Z-estimators are estimators that solve for the zero of estimating equations, and include M-estimators with differentiable loss functions. When $\htc$, $\hgc$, and $\hgm$ are Z-estimators, $\hat{\zeta}=(\htc,\hgc,\hgm)$ is also a Z-estimator. Many works in the theoretical bootstrap literature (e.g., Chapters 10.3 and 13 of \cite{KosorokEmpiricalProcessTextbook}) show that under certain standard and fairly mild regularity conditions, Z-estimators satisfy the bootstrap consistency criteria in Assumption \ref{assump:ZetaBootstrapConsistency} or a stronger version of it. In Appendix \ref{sec:ZestimatorsBootstrapConsistencyConditions}, by direct application of such results, we present some sufficient (although not necessary) conditions under which Assumption \ref{assump:ZetaBootstrapConsistency} holds for Z-estimators.
    
\end{remark}

\begin{remark}[Hadamard differentiable estimators] As another route toward verifying Assumption~\ref{assump:ZetaBootstrapConsistency} holds, when $\htc$, $\hgc$, and $\hgm$, are each Hadamard differentiable functions of the empirical distributions of $X$ or $\tilde{X}$, $\hat{\zeta}=(\htc,\hgc,\hgm)$ will also be a Hadamard differentiable function of a particular empirical distribution. In Appendix \ref{sec:UsingHadamardDifferentiablityToShowConsistency} and Theorem \ref{theorem:HadamardDiffImpliesBootstrapConsitency}, we characterize precisely when and how Hadamard differentiability of the component estimators implies Assumption \ref{assump:ZetaBootstrapConsistency}.
A number of estimators including quantiles and trimmed means are Hadamard differentiable functions of the empirical distribution function. 
    
\end{remark}

The following theorem shows that Algorithm \ref{alg:FullPercentileBootstrap} provides asymptotically valid confidence intervals under certain conditions. 

\begin{theorem} 
\label{theorem:FullPercentileBootstrapCIsValid}
    Under Assumption \ref{assump:ZetaBootstrapConsistency}, if $\hat{\Omega}=\Omega +o_p(1)$, Algorithm \ref{alg:FullPercentileBootstrap} returns asymptotically valid confidence intervals $\mathcal{C}_1^{1-\alpha},\dots,\mathcal{C}_d^{1-\alpha}$, in the sense that $$\lim\limits_{\ntot,B \to \infty} \mathbb{P}(\theta_j \in \mathcal{C}_j^{1-\alpha})=1-\alpha \quad \text{for all } j \in \{1,\dots,d\}.$$ 
\end{theorem}

\subsection{A faster bootstrap procedure}
Algorithm \ref{alg:FullPercentileBootstrap} can be slow if the incomplete sample is large because it requires computing $\hgmarg{,*}$ for $B$ different bootstrap draws, which requires evaluating $\cal A$ on a large data set with each draw. (For the percentile bootstrap it is often recommended to choose $B=2{,}000$ or larger (e.g., \cite{LittleAndRubin}).) Next, we propose a convolution-based speed-up (Algorithm \ref{alg:QuickConvolutionBootstrap}) that replaces the computation of $\hgmarg{,*}$ with a Gaussian approximation. We show that this is valid when $\hgm=\calA(\tilde{\cx}; \cwmain)$ is asymptotically Gaussian and when a consistent estimator of its asymptotic variance is readily available. This convolution-based speed-up exploits the fact that $\hgm$ is asymptotically uncorrelated with $\htc$ and $\hgc$.

Implementation of the speed-up for Algorithm \ref{alg:FullPercentileBootstrap} merely requires a consistent estimator of $\SigGM$. Under Assumptions \ref{assump:SamplingLabelling} and \ref{assump:AsymptoticLinearity}, standard statistical software can typically be used to return a matrix $\widehat{\var}(\hgm)$ estimating $\var(\hgm)$ such that $\ntot \widehat{\var}(\hgm) \xrightarrow{p} \SigGM$ as $\ntot \to \infty$ (e.g., see Remark \ref{remark:SoftwareOftenGivesCovMatEst}). Letting $\widehat{\var}(\hgm)$ be such a readily computable estimate of $\var(\hgm)$, our proposed algorithm is as follows: 
\begin{varalgorithm}{3}
\caption{Speedup for Predict-Then-Debias Bootstrap}\label{alg:QuickConvolutionBootstrap}
\begin{algorithmic}[1]
\State Choose tuning matrix $\hat{\Omega}$ that is consistent for $\Omega$ 
\State $\hgm \gets \calA(\tilde{\cx}; \cwmain)$ \Comment{Naive estimate of $\theta$ using the incomplete sample's proxy data}
\State $\hSMatGM \gets \widehat{\var}(\hgm)$ \Comment{Often implemented in statistical software}
\State $\hat{L}_{\gamma} \gets \text{Cholesky}(\hSMatGM)$ \Comment{Cholesky decomposition such that $\hat{L}_{\gamma}\hat{L}_{\gamma}^\tran= \hSMatGM$}
\For{$b= 1,\dots,B$}
\State Let $\ci$ be a vector of length $\ntot$ generated by sampling $\{1,2,\dots,\ntot\}$ with replacement
\State $(\cwcalibStar, \cx^*,\tilde{\cx}^*) \gets (\cwcalib_{\ci},\cx_{\ci \cdot},\tilde{\cx}_{\ci \cdot})$
\State $\htcarg{,\bparen} \gets \calA(\cx^*; \cwcalibStar)$
\State $\hgcarg{,\bparen} \gets \calA(\tilde{\cx}^*; \cwcalibStar)$
\State $Z^{(b)} \sim \mathcal{N}(0,I_{d\times d})$ \Comment{Independent draw from standard multivariate normal}
\EndFor
\State Select tuning matrix $\hat{\Omega}$ (e.g., using \ref{alg:ConvBootTuningSubroutine})
\State $\htPPhomB \gets \hat{\Omega}( \hgm + \hat{L}_{\gamma}Z^{(b)}) + (\htcarg{,\bparen}-\hat{\Omega} \hgcarg{,\bparen} )$ for $b=1,\dots, B$ \label{line:CalcThetaPTDBootConvApprox}
\State \Return  $\mathcal{C}_j^{1-\alpha} \gets \big( \text{Quantile}_{\alpha/2}( \{\htPPhomB_j\}_{b=1}^B), \text{Quantile}_{1-\alpha/2}( \{\htPPhomB_j\}_{b=1}^B ) \big) \quad \forall_{j \in \{1,\dots,d\}}$
\end{algorithmic}
\end{varalgorithm}

The following theorem shows that under certain conditions, Algorithm \ref{alg:QuickConvolutionBootstrap} gives asymptotically valid confidence intervals. 

\begin{theorem}
\label{theorem:GaussianConvBootCIsValid}
    Under Assumptions \ref{assump:SamplingLabelling}--\ref{assump:ZetaBootstrapConsistency}, if $\hat{\Omega} \xrightarrow{p} \Omega$ and $\ntot \widehat{\var}(\hgm) \xrightarrow{p} \SigGM \succ 0$ as $\ntot \to \infty$, Algorithm \ref{alg:QuickConvolutionBootstrap} gives asymptotically valid confidence intervals $\mathcal{C}_1^{1-\alpha},\dots, \mathcal{C}_d^{1-\alpha}$ in the sense that $$\lim\limits_{\ntot,B \to \infty} \mathbb{P}(\theta_j \in \mathcal{C}_j^{1-\alpha})=1-\alpha \quad \text{for all } j \in \{1,\dots,d\}.$$ 
\end{theorem}

 We note that while standard software can often be used to return a consistent estimator of $\SigGM$ (which is all that is required for Algorithm \ref{alg:QuickConvolutionBootstrap}), it cannot typically be used to find a consistent estimator for $\SigTGC$ (which is required for constructing CLT-based confidence intervals via Algorithm \ref{alg:CLTBasedCIs} in Appendix~\ref{sec:CLTBasedCIsAlgorithmAndTheory}). Therefore, Algorithm \ref{alg:QuickConvolutionBootstrap} is easier to generalize to new estimators than the purely CLT-based Algorithm \ref{alg:CLTBasedCIs} while having faster runtime than Algorithm \ref{alg:FullPercentileBootstrap}.

\subsection{Subroutines for estimating optimal tuning matrix}

In this subsection, we present subroutines for Algorithms \ref{alg:FullPercentileBootstrap} and \ref{alg:QuickConvolutionBootstrap} to compute the tuning matrix $\hat{\Omega}$. The subroutines are designed to have the same computational complexity as the corresponding algorithm. The subroutines can easily be modified such that $\hat{\Omega}$ estimates the optimal $d \times d$ tuning matrix rather than the optimal diagonal tuning matrix by modifying the last line to return $\hat{\Omega} = \hSigTGC (\hSigGC+\hSigGM)^{-1}$; see Equation \eqref{eq:HatOmOpt}.

\begin{varalgorithm}{Subroutine 1}\floatname{algorithm}{}
\caption{Estimate optimal diagonal tuning matrix in Algorithm  \ref{alg:FullPercentileBootstrap}}\label{alg:FullBootTuningSubroutine}
\begin{algorithmic}[1]

\State  \textbf{Input} $B$ bootstrap draws $\big(\htcarg{,\bparen},\hgcarg{,\bparen},\hgmarg{,\bparen} \big)_{b=1}^B$ of $(\htc,\hgc,\hgm)$
\State $\hSigTGC \gets \ntot \widehat{\text{Cov}}( \{ \htcarg{,\bparen} \}_{b=1}^B, \{ \hgcarg{,\bparen} \}_{b=1}^B )$
\State $\hSigGC \gets \ntot \widehat{\text{Var}}( \{ \hgcarg{,\bparen} \}_{b=1}^B )$
\State $\hSigGM \gets \ntot \widehat{\text{Var}}( \{ \hgmarg{,\bparen} \}_{b=1}^B )$
\State \Return  $\hat{\Omega} \gets \text{Diag} \Big(\frac{[\hSigTGC]_{11}}{[\hSigGC]_{11}+[\hSigGM]_{11}}, \dots, \frac{[\hSigTGC]_{dd}}{[\hSigGC]_{dd}+[\hSigGM]_{dd}} \Big)$
\end{algorithmic}
\end{varalgorithm}

\begin{varalgorithm}{Subroutine 2}\floatname{algorithm}{}
\caption{Estimate optimal diagonal tuning matrix in Algorithm \ref{alg:QuickConvolutionBootstrap}}\label{alg:ConvBootTuningSubroutine}
\begin{algorithmic}[1]

\State  \textbf{Input} $\hSMatGM$ and $B$ bootstrap draws $\big(\htcarg{,\bparen},\hgcarg{,\bparen}\big)_{b=1}^B$ of $(\htc,\hgc)$ and 
\State $\hSigTGC \gets \ntot \widehat{\text{Cov}}( \{ \htcarg{,\bparen} \}_{b=1}^B, \{ \hgcarg{,\bparen} \}_{b=1}^B )$
\State $\hSigGC \gets \ntot \widehat{\text{Var}}( \{ \hgcarg{,\bparen} \}_{b=1}^B )$
\State $\hSigGM \gets \ntot \hSMatGM$
\State \Return  $\hat{\Omega} \gets \text{Diag} \Big(\frac{[\hSigTGC]_{11}}{[\hSigGC]_{11}+[\hSigGM]_{11}}, \dots, \frac{[\hSigTGC]_{dd}}{[\hSigGC]_{dd}+[\hSigGM]_{dd}} \Big)$
\end{algorithmic}
\end{varalgorithm}

\subsection{Cluster and stratified bootstraps}\label{sec:ClusterAndStratifiedBootstrapPreview}

In many applications of interest, it is more economical to measure $X$ for entire clusters of samples (e.g., all samples in a geographical unit) and forgo measuring $X$ entirely on the remaining clusters. Such cases will violate Assumption \ref{assump:SamplingLabelling} and can render the confidence intervals from Algorithms \ref{alg:FullPercentileBootstrap} and  \ref{alg:QuickConvolutionBootstrap} too narrow. These algorithms can readily be extended using a cluster bootstrap scheme to appropriately construct confidence intervals in such settings. The cluster bootstrap modification involves resampling entire clusters with replacement as opposed to resampling individual samples with replacement---see Section \ref{sec:ClusterBootstrapDescription} and Algorithm \ref{alg:ClusterBootstrap} for details.

Another common setting is stratified sampling, which can often reduce the number of samples needed.
In the stratified sampling that we consider, the population is partitioned into strata and a fixed number of incomplete samples and complete samples are drawn from each strata. 
In Appendix~\ref{sec:StratifiedBootstrapDescription} we present a modification of Algorithm \ref{alg:FullPercentileBootstrap} (Algorithm \ref{alg:StratifiedBootstrap}) that constructs confidence intervals for the PTD estimator that account for the stratified sampling scheme. We caution readers that Algorithm \ref{alg:StratifiedBootstrap} is only designed to work in regimes where there is a small number of large strata and instead point readers to Section 6.2.4 of \cite{ShaoAndTuTextbook} for variants of the bootstrap for stratified samples that are designed to work in other regimes. Theoretical justifications of Algorithms \ref{alg:ClusterBootstrap} and \ref{alg:StratifiedBootstrap} are out of scope for the current work, but we refer the reader to \cite{Davison_Hinkley_ChapterWithClusterBootstrap,ShaoAndTuTextbook} and references therein for details on these methods and \cite{ClusterBootstrapTheoreticalJustification} for a theoretical justification of the cluster bootstrap.

\section{Experiments}\label{sec:RealDataExperiments}

In this section, we present a variety of experiments using four different real datasets to validate our method and compare it to the classical approach (i.e.,  only using the gold-standard data from the complete sample). We also consider a number of variations of the PTD approach that involve different algorithms for constructing confidence intervals (see Section \ref{sec:ExperimentsVaryingCIMethod}) and different tuning matrix choices (see Section \ref{sec:ExperimentsVaryingTuningMatrix}). 

Our experiments focus on regression tasks such as linear regression, logistic regression, and quantile regression, and some of them involve weighted, stratified, or clustered labelling schemes. Let $Y$ and $Z$ be the subvectors of $X$ corresponding to the response variable and covariate vector, respectively, and define $\tilde{Y}$ and $\tilde{Z}$ as the analogous subvectors of $\tilde X$. 
Our regression experiments fall into 3 main categories:

\begin{enumerate}
    \item \textbf{Error-in-response regressions:} In an error-in-response regression, the investigator has access to a large incomplete sample with measurements of $\tilde{Y}$ and $Z$ and can obtain access to a much smaller complete sample with measurements of $Y$, $\tilde{Y}$ and $Z$. 
    \item \textbf{Error-in-covariate regressions:} In an error-in-covariate regression, the investigator has access to a large incomplete sample with measurements of $Y$ and $\tilde{Z}$ and can obtain access to a much smaller complete sample with measurements of $Y$, $Z$, and $\tilde{Z}$.
    \item \textbf{Error-in-both regressions:} In an error-in-both regression, the investigator has access to a large incomplete sample with measurements of $\tilde{Y}$ and $\tilde{Z}$ and can obtain access to a much smaller complete sample with measurements of $Y$, $\tilde{Y}$, $Z$, and $\tilde{Z}$.
\end{enumerate}
Recent works on prediction-powered inference \citep{OriginalPPI,PPI++,MiaoLuNeurIPS,PPBootNote} only test their methods for error-in-response regressions, so we focus on error-in-covariate and error-in-both regressions.

The experiments conducted are summarized in Table \ref{table:ExperimentSummary}. We describe the experimental setup and datasets in more detail in the following subsections.

\begin{table}[hbt!]
\caption{Summary of experiments. The penultimate column gives the expected number of complete samples in each simulation. The final column gives the algorithms that were tested to construct confidence intervals for the PTD estimator.
}
\label{table:ExperimentSummary}
\centering
\scriptsize
\begin{tabular}{c c c c c c r r c} 
\toprule
Exp \#  & Dataset & Model & Error Regime & Sampling Scheme & $\ntot \ \ \ $  & $\e[\nc]$ & Algorithms \\ 
\midrule
1  & AlphaFold &  Logistic Reg & Error-in-response & Weighted &  $7{,}500$ & $1{,}000$ &  \ref{alg:FullPercentileBootstrap},\ref{alg:QuickConvolutionBootstrap},\ref{alg:CLTBasedCIs} \\ [1ex] 
2  & Housing Price & Linear Reg & Error-in-covariate & Uniform &  $5{,}000$ & $500$ & \ref{alg:FullPercentileBootstrap},\ref{alg:QuickConvolutionBootstrap},\ref{alg:CLTBasedCIs} \\ [1ex] 
3   & Housing Price &  Quantile Reg & Error-in-covariate & Uniform &  $5{,}000$ & $1{,}000$ & \ref{alg:FullPercentileBootstrap},\ref{alg:QuickConvolutionBootstrap} \\ [1ex] 
4   & Tree cover &  Linear Reg & Error-in-both & Uniform & $5{,}000$ & $500$ & \ref{alg:FullPercentileBootstrap},\ref{alg:QuickConvolutionBootstrap},\ref{alg:CLTBasedCIs} \\ [1ex] 
5  & Tree cover &  Linear Reg & Error-in-both & Clustered & $\sim$$10{,}000 \ \ \ \ $ & $1{,}000$ & \ref{alg:ClusterBootstrap}, \ref{alg:QuickConvolutionBootstrap}/\ref{alg:ClusterBootstrap} \\ [1ex] 
6 &  Tree cover &  Logistic Reg & Error-in-both & Uniform & $8{,}000$ & $1{,}000$ & \ref{alg:FullPercentileBootstrap},\ref{alg:QuickConvolutionBootstrap},\ref{alg:CLTBasedCIs} \\ [1ex] 
7  & Census &  Linear Reg & Error-in-covariate & Stratified & $6{,}000$ & $1{,}000$ & \ref{alg:StratifiedBootstrap} \\ [1ex]

 \bottomrule
\end{tabular}
\smallskip

\end{table}

\subsection{Experimental setup}

For each experiment we use the following validation procedure. We start with $M$ samples where $X$ and $\tilde{X}$ are jointly measured, where $X=(Y,Z)$ and $\tilde{X}=(\tilde{Y},\tilde{Z})$. We calculate the ``ground truth" regression coefficients $\theta$ by regressing $Y$ on $Z$ using all $M$ samples (dashed lines in Figure \ref{fig:ViolinPlotFigure}) and the ``naive" regression coefficients $\hga$ by regressing $\tilde{Y}$ on $\tilde{Z}$ using all $M$ samples (red lines in Figure \ref{fig:ViolinPlotFigure}). Even if the $M$ samples are drawn from a superpopulation, we can treat $\theta$ estimated from the $M$ samples as our estimand of interest because all simulations are based on these $M$ empirical samples. We then conduct 500 simulations in which we \begin{enumerate}
    \item Randomly select a subsample of the $M$ samples of size $\ntot$. A subset of the $\ntot$ samples are randomly assigned to the complete sample, where all variables are observed. For all samples not assigned to the complete sample, the $\xna$ values are withheld.
    \item Compute the classical estimator of $\theta$ and the corresponding 90\% confidence interval using only data from the complete sample. Confidence intervals are calculated using the sandwich estimator for $\var(\htc)$.
    \item Run various PTD-based approaches using $(X,\tilde{X})$ from the complete sample and $\tilde{X}$ from the remaining samples to estimate $\theta$ and a corresponding 90\% confidence interval. In particular, we consider a number of choices of tuning matrices $\hat{\Omega}$ and Algorithms   \ref{alg:FullPercentileBootstrap}-- \ref{alg:CLTBasedCIs}, to construct 90\% confidence intervals for $\htPPhom$. In some cases, we forgo Algorithm \ref{alg:CLTBasedCIs} or both Algorithms \ref{alg:QuickConvolutionBootstrap} and \ref{alg:CLTBasedCIs}  because of lack of implemented analytic expressions. 
\end{enumerate} 
Unless otherwise specified, step 1 is done by uniform random sampling without replacement; however, in some of our experiments we use a weighted, stratified, or cluster sampling scheme.
Finally, we calculate the coverage for each method by calculating the percentage of the $500$ simulations in which the 90\% confidence intervals contained $\theta$ (estimated from regressing $Y$ on $Z$ using all $M$ samples).

\subsection{Datasets}

We briefly describe the datasets used and experiments conducted. Further details are presented in Appendix \ref{sec:FurtherDetailsOnExperiments}.

\paragraph{AlphaFold Experiments:} We used a dataset of $M=10{,}802$ samples that originated from \cite{bludau2022structural} and was downloaded from Zenodo \citep{PPIZenodo}. Each sample had indicators $Z_{\text{Acet}},Z_{\text{Ubiq}} \in \{0,1\}$ of whether there was acetylation and ubiquatination, and an indicator $Y_{\text{IDR}} \in \{0,1\}$ of whether the protein region was an internally disordered region (IDR) coupled with a prediction of $Y_{\text{IDR}}$ based on AlphaFold \citep{AlphaFoldPaper}. We test Algorithms \ref{alg:FullPercentileBootstrap}--\ref{alg:CLTBasedCIs} when the estimands are the regression coefficients for a logisitic regression of $Y_{\text{IDR}}$ on $(Z_{\text{Acet}},Z_{\text{Ubiq}},Z_{\text{Acet}} \times Z_{\text{Ubiq}})$. $Y_{\text{IDR}}$ was withheld on all samples outside from a randomly selected complete sample. The complete sample was selected according to a weighted sampling scheme such that for each of the 4 possible combinations of $Z_{\text{Acet}}$ and $Z_{\text{Ubiq}}$, there were 250 complete samples in expectation. 

\paragraph{Housing Price Experiments:} 

We used a dataset of gold-standard measurements and remote sensing-based estimates of economic and environmental variables from \cite{MOSAIKSPaper,MOSAIKSSourceCode} that was used in \cite{ProctorPaper} to study multiple imputation methods and in \cite{KerriRSEPaper} to study the Predict-Then-Debias method. Each of the $M=46{,}418$ samples corresponded to a distinct $\sim 1 \text{km} \times 1 \text{km}$ grid cell, and included grid cell-level averages of housing price, income, nightlight intensity, and road length. We considered settings where the estimand was the regression coefficient of housing price on income, nightlight intensity, and road length, and where outside of a small complete sample, gold-standard measurements of nightlights and road length were unavailable (but proxies based on daytime satellite imagery were available for all samples). In Experiment 2 the estimands were regression coefficients from a linear regression, and the estimands in Experiment 3 were regression coefficients for a quantile regression with $q=0.5$.

\paragraph{Tree Cover Experiments:} We used a dataset of $M=67{,}968$ samples of $\sim 1 \text{km} \times 1 \text{km}$ grid cells taken from the previously mentioned data source \citep{MOSAIKSPaper,MOSAIKSSourceCode}. The variables included the percent of tree cover and grid cell-level averages of population and elevation. We considered settings where the estimands were the regression coefficients of tree cover on elevation and population, and where outside of a small complete sample, gold-standard measurements of tree cover and population were unavailable (but satellite-based proxies were available for all samples). In Experiments 4 and 5 the estimands were regression coefficients from a linear regression, while in Experiment 6 the estimands were regression coefficients from a logistic regression where the tree cover was binarized according to a meaningful threshold from a forestry perspective \citep{UDSAForestServiceReport}. In Experiment 5, the data and the complete samples were sampled via cluster sampling (where each cluster corresponded to a 0.5° $\times$ 0.5° grid cell), and the cluster bootstrap method (Algorithm \ref{alg:ClusterBootstrap}) was tested.

\paragraph{Census Experiments:} We used a dataset of income, age, and disability status of $M=200{,}227$ individuals from California taken from the 2019 US Census survey and downloaded via the Folktables interface \citep{FolktablesPaper}. We considered a setting where the estimands are the regression coefficients of income on age and disability status, and where disability status was only measured on a small complete sample. Outside the complete sample we used predictions of disability status based on a machine learning model that we trained using the previous year's census data. In Experiment 7 the incomplete sample and complete sample were taken according to stratified random sampling (with $K=4$ strata based on age buckets), and we tested Algorithm \ref{alg:StratifiedBootstrap}.

\subsection{Point estimates and confidence interval size}

\begin{figure}
    \centering
    \includegraphics[width=0.95\linewidth]{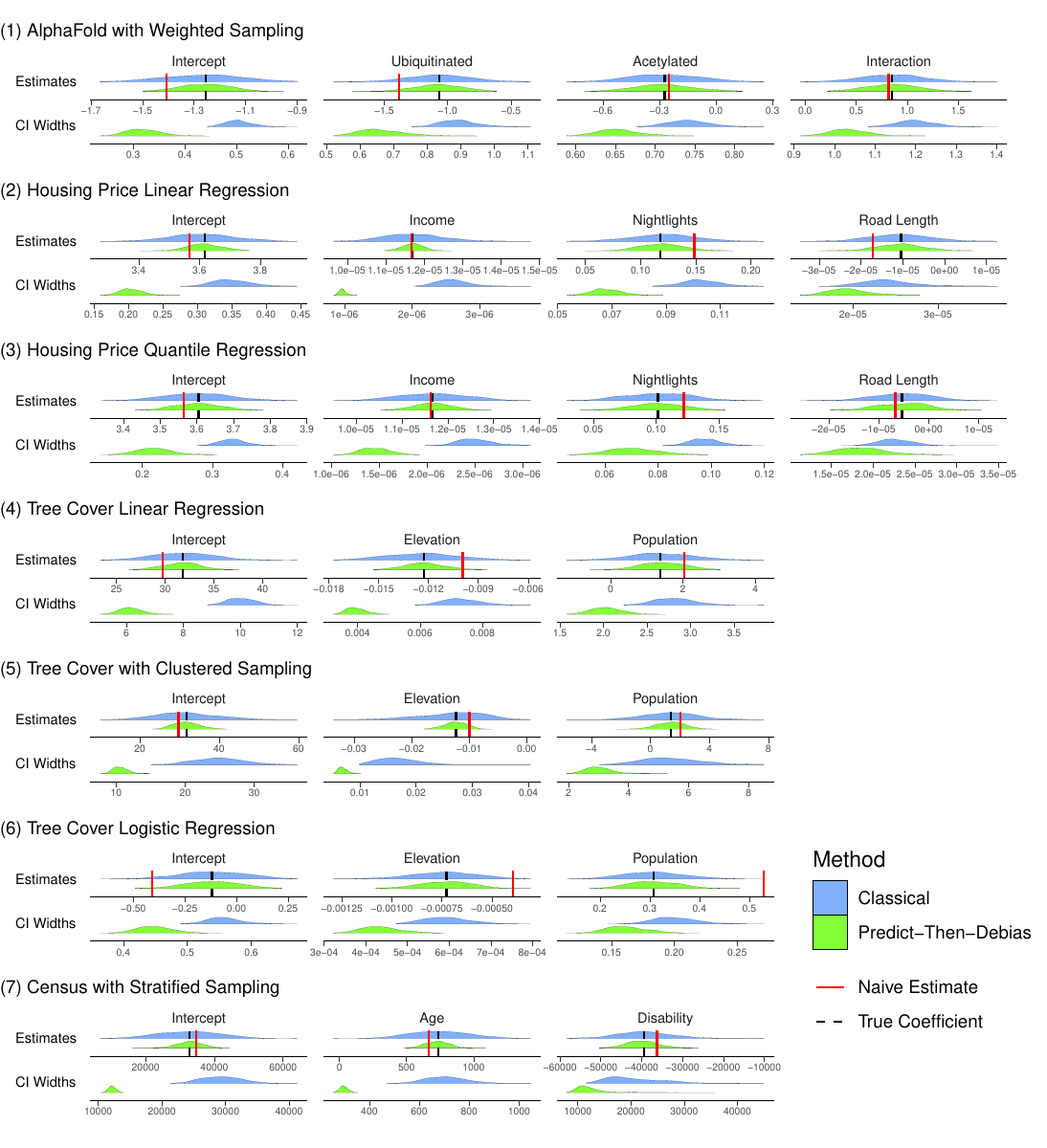}
    \caption{Half-violin plots of point estimates and confidence interval widths from the 7 experiments, each with 500 simulations. The panel names give the coefficient name and the number in parenthesis in each title gives the corresponding experiment number, according to the enumeration of experiments in Table \ref{table:ExperimentSummary}. For the green half-violin plots,  
    Algorithm \ref{alg:FullPercentileBootstrap} was used to construct confidence intervals (except for Experiments 5 and 7, where Algorithms \ref{alg:ClusterBootstrap} and \ref{alg:StratifiedBootstrap} were used, respectively). 
    }
    \label{fig:ViolinPlotFigure}
\end{figure}

For each experiment and regression coefficient, Figure \ref{fig:ViolinPlotFigure} gives the distribution across the 500 simulations of the point estimates and widths of the 90\% confidence intervals for the classical approach and for the PTD approach. In this figure, we only present the results for the PTD approach when the tuning matrix $\hat{\Omega}$ is chosen to estimate the asymptotically optimal diagonal tuning matrix given in Equation \eqref{eq:OmegaOptDiag} and when confidence intervals are calculated using fully bootstrap approaches (e.g., Algorithm \ref{alg:FullPercentileBootstrap}). Figure \ref{fig:ViolinPlotFigure} shows that in a fair number of cases the naive estimator has substantial bias, and that consistent with cautionary notes in the literature \citep{MythsEpi5Paper,Kluger24BiasDirection} the bias is often not attenuation towards zero. Meanwhile the classical and PTD estimators are consistently unbiased (they are centered on the dashed line). The PTD estimator consistently has lower variance and narrower confidence intervals than the classical estimator as guaranteed by Proposition \ref{prop:MoreEfficientThanclassical}.

Section \ref{sec:ExperimentsVaryingCIMethod} and Figure \ref{fig:ExperimentsVaryingCIMethod} present the empirical coverages and confidence interval widths for the faster alternatives to constructing confidence intervals for the PTD estimator. Meanwhile, Section \ref{sec:ExperimentsVaryingTuningMatrix} and Figure \ref{fig:ExperimentsVaryingTuningMatrix} presents the empirical coverages and confidence interval widths for other choices for the tuning matrix besides those estimating the asymptotically optimal diagonal tuning matrix.

\subsection{Comparing different confidence interval methods}\label{sec:ExperimentsVaryingCIMethod}

In Figure \ref{fig:ExperimentsVaryingCIMethod}, we present how the confidence interval widths and empirical coverage varied with the algorithm used to construct confidence intervals for $\htPPhom$. With the exception of the AlphaFold and the quantile regression experiments, in which mild overcoverage was observed, the confidence intervals for $\htPPhom$ had empirical coverages (across 500 simulations) that were close to the target coverage of $0.9$. The overcoverage in the AlphaFold experiment can be explained by the fact that each simulation involved sampling $7{,}500$ points without replacement from a dataset with $10{,}802$ points. Had we instead sampled with replacement or used a larger dataset, our superpopulation inference approach would not substantially overestimate the simulation scheme-specific variance of $\hgm$.

The 3 different approaches for constructing confidence intervals $\htPPhom$ yielded similar empirical coverages and confidence interval widths across the 7 experiments. Therefore, we recommend that investigators mainly consider runtime and ease-of-implementation when choosing between constructing CLT-based, convolution bootstrap-based, and fully bootstrap approaches to constructing confidence intervals for $\htPPhom$. Of these 3 approaches, CLT-based approaches are hardest to implement and generalize (requiring asymptotic variance calculations or an existing software implementation) but have the lowest runtime. On the other end of the spectrum, fully bootstrap approaches can be implemented in a few lines of code and require no asymptotic variance calculations but have the longest runtime. The convolution bootstrap approach is a compromise that leverages existing software that calculates asymptotic variance approximations and has intermediate runtime.

\begin{figure}[t]
    \centering 
    \includegraphics[width=0.95 \hsize]{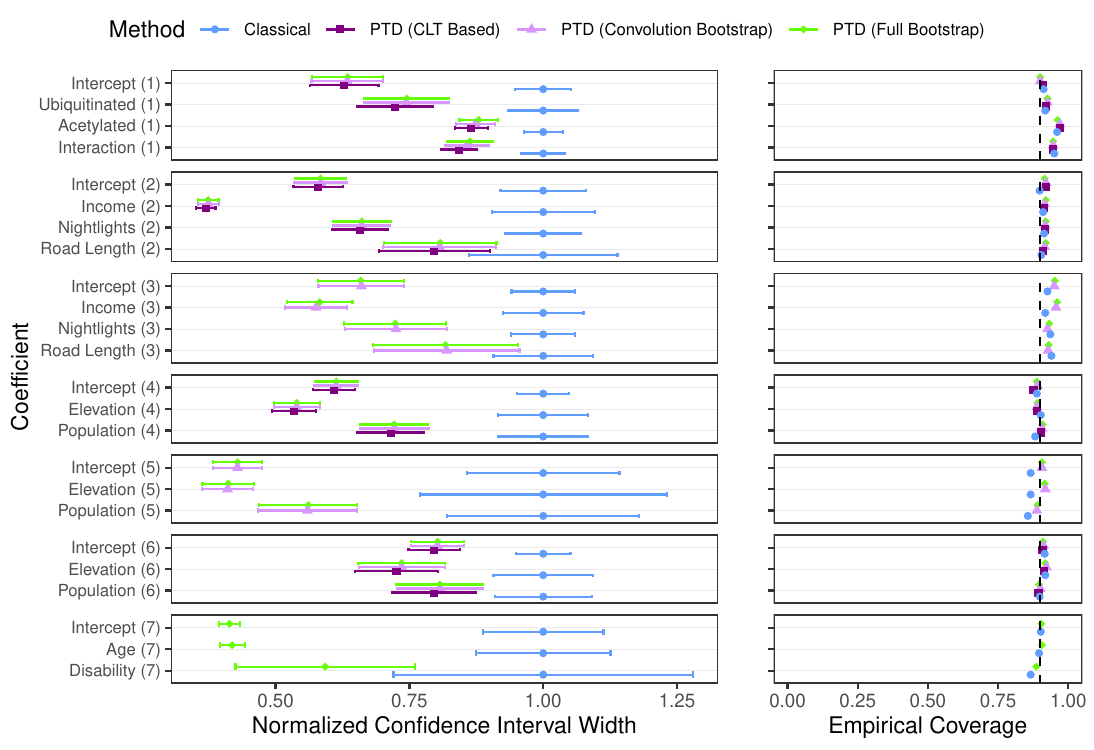}
    \caption{Confidence interval widths and empirical coverage for different confidence interval construction approaches.  Coefficients are normalized by the mean confidence interval width of the classical estimator. In the left column each point gives the average width of the 90\% confidence interval across 500 simulations for a given regression coefficient and method. The error bars give $\pm 1$ standard deviations of the confidence interval widths. The number in parenthesis on the y-axis denotes which experiment is being plotted, according to the enumeration of experiments in Table \ref{table:ExperimentSummary}. The right panel gives the empirical coverage across the 500 simulations for each method, experiment and coefficient, and the dashed vertical line is the desired coverage of 0.9. (CLT-based and convolution bootstrap-based speedups to the PTD method were not implemented in all instances, given that their implementation requires additional mathematical calculations.)}
    \label{fig:ExperimentsVaryingCIMethod}
\end{figure}

\subsection{Comparing different tuning matrix choices}\label{sec:ExperimentsVaryingTuningMatrix}

The results presented in Figures \ref{fig:ViolinPlotFigure} and \ref{fig:ExperimentsVaryingCIMethod} all use a tuning matrix $\hat{\Omega}$ that estimates the optimal diagonal tuning matrix $\Omega_{\text{opt}}^{\text{(diag)}}$ given in Equation \eqref{eq:OmegaOptDiag}. We next present additional results when using the tuning matrix $\hat{\Omega}_{\text{opt}}$ given in Equation \eqref{eq:HatOmOpt} (which estimates the optimal tuning matrix among all $d \times d$ matrices) and also when using the untuned PTD estimator (which has $\hat{\Omega}=I_{d \times d}$). In Figure \ref{fig:ExperimentsVaryingTuningMatrix}, for each regression coefficient, experiment, and tuning matrix choice we show the average and standard deviation of the 90\% confidence interval widths across the 500 simulations. For ease of comparison across coefficients and experiments, the confidence interval widths are only presented for full percentile bootstrap approaches (e.g., Algorithms \ref{alg:FullPercentileBootstrap}, \ref{alg:ClusterBootstrap}, and \ref{alg:StratifiedBootstrap}) and are normalized by the average confidence interval width of the classical estimator across the 500 simulations. Figure \ref{fig:ExperimentsVaryingTuningMatrix} also shows the empirical coverage for each tuning matrix choice.

\begin{figure}[t]
    \centering 
    \includegraphics[width=0.95 \hsize]{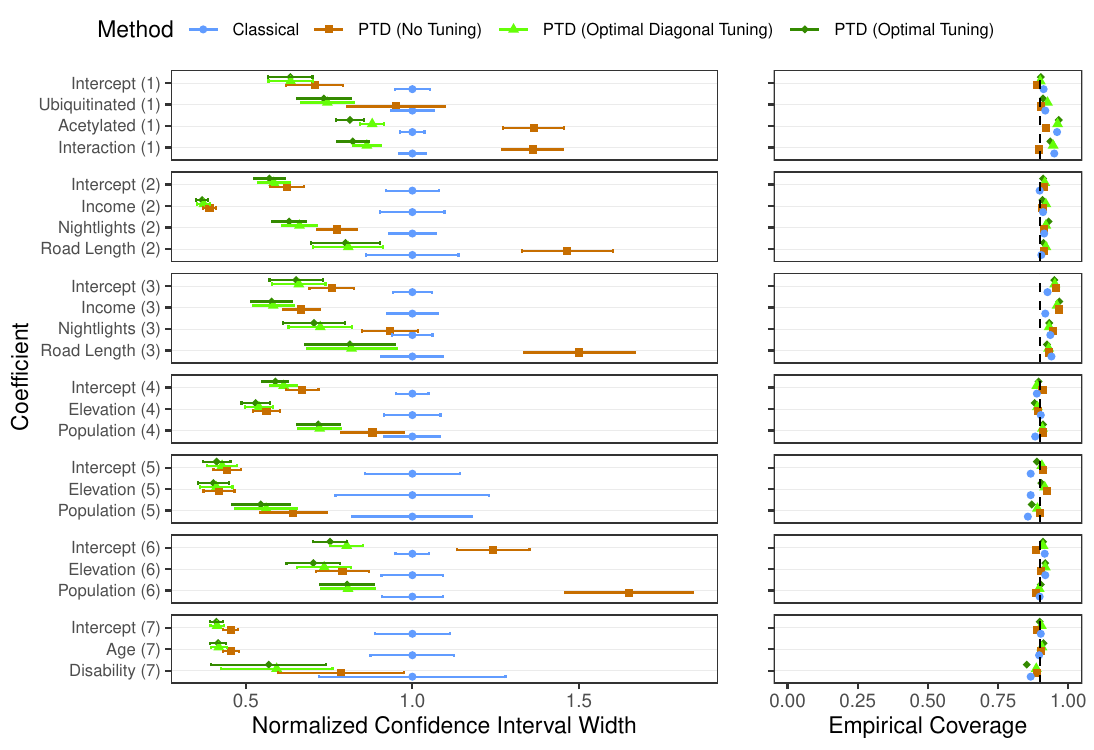}
    \caption{Confidence interval widths and empirical coverage for different tuning matrix choices. Other aspects of the plot are as in Figure~\ref{fig:ExperimentsVaryingCIMethod}.}
    \label{fig:ExperimentsVaryingTuningMatrix}
\end{figure}

The confidence intervals when using the optimal diagonal tuning matrix are typically similar in size to those when using the optimal tuning matrix. Moreover, the slightly narrower confidence intervals of the optimal tuning matrix comes at a cost of slightly poorer coverage, particularly in the stratified and clustered sampling experiments, which we suspect is driven by the high number of tuning parameters relative to the effective sample size. Meanwhile, the untuned PTD estimator sometimes has confidence intervals of comparable size to those of the optimally tuned PTD estimator, but it can also have confidence intervals that are much larger even than those of the classical estimator. Overall, this suggests that using a diagonal tuning matrix is enough to give intervals that are nearly as small as those given by the full tuning matrix, while an untuned PTD estimator is best avoided. Since the diagonal tuning is in general more stable because it has fewer parameters, we recommend this as the default choice.

\section*{Software and reproducibility code}

An \texttt{R} package with implementations of the Predict-Then-Debias bootstrap can be installed from \url{https://github.com/DanKluger/PTDBoot}. A Python implementation is available at \url{https://github.com/Earth-Intelligence-Lab/ppi_regression/tree/main}.
The code that was used to run the experiments and generate the figures can be found at \url{https://github.com/DanKluger/Predict-Then-Debias_Bootstrap}.

\section*{Acknowledgments}
This work was supported by the MIT Institute for Data Systems and Society Michael Hammer Postdoctoral Fellowship, by the U.S. Department of
Energy Computational Science Graduate Fellowship under Award Number DE-SC0023112, and by a Stanford Data Science postdoctoral fellowship. The authors also wish to thank Lihua Lei and Qidong Yang for helpful discussions and comments.

\bibliographystyle{apalike}
\bibliography{PTD_Bootstrap}

\appendix
\setcounter{table}{0}
\renewcommand{\thetable}{A\arabic{table}}
\setcounter{figure}{0}
\renewcommand{\thefigure}{A\arabic{figure}}

\newpage

 \section{Constructing confidence intervals using the CLT}\label{sec:CLTBasedCIsAlgorithmAndTheory}

To construct confidence intervals for $\htPPhom$ based on Proposition \ref{prop:PPEstimatorCLT}, the investigator needs to have consistent estimates of $\SigTC$, $\SigGC$, $\SigGM$, and $\SigTGC$. This is formalized in the following assumption, and remarks are given at the end of the section describing how an investigator can obtain consistent estimates of $\SigTC$, $\SigGC$, $\SigGM$, and $\SigTGC$.

\begin{assumption}(Consistent asymptotic variance estimators)\label{assump:ConsistentCovarianceMatrices} 
 The investigator can use the data to construct matrices $\hSigTC$, $\hSigGC$, $\hSigGM$ and $\hSigTGC$ such that as $\ntot \to \infty$, $\hSigTC \xrightarrow{p} \SigTC$, $\hSigGC \xrightarrow{p} \SigGC$, $\hSigGM \xrightarrow{p} \SigGM$, and $\hSigTGC \xrightarrow{p} \SigTGC$.
   
\end{assumption}

\begin{varalgorithm}{4}
\caption{Predict-Then-Debias CLT-based confidence intervals}\label{alg:CLTBasedCIs}
\begin{algorithmic}[1]
\State Compute $\hSigTC$, $\hSigGC$, $\hSigGM$, and $\hSigTGC$ satisfying Assumption \ref{assump:ConsistentCovarianceMatrices}
\State $\hgm \gets \calA(\tilde{\cx}; \cwmain)$ \Comment{Naive estimate of $\theta$ using the incomplete sample's proxy data}
\State $\hgc \gets \calA(\tilde{\cx}; \cwcalib)$ \Comment{Naive estimate of $\theta$ using the complete sample's proxy data}
\State $\htc \gets \calA(\cx; \cwcalib)$ \Comment{Estimate of $\theta$ using the complete sample's good data}
\State Select tuning matrix $\hat{\Omega}$ (e.g., using \ref{alg:CLTBasedTuningSubroutine} or Equation \eqref{eq:HatOmOpt})
\State $\htPPhom \gets \hat{\Omega} \hgm + (\htc -\hat{\Omega} \hgc) $
\State $\hat{\Sigma} \gets \hSigTC-\hSigTGC \hat{\Omega}^\tran -\hat{\Omega} [\hSigTGC]^\tran + \hat{\Omega} (\hSigGM+\hSigGC) \hat{\Omega}^\tran$
\State \Return  $\mathcal{C}_j^{1-\alpha} \gets \Big( \htPPhom_j- z_{1-\alpha/2} \sqrt{\hat{\Sigma}_{jj}/\ntot}, \htPPhom_j+z_{1-\alpha/2} \sqrt{\hat{\Sigma}_{jj}/\ntot}\Big) \quad \forall_{j \in \{1,\dots,d \}}$
\end{algorithmic}
\end{varalgorithm}

\begin{varalgorithm}{Subroutine 3}\floatname{algorithm}{}
\caption{Estimate optimal diagonal tuning matrix with plug-in estimator}\label{alg:CLTBasedTuningSubroutine}
\begin{algorithmic}[1]

\State  \textbf{Input} $\hSigGC$, $\hSigGM$, and $\hSigTGC$ which estimate $\SigGC$, $\SigGM$, and $\SigTGC$.
\State \Return  $\hat{\Omega} \gets \text{Diag} \Big(\frac{[\hSigTGC]_{11}}{[\hSigGC]_{11}+[\hSigGM]_{11}}, \dots, \frac{[\hSigTGC]_{dd}}{[\hSigGC]_{dd}+[\hSigGM]_{dd}} \Big)$ 
\end{algorithmic}
\end{varalgorithm}

\begin{corollary}\label{cor:CLTBasedCIsWork}
    Under Assumptions \ref{assump:SamplingLabelling}, \ref{assump:AsymptoticLinearity}, and \ref{assump:ConsistentCovarianceMatrices}, if $\hat{\Omega}=\Omega +o_p(1)$, Algorithm \ref{alg:CLTBasedCIs} returns asymptotically valid confidence intervals $\mathcal{C}_1^{1-\alpha},\dots,\mathcal{C}_d^{1-\alpha}$, in the sense that $$\lim\limits_{\ntot \to \infty} \mathbb{P}(\theta_{j} \in \mathcal{C}_j^{1-\alpha})=1-\alpha \quad \text{for all } j \in \{1,\dots,d\}.$$ 
\end{corollary}

\begin{proof}
     Note that by Assumptions \ref{assump:SamplingLabelling} and \ref{assump:AsymptoticLinearity} as well as Proposition \ref{prop:PPEstimatorCLT}, $\sqrt{\ntot}(\htPPhom-\theta) \xrightarrow{d} \mathcal{N}(0,\Sigma)$, where $$\Sigma=  \SigTC-\SigTGC \Omega^\tran -\Omega [ \SigTGC]^\tran +\Omega (\SigGM+\SigGC) \Omega^\tran.$$ By Assumption \ref{assump:ConsistentCovarianceMatrices} and since $\hat{\Omega} \xrightarrow{p} \Omega$, if we let $\hat{\Sigma}$ be the matrix returned by Algorithm \ref{alg:CLTBasedCIs}, $\hat{\Sigma} \xrightarrow{p} \Sigma$ by the continuous mapping theorem.

 Now fix $j \in \{1,\dots,d\}$. By the previous results $\sqrt{\ntot}(\htPPhom_j-\theta_j) \xrightarrow{d} \mathcal{N}(0,\Sigma_{jj})$ and $\hat{\Sigma}_{jj} \xrightarrow{p} \Sigma_{jj}$. Therefore by Slutsky's lemma, $\hat{\Sigma}_{jj}^{-1/2} \sqrt{\ntot}(\htPPhom_j-\theta_j) \xrightarrow{d} \mathcal{N}(0,1).$ Now let $\Phi(\cdot)$ denote the CDF of a standard normal distribution, let $z_{1-\alpha/2}$ be the $1-\alpha/2$ quantile of a standard normal, and let $\mathcal{C}_j^{1-\alpha}$ be the confidence interval returned by Algorithm \ref{alg:CLTBasedCIs}. Hence $$\begin{aligned} \lim_{\ntot \to \infty} \mathbb{P}(\theta_j \in \mathcal{C}_j^{1-\alpha}) & =  \lim_{\ntot \to \infty} \mathbb{P} \Big( \htPPhom_j- z_{1-\alpha/2} \sqrt{\hat{\Sigma}_{jj}/\ntot} < \theta_j <  \htPPhom_j+z_{1-\alpha/2} \sqrt{\hat{\Sigma}_{jj}/\ntot} \Big)
  \\ & =   \lim_{\ntot \to \infty} \mathbb{P} \Big(  -z_{1-\alpha/2}  < \hat{\Sigma}_{jj}^{-1/2}\sqrt{\ntot}(\htPPhom_j-\theta_j) <  z_{1-\alpha/2}  \Big)
   \\ & =   \lim_{\ntot \to \infty}  \mathbb{P} \big(  \hat{\Sigma}_{jj}^{-1/2}\sqrt{\ntot}(\htPPhom_j-\theta_j) <  z_{1-\alpha/2} \big) 
   \\ & \quad -  \lim_{\ntot \to \infty} \mathbb{P} \big(  \hat{\Sigma}_{jj}^{-1/2}\sqrt{\ntot}(\htPPhom_j-\theta_j) \leq  -z_{1-\alpha/2}\big) 
   \\ & = \Phi(z_{1-\alpha/2})-\Phi(-z_{1-\alpha/2})
   \\ & = 1-\alpha/2-\alpha/2=1-\alpha
 \end{aligned}$$ Above the penultimate step holds by the definition of convergence in distribution and continuity of $\Phi(\cdot)$, since $\hat{\Sigma}_{jj}^{-1/2} \sqrt{\ntot}(\htPPhom_j-\theta_j) \xrightarrow{d} \mathcal{N}(0,1)$, and the last step holds by symmetry of the standard normal and the fact that $z_{1-\alpha/2}$ is the $1-\alpha/2$ quantile of a standard normal.
\end{proof}

The following remarks give some guidance on how investigators can readily construct covariance matrix estimators that satisfy Assumption \ref{assump:ConsistentCovarianceMatrices}.

\begin{remark}\label{remark:SoftwareOftenGivesCovMatEst} When calculating $\htc = \calA(\cx; \cwcalib)$, $\hgc = \calA(\tilde{\cx}; \cwcalib)$, and $\hgm = \calA(\tilde{\cx}; \cwmain)$ via software that evaluates $\calA(\cdot ; \cdot)$ given a data matrix and weight vector as inputs, standard statistical software commonly also returns estimated covariance matrices $\widehat{\cov}(\htc)$, $\widehat{\cov}(\hgc)$, and $\widehat{\cov}(\hgm)$. Often, setting $\hSigTC=\ntot \widehat{\cov}(\htc)$, $\hSigGC=\ntot \widehat{\cov}(\hgc)$, $\hSigGM=\ntot \widehat{\cov}(\hgm)$ will give covariance matrices such that $\hSigTC \xrightarrow{p} \SigTC$, $\hSigGC \xrightarrow{p} \SigGC$, and $\hSigGM \xrightarrow{p} \SigGM$. For example, if $\calA(\cdot;\cdot)$ evaluates an M-estimator (such as a GLM or a sample quantile), under Assumptions \ref{assump:SamplingLabelling}, \ref{assump:AsymptoticLinearity} and other fairly mild regularity conditions, the sandwich estimators $\widehat{\cov}(\htc)$, $\widehat{\cov}(\hgc)$, and $\widehat{\cov}(\hgm)$ will be consistent for $\SigTC$, $\SigGC$, and $\SigGM$ after rescaling by $\ntot$. Moreover, for many common M-estimators, standard statistical software will compute sandwich estimators for $\calA(\cdot ; \cdot)$ given a data matrix and weight vector as input. 
\end{remark}

While $\SigTC$, $\SigGC$, and $\SigGM$ can often be consistently estimated using standard statistical software, most statistical software packages would not be equipped to estimate $\SigTGC$, because it is the asymptotic covariance matrix of two different estimators. That being said, in many cases it is possible for an investigator to derive a formula for $\SigTGC$ that can then be used to construct a consistent estimator of $\SigTGC$. In the next remark, we describe how to construct a consistent estimator of $\SigTGC$ for $M$-estimation tasks (and a similar analytic approach can also be taken to construct consistent estimators of $\SigTC$, $\SigGC$, and $\SigGM$).

\begin{remark}(Estimating $\SigTGC$  for M-estimation tasks) \label{remark:CrossCovMestimation} Using the same setup and notation as in Section \ref{sec:RegularityMestimation}, and assuming the regularity conditions of Proposition \ref{prop:MestimatorsAsymptoticallyLinear} and that the loss function is smooth enough for Hessians and expectations to be swapped, $\Psi(x)=-\e[\ddot{l}_{\theta}(X)]^{-1} \dot{l}_{\theta}(x)$ and $\tilde{\Psi}(\tilde{x})=-\e[\ddot{l}_{\gamma}(\tilde{X})]^{-1} \dot{l}_{\gamma}(\tilde{x})$. Moreover note that by Assumption $\ref{assump:SamplingLabelling}$ and the tower property $\e[W\Psi(X)]=0$, $\e[\tilde{W}\tilde{\Psi}(\tilde{X})]=0$, $\e[W\ddot{l}_{\theta}(X)]=\e[\ddot{l}_{\theta}(X)]$ and $\e[W\ddot{l}_{\gamma}(\tilde{X})]=\e[\ddot{l}_{\gamma}(\tilde{X})]$. Hence, $$\SigTGC = \e\big[W \Psi(X) [W \tilde{\Psi}(\tilde{X})]^\tran \big]=\big(\e[W\ddot{l}_{\theta}(X)] \big)^{-1} \e\big[W^2 \dot{l}_{\theta}(X) [\dot{l}_{\gamma}(\tilde{X})]^\tran \big] \big( \e[W\ddot{l}_{\gamma}(\tilde{X})] \big)^{-1}.$$ Setting $$\hat{D}_{1} = \frac{1}{\ntot} \sum_{i=1}^\ntot W_i \ddot{l}_{\htc}(X_i), \hat{C}_{12} = \frac{1}{\ntot} \sum_{i=1}^\ntot W_i^2 \dot{l}_{\htc}(X_i) [\dot{l}_{\hgc}(\tilde{X}_i)]^\tran, \text{ and }  \hat{D}_{2}= \frac{1}{\ntot} \sum_{i=1}^\ntot W_i \ddot{l}_{\hgc}(\tilde{X}_i),$$
it follows that under additional regularity conditions beyond those in Proposition \ref{prop:MestimatorsAsymptoticallyLinear} (e.g., if $\dot{l}_{\vartheta}$ and $\ddot{l}_{\vartheta}$ are locally Lipschitz in a neighborhood of $\vartheta=\theta$ and $\vartheta=\gamma$) $\hSigTGC=\hat{D}_1^{-1} \hat{C}_{12} \hat{D}_2^{-1} \xrightarrow{p} \SigTGC$. Therefore, under such regularity conditions, an investigator can consistently estimate $\SigTGC$ with the above estimator provided that they have a formula for the gradient and Hessian of the loss function with respect to $\vartheta$ evaluated at $\vartheta=\htc$ and $\vartheta=\hgc$.
\end{remark}

An alternative approach to estimating $\SigTGC$ that does not require analytic calculations is to use the bootstrap. The method in \cite{MiaoLuNeurIPS} proposes using the bootstrap to produce a consistent estimator of $\SigTGC$ and subsequently applying (a uniform sampling version of) Algorithm \ref{alg:CLTBasedCIs}. In this paper, we do not consider the approach of using the bootstrap-based estimator of $\SigTGC$ for use in Algorithm \ref{alg:CLTBasedCIs} for two reasons. First, Algorithm \ref{alg:CLTBasedCIs} is intended to have fast runtime whereas the bootstrap can be time consuming to implement. Second, if users are confronting a setting where it is difficult to find an analytic formula for $\SigTGC$, Algorithms \ref{alg:FullPercentileBootstrap} and \ref{alg:QuickConvolutionBootstrap} will provide asymptotically valid confidence intervals in certain settings of interest where the method in \cite{MiaoLuNeurIPS} could fail to provide asymptotically valid confidence intervals. In particular, in the settings of Theorems \ref{theorem:FullPercentileBootstrapCIsValid} (or Theorem \ref{theorem:GaussianConvBootCIsValid}), Algorithms \ref{alg:FullPercentileBootstrap} (or Algorithm \ref{alg:QuickConvolutionBootstrap}) provides asymptotically valid confidence intervals; however, the method in \cite{MiaoLuNeurIPS} uses a bootstrap-based estimator of the variance which requires an additional uniform integrability condition on the squared bootstrap pivots (e.g., see Chapter 3.1.6 of \cite{ShaoAndTuTextbook}). 


\section{Bootstrap for cluster or stratified sampling settings}\label{sec:ClusterAndStratifiedBootstrapMoreInDepth}

\subsection{Cluster bootstrap}\label{sec:ClusterBootstrapDescription}

Suppose that the $\ntot$ samples $(X_i,\tilde{X}_i)_{i=1}^\ntot$ are partitioned into $K$ clusters and that within each cluster, $X_i$ is either unobserved on all samples or is observed on all samples. In particular, let $C_1,\dots,C_K$ denote clusters which form a partition of $\{1,\dots,\ntot\}$ (i.e., for each $k$, $C_k \subset \{1,\dots,\ntot\}$ satisfying $\cup_{k=1}^K C_k = \{1,\dots,\ntot\}$ and $C_k \cap C_{k'} = \varnothing$ for all $k \neq k'$). In the cluster labelling schemes that we consider, $X_i$ is originally unobserved on all samples and then subsequently measured via the following procedure:

\begin{enumerate}
    \item Draw $\xi_1,\dots, \xi_K \stackrel{\text{Ind.}}{\sim} \text{Bernoulli}(\pi_k)$ for some $\pi_1,\dots,\pi_K \in (0,1)$.
    \item For each $k \in \{1,\dots,K\}$, if $\xi_k=1$, collect observations of $X_i$ for each $i \in C_k$ and if $\xi_k=0$ forgo collecting the $X_i$ observations for which $i \in C_k$.
\end{enumerate}
 
 According to the above sampling scheme, note that for each $k \in \{1,\dots,K\}$,  the inverse probability weights are given by $W_i=\xi_k/\pi_k$ and $\bar{W}_i=(1-\xi_k)/(1-\pi_k)$ for each $i \in C_k$. Letting $\cwcalib=(W_1,\dots,W_{\ntot})$ and $\cwmain=(\bar{W}_1,\dots,\bar{W}_{\ntot})$, Algorithm \ref{alg:ClusterBootstrap} gives a cluster bootstrap modification to Algorithm \ref{alg:FullPercentileBootstrap} that corrects confidence interval widths to account for the correlations induced by above sampling scheme.

\begin{varalgorithm}{5}
\caption{Predict-Then-Debias Bootstrap (cluster sampling setting)}\label{alg:ClusterBootstrap}
\begin{algorithmic}[1]
\For{$b= 1,\dots,B$}
\State Sample $k_1,k_2,\dots,k_K \stackrel{\text{iid}}{\sim} \text{Unif}(\{1,\dots,K\})$
\State $\mathcal{I} \gets \text{concatenate}(C_{k_1},\dots,C_{k_K})$ \Comment{$\vert \ci \vert$ may not equal $\ntot$}
\State $(\cwcalibStar,\cwmainStar, \cx^*,\tilde{\cx}^*) \gets (\cwcalib_{\ci},\cwmain_{\ci},\cx_{\ci \cdot},\tilde{\cx}_{\ci \cdot})$
\State $\htcarg{,\bparen} \gets \calA(\cx^*; \cwcalibStar)$
\State $\hgcarg{,\bparen} \gets \calA(\tilde{\cx}^*; \cwcalibStar)$
\State $\hgmarg{,\bparen} \gets \calA(\tilde{\cx}^*; \cwmainStar)$ \label{line:CalcHgmInClusterBoot}
\EndFor
\State Select tuning matrix $\hat{\Omega}$ (e.g., using \ref{alg:FullBootTuningSubroutine})
\State $\htPPhomB \gets \hat{\Omega} \hgmarg{,\bparen} + (\htcarg{,\bparen}-\hat{\Omega} \hgcarg{,\bparen} )$ for $b=1,\dots, B$
\State \Return  $\mathcal{C}_j^{1-\alpha} \gets \big( \text{Quantile}_{\alpha/2}( \{\htPPhomB_j\}_{b=1}^B), \text{Quantile}_{1-\alpha/2}( \{\htPPhomB_j\}_{b=1}^B ) \big) \quad \forall_{j \in \{1,\dots,d\}}$

\end{algorithmic}
\end{varalgorithm}

We note that Algorithm \ref{alg:ClusterBootstrap} can also be sped up using the same convolution approach as in Algorithm \ref{alg:QuickConvolutionBootstrap} when $\hgm$ is asymptotically normal. In particular one can replace the relatively slow step in Line \ref{line:CalcHgmInClusterBoot} of Algorithm \ref{alg:ClusterBootstrap}, by independently drawing $\hgmarg{,\bparen} \sim \mathcal{N}\big(\hgm, \widehat{\cov}(\hgm) \big)$, where $\widehat{\cov}(\hgm)$ is the estimated covariance matrix of $\hgm$ that only needs to be calculated once outside of the for loop. To use this speedup, $\widehat{\cov}(\hgm)$ should appropriately account for the cluster sampling scheme (e.g., this can be done fairly quickly using the \texttt{vcovCL} function in R).

\subsection{Bootstrap for stratified sampling}\label{sec:StratifiedBootstrapDescription}

Suppose that the population of samples $(X_i,\tilde{X}_i)_{i=1}^M$ can be partitioned into $K$ disjoint strata $S_1,\dots,S_K$. In particular, let $S_1,\dots,S_K$ denote strata which form a partition of $\{1,\dots,M\}$ (i.e., for each $k$, $S_k \subset \{1,\dots,M\}$ satisfying $\cup_{k=1}^K S_k = \{1,\dots, M \}$ and $S_k \cap S_{k'} = \varnothing$ for all $k \neq k'$). We consider stratified sampling schemes where initially none of the samples $(X_i,\tilde{X}_i)_{i=1}^M$ are observed and then subsequently for each $k=1,\dots,K$ the investigator:

\begin{enumerate}
    \item Obtains a uniform random subsample without replacement of $\big((X_i,\tilde{X}_i) \big)_{i \in S_k}$ of size $n_k^{\bullet}$. For each $i \in S_k$ let $I_i^{\bullet}$ be an indicator of whether $(X_i,\tilde{X}_i)$ was observed in this subsample. 
    \item Obtains a uniform random subsample without replacement of $(\tilde{X}_i )_{i \in S_k}$ of size $n_k^{\circ}$ and forgoes collecting $X_i$ observations on this subsample. For each $i \in S_k$ let $I_i^{\circ}$ be an indicator of whether $\tilde{X}_i$ was observed in this subsample.  
\end{enumerate}
 
 According to the above sampling scheme, note that for each $k \in \{1,\dots,K\}$, the inverse probability weights are given by $W_i^{\bullet}=I_i^{\bullet}  \vert S_k \vert /n_k^{\bullet}$ and $W_i^{\circ}= I_i^{\circ}\vert S_k \vert /n_k^{\circ}$ for each $i \in S_k$. Letting $\cwcalib=(W_1^{\bullet},\dots,W_M^{\bullet})$ and $\cwmain=(W^{\circ}_1,\dots,W_M^{\circ})$, Algorithm \ref{alg:StratifiedBootstrap} gives a stratified bootstrap modification to Algorithm \ref{alg:FullPercentileBootstrap} that corrects confidence interval widths to account for the above sampling scheme. We remark that this version of the stratified bootstrap is only designed to work in settings where there is a small number of large strata.

\begin{varalgorithm}{6}
\caption{Predict-Then-Debias Bootstrap (stratified sampling setting)}\label{alg:StratifiedBootstrap}
\begin{algorithmic}[1]
\For{$b= 1,\dots,B$}
\For{$k=1,\dots,K$}
\State  Sample $i_1,\dots,i_{n_k^{\bullet}} \stackrel{\text{iid}}{\sim} \text{Unif}( \{i \in S_k  \ : I_i^{\bullet}=1\})$
\State  Sample $l_1,\dots,l_{n_k^{\circ}} \stackrel{\text{iid}}{\sim} \text{Unif}( \{i \in S_k  \ : I_i^{\circ}=1\})$
\State $\mathcal{I}_k^{\bullet} \gets (i_1,\dots,i_{n_k^{\bullet}})$ \quad \text{and} \quad $\mathcal{I}_k^{\circ} \gets (l_1,\dots,l_{n_k^{\circ}})$
\EndFor
\State $\mathcal{I}^{\bullet} \gets \text{concatenate}(\mathcal{I}_1^{\bullet},\dots,\mathcal{I}_K^{\bullet})$  
\State $\mathcal{I}^{\circ} \gets \text{concatenate}(\mathcal{I}_1^{\circ},\dots,\mathcal{I}_K^{\circ})$ 
\State $(\cwcalibStar, \cx^{\bullet,*},\tilde{\cx}^{\bullet,*}) \gets (\cwcalib_{\ci^{\bullet}},\cx_{\ci^{\bullet} \cdot},\tilde{\cx}_{\ci^{\bullet} \cdot})$
\State $(\cwmainStar,\tilde{\cx}^{\circ,*}) \gets (\cwmain_{\ci^{\circ}},\tilde{\cx}_{\ci^{\circ} \cdot})$
\State $\htcarg{,\bparen} \gets \calA(\cx^{\bullet,*}; \cwcalibStar)$, $\hgcarg{,\bparen} \gets \calA(\tilde{\cx}^{\bullet,*}; \cwcalibStar)$, \text{ and } $\hgmarg{,\bparen} \gets \calA(\tilde{\cx}^{\circ,*}; \cwmainStar)$ 
\EndFor
\State Select tuning matrix $\hat{\Omega}$ (e.g., using   \ref{alg:FullBootTuningSubroutine})
\State $\htPPhomB \gets \hat{\Omega} \hgmarg{,\bparen} + (\htcarg{,\bparen}-\hat{\Omega} \hgcarg{,\bparen} )$ for $b=1,\dots, B$
\State \Return  $\mathcal{C}_j^{1-\alpha} \gets \big( \text{Quantile}_{\alpha/2}( \{\htPPhomB_j\}_{b=1}^B), \text{Quantile}_{1-\alpha/2}( \{\htPPhomB_j\}_{b=1}^B ) \big) \quad \forall_{j \in \{1,\dots,d\}}$

\end{algorithmic}
\end{varalgorithm}

\section{Proofs of theoretical results from Section \ref{sec:TheoreticalSectionPointEst}}\label{sec:ProofOfPropositions}

In this appendix we prove the propositions displayed in Section \ref{sec:TheoreticalSectionPointEst}. The proofs of those propositions, as well as those of many other theoretical results in the paper, rely on the following lemma that establishes asymptotic normality of $\hat{\zeta}$. Before presenting this lemma and these proofs we display formulas from the main text that are regularly used in these proofs and we also present some helpful notation. In particular, recall that using Assumption \ref{assump:AsymptoticLinearity}, we let $\Psi,\tilde{\Psi} : \mathbb{R}^p \to \mathbb{R}^d$ be functions such that each component of $\big(\Psi(X),\tilde{\Psi}(\tilde{X}) \big)$ has mean $0$ and finite variance, and such that $$\sqrt{\ntot} \Bigg( \begin{bmatrix}
    \htc \\ \hgc \\ \hgm \end{bmatrix}  -\begin{bmatrix}
    \theta \\ \gamma \\ \gamma \end{bmatrix}- \frac{1}{\ntot} \sum_{i=1}^\ntot \begin{bmatrix}  W_i \Psi(X_i)  \\  W_i \tilde{\Psi}(\tilde{X}_i) \\   \bar{W}_i \tilde{\Psi}(\tilde{X}_i) \end{bmatrix} \Bigg) \xrightarrow{p} 0.$$ Further define
\begin{equation}\label{eq:zetaSigmaZetaDef}
    \hat{\zeta} \equiv \begin{bmatrix} \htc \\  \hgc \\ \hgm \end{bmatrix}, \quad  \zeta \equiv \begin{bmatrix} \theta \\ \gamma \\ \gamma \end{bmatrix}, \quad \text{and} \quad \Sigma_{\zeta} \equiv \begin{bmatrix}
    \SigTC & \SigTGC & 0
    \\ [\SigTGC]^\tran & \SigGC & 0
    \\ 0 & 0 & \SigGM
\end{bmatrix},\end{equation} where recall that each of the displayed blocks of $\Sigma_{\zeta}$ is a $d \times d$ matrix, and $\SigTC \equiv \var \big( W \Psi(X) \big)$, $\SigGC \equiv \var \big( W \tilde{\Psi}(\tilde{X}) \big)$, $\SigTGC \equiv \cov \big( W \Psi(X),W \tilde{\Psi}(\tilde{X}) \big),$ and  $\SigGM \equiv \var \big( \bar{W} \tilde{\Psi}(\tilde{X}) \big).$

\begin{lemma}\label{lemma:JointCLTzeta} Under Assumptions \ref{assump:SamplingLabelling} and \ref{assump:AsymptoticLinearity}, $\sqrt{\ntot}(\hat{\zeta}-\zeta) \xrightarrow{d} \mathcal{N}(0,\Sigma_{\zeta}).$
\end{lemma} \begin{proof}
    
Let $V_i \equiv \big(W_i \Psi(X_i), W_i \tilde{\Psi}(\tilde{X}_i), \bar{W}_i \tilde{\Psi}(\tilde{X}_i) \big)$, and let $V$ denote a random variable with the same distribution as $V_i$. Note that by Assumption \ref{assump:AsymptoticLinearity} and rearranging terms, \begin{equation}\label{eq:ALexpansionAll}
    \sqrt{\ntot}(\hat{\zeta}-\zeta)= \sqrt{\ntot} \Bigg( \begin{bmatrix}  \htc \\  \hgc \\  \hgm \end{bmatrix} - \begin{bmatrix}  \theta \\  \gamma \\  \gamma \end{bmatrix} \Bigg) = \frac{1}{\sqrt{\ntot}} \sum_{i=1}^\ntot V_i +o_p(1),
\end{equation} where
$o_p(1)$ denotes a vector of terms that converge to $0$ in probability as $\ntot \to \infty$. Note that under Assumption \ref{assump:SamplingLabelling}, $$\e[V] = \begin{bmatrix} \e \big[ \e[W \Psi(X) \giv \tilde{X}]  \big] \\ \e \big[ \e[W \tilde{\Psi}(\tilde{X}) \giv \tilde{X} ] \big] \\ \e\big[ \e[\bar{W} \tilde{\Psi}(\tilde{X}) \giv \tilde{X} ] \big] 
\end{bmatrix} = \begin{bmatrix} \e \big[ \Psi(X) \e[W  \giv \tilde{X}]  \big] \\ \e \big[ \tilde{\Psi}(\tilde{X}) \e[W \giv \tilde{X} ] \big] \\ \e\big[ \tilde{\Psi}(\tilde{X}) \e[\bar{W}  \giv \tilde{X} ] \big] 
\end{bmatrix}  =\begin{bmatrix} \e [ \Psi(X)]  \\ \e [ \tilde{\Psi}(\tilde{X})  ]  \\ \e[ \tilde{\Psi}(\tilde{X}) ]
\end{bmatrix} =0.$$ Above the first equality holds by the tower property. The second equality above holds because $W \indep X \giv \tilde{X}$ (by the missing at random assumption in Assumption \ref{assump:SamplingLabelling}, $I \indep X \giv \tilde{X}$ and $W$ is a deterministic function of $I$ and $\tilde{X}$). The third equality above holds by linearity of conditional expectation and by our definitions that $W=I/\pi(\tilde{X})$, $\bar{W}=(1-I)/(1-\pi(\tilde{X}))$, and $\pi(\tilde{x})=\mathbb{P}(I=1\giv \tilde{X}= \tilde{x}) \Rightarrow \pi(\tilde{X})= \e[I \giv \tilde{X}]$. The final equality above holds by Assumption \ref{assump:AsymptoticLinearity}.

Since $\e[V]=0$, by combining Equation \eqref{eq:ALexpansionAll}, the multivariate CLT, and Slutsky's lemma, it is clear that $\sqrt{\ntot}(\hat{\zeta}-\zeta) \xrightarrow{d} \mathcal{N}(0, \Sigma_V )$, where $$\Sigma_V = \begin{bmatrix}
    \var \big(W \Psi(X) \big) & \cov \big(W \Psi(X),W \tilde{\Psi}(\tilde{X}) \big) & \cov \big(W \Psi(X), \bar{W} \tilde{\Psi}(\tilde{X}) \big)
    \\  \cov \big(W \tilde{\Psi}(\tilde{X}),W \Psi(X) \big) & \var \big(W \tilde{\Psi}(\tilde{X}) \big) & \cov \big(W \tilde{\Psi}(\tilde{X}), \bar{W} \tilde{\Psi}(\tilde{X}) \big)
    \\   \cov \big( \bar{W} \tilde{\Psi}(\tilde{X}),W \Psi(X) \big) & \cov \big( \bar{W} \tilde{\Psi}(\tilde{X}), W \tilde{\Psi}(\tilde{X}) \big) & \var \big( \bar{W} \tilde{\Psi}(\tilde{X}) \big)
\end{bmatrix}.$$ Thus to complete the proof it remains to show $\Sigma_V=\Sigma_{\zeta}$ defined in Equation \eqref{eq:zetaSigmaZetaDef}.
Letting $B_1 \equiv \cov \big(W \Psi(X), \bar{W} \tilde{\Psi}(\tilde{X}) \big)$ and $B_2 \equiv \cov \big(W \tilde{\Psi}(\tilde{X}), \bar{W} \tilde{\Psi}(\tilde{X}) \big)$, and recalling the formulas for $\SigTC$, $\SigGC$,$\SigGM$, and $\SigTGC$ it is clear that $$\Sigma_V = \begin{bmatrix}
    \SigTC & \SigTGC & B_1
    \\ [ \SigTGC]^\tran & \SigGC & B_2
    \\ B_1^\tran & B_2^\tran & \SigGM.
\end{bmatrix}$$ Hence it suffices to show that $B_1=0$ and $B_2=0$. To do this recall that since $\e[V]=0$, $\e[W \Psi(X)]=0$,  $\e[W \tilde{\Psi}(\tilde{X})]=0$, and  $\e[\bar{W} \tilde{\Psi}(\tilde{X})]=0$. Hence, $$B_1=\e\Big[ W \Psi(X) \big[ \bar{W} \tilde{\Psi}(\tilde{X}) \big]^\tran \Big]=\e \Big[ \frac{I(1-I)}{\pi(\tilde{X})\big(1-\pi(\tilde{X})\big)} \Psi(X) \big[ \tilde{\Psi}(\tilde{X}) \big]^\tran \Big]=0$$ and $$B_2=\e\Big[ W \tilde{\Psi}(\tilde{X}) \big[ \bar{W} \tilde{\Psi}(\tilde{X}) \big]^\tran \Big]=\e \Big[ \frac{I(1-I)}{\pi(\tilde{X})\big(1-\pi(\tilde{X})\big)} \tilde{\Psi}(\tilde{X}) \big[ \tilde{\Psi}(\tilde{X}) \big]^\tran \Big]=0.$$ The last steps in the above two displays hold because
$I \in \{0,1\}$ always and therefore $I(1-I)=0$ always.

\end{proof}

\subsection{Proof of Proposition \ref{prop:PPEstimatorCLT}}

Fix $\Omega$ such that $\hat{\Omega} \xrightarrow{p} \Omega$. Letting $A_{\Omega}=\begin{bmatrix} I_{d \times d} & -\Omega & \Omega \end{bmatrix}$ and $A_{\hat{\Omega}}=\begin{bmatrix} I_{d \times d} & -\hat{\Omega} & \hat{\Omega} \end{bmatrix}$, observe that $A_{\hat{\Omega}}=A_{\Omega}+o_p(1)$. Note that because $\hat{\zeta}=(\htc,\hgc,\hgm)$, $\zeta=(\theta,\gamma,\gamma)$, and $\htPPhom=\hat{\Omega} \hgm + (\htc -\hat{\Omega} \hgc)$, $A_{\hat{\Omega}} \hat{\zeta}=\htPPhom$ and $A_{\hat{\Omega}} \zeta=\theta$. Thus applying Lemma \ref{lemma:JointCLTzeta} and Slutsky's lemma gives $$\sqrt{\ntot}(\htPPhom -\theta) =A_{\hat{\Omega}} \big(\sqrt{\ntot}(\hat{\zeta} -\zeta) \big)=A_{\Omega} \big(\sqrt{\ntot}(\hat{\zeta} -\zeta) \big)+o_p(1) \xrightarrow{d} \mathcal{N}(0, A_{\Omega} \Sigma_{\zeta} A_{\Omega}^\tran),$$
where the second equality follows because $\sqrt{\ntot}(\hat{\zeta} -\zeta) = O_p(1)$.
    By matrix multiplication and by Equation \eqref{eq:zetaSigmaZetaDef}, $$A_{\Omega} \Sigma_{\zeta} A_{\Omega}^\tran= \SigTC-\SigTGC \Omega^\tran -\Omega [ \SigTGC]^\tran +\Omega ( \SigGC + \SigGM) \Omega^\tran \equiv \SigPTD(\Omega).$$

\subsection{Proof of Proposition \ref{prop:MoreEfficientThanclassical}}

First note that $\sqrt{\ntot}(\htc-\theta) \xrightarrow{d} \mathcal{N}(0,\SigTC)$, by considering the first $d$ coordinates in the CLT from Lemma \ref{lemma:JointCLTzeta}. Now recall that by Formula \eqref{eq:OptOmegaAsymp}, $\Omega_{\text{opt}} \equiv \SigTGC (\SigGC + \SigGM)^{-1}$ and by Formula \eqref{eq:PPOmAsympVar}, $\SigPTD(\Omega) \equiv \SigTC-\SigTGC \Omega^\tran -\Omega [ \SigTGC]^\tran +\Omega ( \SigGC + \SigGM) \Omega^\tran$. Hence, $$\SigPTD(\Omega_{\text{opt}})=  \SigTC-\SigTGC (\SigGC+\SigGM)^{-1}[ \SigTGC]^\tran \equiv \Sigma_{\textnormal{TPTD}}.$$ Thus by Proposition \ref{prop:PPEstimatorCLT}, $\sqrt{\ntot} \big( \htPPhomOpt - \theta \big) \xrightarrow{d} \mathcal{N} (0,\Sigma_{\textnormal{TPTD}}).$ To complete the proof note that $\SigGC \succeq 0$ and $\SigGM \succeq 0$ and therefore $\SigTGC (\SigGC+\SigGM)^{-1}[ \SigTGC]^\tran \succeq 0$. Thus $\Sigma_{\textnormal{TPTD}} \preceq \SigTC$.

\section{Proofs of theoretical results from Section \ref{sec:BootstrapCIOverallSection}}\label{sec:BootstrapMethodProofs}

In this appendix, we prove Theorems \ref{theorem:FullPercentileBootstrapCIsValid} and \ref{theorem:GaussianConvBootCIsValid} which give guarantees that under certain assumptions, Algorithms \ref{alg:FullPercentileBootstrap} and \ref{alg:QuickConvolutionBootstrap} give asymptotically valid confidence intervals. Because the proofs are lengthy they are broken into 3 main parts:

\begin{enumerate}
    \item The first part is showing that the bootstrap is consistent for various pivots of interest. In particular, we refresh the reader with Assumption \ref{assump:ZetaBootstrapConsistency}, which gives Bootstrap consistency for bootstrapped pivots of the form $\sqrt{\ntot}(v^\tran \hat{\zeta}^*-v^\tran \hat{\zeta})$ where $v \in \mathbb{R}^{3d}$, and we also introduce some helpful notation for studying bootstrap consistency (we point the reader to Section \ref{sec:WhenDoWeHaveZetaBootstrapConsistency} for sufficient conditions under which Assumption \ref{assump:ZetaBootstrapConsistency} holds). Subsequently, in Theorem \ref{theorem:ConsistencyOfGaussianConvBoot} we prove bootstrap consistency for pivots that arise in Algorithm \ref{alg:QuickConvolutionBootstrap} that are approximately drawn from the bootstrap distribution, which is a critical step in the proof of Theorem \ref{theorem:GaussianConvBootCIsValid}.
    \item The second part provides auxiliary lemmas, including Lemma \ref{lemma:IgnoreOp1TermsForBootCI} that gives conditions under which the percentile bootstrap leads to asymptotically valid confidence intervals. A notable difference between Lemma \ref{lemma:IgnoreOp1TermsForBootCI} and standard results about when the percentile bootstrap is valid (e.g., Theorem 4.1 in \cite{ShaoAndTuTextbook}), is that Lemma \ref{lemma:IgnoreOp1TermsForBootCI} allows $o_p(1)$ terms to be ignored.
    \item The third part is proofs of asymptotic validity of the confidence intervals from Algorithm \ref{alg:FullPercentileBootstrap} (see Section \ref{sec:ProofOfFullBootPercentileCorollary}) and Algorithm \ref{alg:QuickConvolutionBootstrap} (see Section \ref{sec:ProofOfGaussianConvCIValidity}). These proofs piece together the specific algorithms and assumptions by applying bootstrap consistency results and the auxiliary lemma about asymptotic validity of percentile bootstrap confidence intervals.
\end{enumerate}

\subsection{Bootstrap consistency notation and results}

We now reintroduce the notions of the bootstrap distribution and of bootstrap consistency, using notation convenient to our setting. Recall that $V_i=(W_i,\bar{W}_i,X_i,\tilde{X}_i)$. We let $\hat{\mathbb{P}}_{\ntot}=\frac{1}{\ntot} \sum_{i=1}^\ntot \delta_{V_i}$ denote the empirical distribution of $V_i$ from the $\ntot$ samples, which can be thought of as a random distribution. Fixing $\hat{\mathbb{P}}_{\ntot}$, we draw $V_1^*,\dots,V_{\ntot}^* \stackrel{\text{iid}}{\sim} \hat{\mathbb{P}}_{\ntot}$ and use $\mathbb{P}_*(A \giv \hat{\mathbb{P}}_{\ntot})$ to denote the probability of an event that depends on the values of $V_1^*,\dots,V_{\ntot}^*$. After drawing $V_1^*,\dots,V_{\ntot}^* \stackrel{\text{iid}}{\sim} \hat{\mathbb{P}}_{\ntot}$, we define $(W_i^*,\bar{W}_i^*,X_i^*,\tilde{X}_i^*)=V_i^*$ to be the components of $V_i^*$ for each $i$, we define $$
\cwcalibStar=\begin{bmatrix}
    W_1^* \\ \vdots \\ W_{\ntot}^* \end{bmatrix}, 
\cwmainStar = \begin{bmatrix}
    \bar{W}_1^* \\ \vdots \\ \bar{W}_{\ntot}^* \end{bmatrix}, 
\cx^* = \begin{bmatrix}
   \horzbar & X_1^{*\tran} & \horzbar \\  & \vdots & \\ \horzbar & X_{\ntot}^{*\tran} & \horzbar \end{bmatrix},
\tilde{\cx}^* = \begin{bmatrix}
\horzbar & \tilde{X}_1^{*\tran} & \horzbar \\  & \vdots & \\ \horzbar & \tilde{X}_{\ntot}^{*\tran} & \horzbar \end{bmatrix},$$ and we define $$\htcarg{,*} = \calA(\cx^*; \cwcalibStar), \hgcarg{,*} = \calA(\tilde{\cx}^*; \cwcalibStar),\hgmarg{,*} = \calA(\tilde{\cx}^*; \cwmainStar), $$ and $\zeta^*=(\htcarg{,*}, \hgcarg{,*} ,\hgmarg{,*})$. We call the distribution of $\zeta^*$ generated by this procedure under a fixed $\hat{\mathbb{P}}_{\ntot}$ the bootstrap distribution of $\hat{\zeta}$. Note that even though many values of $X_i$ (and in turn rows of $\cx^*$) are unobserved, for any draw $\hat{\zeta}^*$ from the bootstrap distribution, $\hat{\zeta}^*$ is still observed and can be evaluated because whenever the $i$th row of $\cx^*$ is missing, the corresponding weight $W_i^*=0$.

We now introduce the notions of pivots and bootstrap consistency. Given a random variable (often called a pivot) $R_{\ntot}$ which depends on $\ntot$ samples of the data drawn from $\mathbb{P}$ and a procedure to randomly generate $R_{\ntot}^*$ that depends on the empirical distribution $\hat{\mathbb{P}}_{\ntot}$, the bootstrap distribution of $R_{\ntot}^*$ is said to be consistent (with respect to the sup-norm) if $$\sup\limits_{x \in \mathbb{R}} \vert \mathbb{P}_*( R_{\ntot}^* \leq x \giv \hat{\mathbb{P}}_{\ntot})- \mathbb{P}(R_{\ntot} \leq x) \vert \xrightarrow{p} 0.$$ When proving and leveraging bootstrap consistency results it is convenient to define the metric $\rho_{\infty}$ on the collection of CDFs such that for any (possibly random) CDFs $H_1$ and $H_2$, 
\begin{equation}\label{eq:rhoSupNorm}
    \rho_{\infty}(H_1,H_2)=\sup\limits_{x \in \mathbb{R}}\vert H_1(x)-H_2(x) \vert.
\end{equation} Therefore if $H_{\text{Boot}}$ and $H$ are the random CDF and CDF given by $H_{\text{Boot}}(x)=\mathbb{P}_*( R_{\ntot}^* \leq x \giv \hat{\mathbb{P}}_{\ntot})$ and $H(x)=\mathbb{P}(R_{\ntot} \leq x)$, respectively, then equivalently the bootstrap distribution of $R_{\ntot}^*$ is consistent if $\rho_{\infty}(H_{\text{Boot}},H)=o_p(1)$. Another way of restating Assumption \ref{assump:ZetaBootstrapConsistency}(i) is that the bootstrap distribution is consistent for $\sqrt{\ntot}(v^\tran \hat{\zeta}^*-v^\tran \hat{\zeta})$ for any arbitrary $v \in \mathbb{R}^{3d}$ (we defer sufficient conditions under which Assumption \ref{assump:ZetaBootstrapConsistency} holds to Section \ref{sec:WhenDoWeHaveZetaBootstrapConsistency}).

Assumption \ref{assump:ZetaBootstrapConsistency} along with Lemma \ref{lemma:IgnoreOp1TermsForBootCI}, proved in the next subsection, can be used to show that Algorithm \ref{alg:FullPercentileBootstrap} provides asymptotically valid confidence intervals. While the bootstrap consistency claim from Assumption \ref{assump:ZetaBootstrapConsistency} involves quantities generated in Algorithm \ref{alg:FullPercentileBootstrap}, $\hgmarg{,*}$ is drawn from a Gaussian distribution rather than from the actual bootstrap distribution in Algorithm \ref{alg:QuickConvolutionBootstrap}. In the next theorem, we prove a bootstrap consistency result for an approximate draw of the pivot $\sqrt{\ntot}(\htPPom-\theta)$ from the bootstrap distribution (which is generated in Algorithm \ref{alg:QuickConvolutionBootstrap} up to $o_p(1)$ terms). 

\begin{theorem}\label{theorem:ConsistencyOfGaussianConvBoot} For a fixed $\Omega$, let $R_{\ntot}=\sqrt{\ntot} (\htPPom -\theta)$ be a pivot and consider the approximate bootstrap pivot given by $ R_{\ntot}^*=\Omega L_{\gamma} Z + R_{\ntot,\Delta}^*$, where $$R_{\ntot,\Delta}^* = \sqrt{\ntot} \big( \htcarg{,*} -\Omega \hgcarg{,*} - (\htc -\Omega \hgc) \big),$$
where $L_{\gamma}$ is a matrix such that $L_{\gamma} L_{\gamma}^\tran = \SigGM$, and where $Z \sim \mathcal{N}(0,I_{d \times d})$ is independent of all data. Under Assumptions \ref{assump:SamplingLabelling}--\ref{assump:ZetaBootstrapConsistency} for each $j \in \{1,\dots,d\}$, $$\sup\limits_{x \in \mathbb{R}} \vert \mathbb{P}_*( [R_{\ntot}^*]_j \leq x \giv \hat{\mathbb{P}}_{\ntot})- \mathbb{P}([R_{\ntot}]_j \leq x) \vert \xrightarrow{p} 0.$$
\end{theorem}

\begin{proof}

Fix $\Omega \in \mathbb{R}^{d \times d}$ and $j \in \{1,\dots,d\}$. Define $u_j \in \mathbb{R}^d$ to be the $j$th row of the matrix $\Omega L_{\gamma}$ so that $[R_{\ntot}^*]_j=u_j^\tran Z+[R_{\ntot,\Delta}^*]_j$. Let $$\hat{H}_{\text{Boot}}(x) =\mathbb{P}_*( [R_{\ntot}^*]_j \leq x \giv \hat{\mathbb{P}}_{\ntot}) \quad \text{and} \quad \hat{H}_{\text{Boot},\Delta}(x) =\mathbb{P}_*( [R_{\ntot,\Delta}^*]_j \leq x \giv \hat{\mathbb{P}}_{\ntot})$$ be the Random CDFs of the bootstrap pivot (and pivot component). Letting $R_{\ntot,\Delta} \equiv \sqrt{\ntot} \big( \htc -\Omega \hgc - (\theta -\Omega \gamma) \big)$, also define $$H_{\ntot}(x) =\mathbb{P}( [R_{\ntot}]_j \leq x ) \quad \text{and} \quad H_{\ntot,\Delta}(x) =\mathbb{P}( [R_{\ntot,\Delta}]_j \leq x ).$$

   Let $\e_Z[\cdot]$ to be the expectation with respect to just the independent $Z$ (which conditions on $\hat{\mathbb{P}}_{\ntot}$), let $T(x,Z) \equiv \vert  \e_Z[  H_{\ntot,\Delta}(x-u_j^\tran Z)] - H_{\ntot}(x) \vert $ for all $x \in \mathbb{R}$, and let $\varphi : \mathbb{R}^d \to \mathbb{R}$ be the probability density function of a $\mathcal{N}(0,I_{d \times d})$ random vector. Then for each $x \in \mathbb{R}$, $$\begin{aligned}  \vert \hat{H}_{\text{Boot}}(x)-H_{\ntot}(x) \vert &
\leq  \vert \e_*[ I\{[R_{\ntot,\Delta}^*]_j \leq x-u_j^\tran Z \} \giv \hat{\mathbb{P}}_{\ntot}]-H_{\ntot}(x) \vert 
\\ & =   \vert  \e_Z [ \hat{H}_{\text{Boot},\Delta} (x-u_j^\tran Z)] -H_{\ntot}(x) \vert 
\\ & \leq   \vert   \e_Z [ \hat{H}_{\text{Boot},\Delta} (x-u_j^\tran Z)] - \e_Z[  H_{\ntot,\Delta}(x-u_j^\tran Z)] \vert + T(x,Z)
\\ & = \Big| \int_{z \in \mathbb{R}^d} \varphi(z) \big(\hat{H}_{\text{Boot},\Delta} (x-u_j^\tran z)-H_{\ntot,\Delta}(x-u_j^\tran z) \big) \rd z \Big| + T(x,Z)
\\ & \leq \int_{z \in \mathbb{R}^d} \varphi(z) \rho_{\infty} (\hat{H}_{\text{Boot},\Delta},H_{\ntot,\Delta}) \rd z + T(x,Z)
\\ & =  \rho_{\infty} (\hat{H}_{\text{Boot},\Delta},H_{\ntot,\Delta}) +T(x,Z).
\end{aligned}$$ Above the 2nd step holds by independence of $Z$ and Fubini's theorem, and recall that $\rho_{\infty}$ is defined at \eqref{eq:rhoSupNorm}. Hence taking the supremum over all $x$, $$\sup_{x \in \mathbb{R}} \vert \hat{H}_{\text{Boot}}(x)-H_{\ntot}(x) \vert \leq  \rho_{\infty} (\hat{H}_{\text{Boot},\Delta},H_{\ntot,\Delta}) + \sup_{x \in \mathbb{R}} T(x,Z)=o_p(1)+\sup_{x \in \mathbb{R}} T(x,Z).$$ Above the last step holds because letting $v_j \in \mathbb{R}^{3d}$ such that $v_j^\tran=e_j^\tran \begin{bmatrix}  I_{d \times d} & -\Omega & 0_{d \times d} \end{bmatrix}$, $$[R_{\ntot,\Delta}]_j=\sqrt{\ntot} v_j^\tran( \hat{\zeta}-\zeta) \quad \text{and} \quad [R_{\ntot,\Delta}^*]_j=\sqrt{\ntot} v_j^\tran ( \hat{\zeta}^*-\hat{\zeta}),$$ so by Assumption \ref{assump:ZetaBootstrapConsistency}, $\rho_{\infty}(\hat{H}_{\text{Boot},\Delta},H_{\ntot,\Delta}) \xrightarrow{p} 0$.

We now show that $\sup_{x \in \mathbb{R}} T(x,Z)=o(1)$. To do this note that since $Z$ is independent of the data $$H_{\ntot,\Delta}(x-u_j^\tran Z)=\e[ I \{[R_{\ntot,\Delta}]_j \leq x - u_j^\tran Z \} \giv Z ]=\e[  I\{ u_j^\tran Z + [R_{\ntot,\Delta}]_j  \leq x \} \giv Z ].$$ Also note that by applying Proposition \ref{prop:PPEstimatorCLT} to the case where $\hat{\Omega}=\Omega$ for all $\ntot$, $$R_{\ntot} \xrightarrow{d} \mathcal{N}(0,\Sigma_{\Omega}^{\circ}+\Sigma_{\Omega}^{\bullet}) \quad \text{where } \quad \Sigma_{\Omega}^{\circ} \equiv \Omega \SigGM \Omega^\tran \quad \text{and} \quad \Sigma_{\Omega}^{\bullet} \equiv \SigTC-\SigTGC \Omega^\tran -\Omega [ \SigTGC]^\tran +\Omega  \SigGC  \Omega^\tran.$$ Further, letting $B_{\Omega}=\begin{bmatrix} I_{d \times d} & -\Omega & 0 \end{bmatrix}$ and observing that $R_{\ntot,\Delta}=B_{\Omega} \big(\sqrt{\ntot} (\hat{\zeta}-\zeta) \big)$, by Lemma \ref{lemma:JointCLTzeta} and the continuous mapping theorem $R_{\ntot,\Delta} \xrightarrow{d} \mathcal{N}(0,B_{\Omega} \Sigma_{\zeta} B_{\Omega}^\tran)$, where by Formula \eqref{eq:zetaSigmaZetaDef}, $B_{\Omega} \Sigma_{\zeta} B_{\Omega}^\tran= \Sigma_{\Omega}^{\bullet}$. Since $Z$ is independent of $R_{\ntot,\Delta}$ for all $\ntot$ and $\Omega L_{\gamma} Z \sim \mathcal{N}(0,\Omega L_{\gamma} L_{\gamma}^\tran \Omega^\tran) =_{\text{dist}} \mathcal{N}(0,\Sigma_{\Omega}^{\circ})$, $ \Omega L_{\gamma} Z+R_{\ntot,\Delta} \xrightarrow{d} \mathcal{N}(0,\Sigma_{\Omega}^{\circ }+\Sigma_{\Omega}^{\bullet})$. Thus letting $\sigma_{j}^2 =[\Sigma_{\Omega}^{\circ}+\Sigma_{\Omega}^{ \bullet }]_{jj}$, it is clear that $[R_{\ntot}]_j \xrightarrow{d} \mathcal{N}(0,\sigma_j^2)$ and also $u_j^\tran Z+[R_{\ntot,\Delta}]_j \xrightarrow{d} \mathcal{N}(0,\sigma_j^2)$. Letting $H(\cdot)$ denote the CDF of $\mathcal{N}(0,\sigma_j^2)$, $\tilde{H}_{\ntot}$ denote the CDF of $u_j^\tran Z+[R_{\ntot,\Delta}]_j$ and recalling that $H_{\ntot}$ is the CDF of $[R_{\ntot}]_j$, by Polya's theorem (Theorem 11.2.9 in \cite{TSH}), $\tilde{H}_{\ntot}(x)$ and $H_{\ntot}(x)$ both converge to $H(x)$, uniformly in $x$. Combining these results,
$$\begin{aligned}
    \sup\limits_{x \in \mathbb{R}} T(x,Z) & = \sup\limits_{x \in \mathbb{R}} \vert \e_Z[  H_{\ntot,\Delta}(x-u_j^\tran Z)] - H_{\ntot}(x) \vert
    \\ & \leq \sup\limits_{x \in \mathbb{R}} \vert \e_Z[  H_{\ntot,\Delta}(x-u_j^\tran Z)] -H(x) \vert + \sup\limits_{x \in \mathbb{R}} \vert H(x) - H_{\ntot}(x) \vert
    \\ & = \sup\limits_{x \in \mathbb{R}} \vert \e_Z[  \e [I\{ u_j^\tran Z + [R_{\ntot,\Delta}]_j \leq x \} \giv Z ] ] -H(x) \vert + \sup\limits_{x \in \mathbb{R}} \vert H(x) - H_{\ntot}(x) \vert
    \\ & = \sup\limits_{x \in \mathbb{R}} \vert \mathbb{P} (u_j^\tran Z +[R_{\ntot,\Delta}]_j \leq x ) -H(x) \vert + \sup\limits_{x \in \mathbb{R}} \vert H(x) - H_{\ntot}(x) \vert
     \\ & = \sup\limits_{x \in \mathbb{R}} \vert \tilde{H}_{\ntot}(x) -H(x) \vert + \sup\limits_{x \in \mathbb{R}} \vert H(x) - H_{\ntot}(x) \vert
     \\ & = o(1),
\end{aligned}$$ where the last step holds from uniform convergence of $\tilde{H}_{\ntot}(\cdot)$ and $H_{\ntot}(\cdot)$ to $H(\cdot)$. Finally combining this with a previous result $$\sup\limits_{x \in \mathbb{R}} \vert \mathbb{P}_*( [R_{\ntot}^*]_j \leq x \giv \hat{\mathbb{P}}_{\ntot})- \mathbb{P}([R_{\ntot}]_j \leq x) \vert =\sup_{x \in \mathbb{R}} \vert \hat{H}_{\text{Boot}}(x)-H_{\ntot}(x) \vert \leq  o_p(1)+\sup_{x \in \mathbb{R}} T(x,Z)=o_p(1).$$
\end{proof}

Theorem \ref{theorem:ConsistencyOfGaussianConvBoot} along with Lemma \ref{lemma:IgnoreOp1TermsForBootCI}, proved in the next subsection, can be used to show that Algorithm \ref{alg:QuickConvolutionBootstrap} provides asymptotically valid confidence intervals.

\subsection{Lemmas for showing validity of percentile bootstrap}

In this section, we present two lemmas that, along with the bootstrap consistency results in the previous subsection, allow us to prove the validity of Algorithms \ref{alg:FullPercentileBootstrap} and \ref{alg:QuickConvolutionBootstrap}. The next lemma shows that Assumption \ref{assump:ZetaBootstrapConsistency} implies that $\sqrt{\ntot}(A \hat{\zeta}-A\zeta)$ and $\sqrt{\ntot}(A \hat{\zeta}^*-A\hat{\zeta})$, when viewed as sequences indexed by integers $\ntot >m_0$ for some $m_0$, are bounded in probability. Subsequently, Lemma \ref{lemma:IgnoreOp1TermsForBootCI} gives sufficient conditions under which the percentile bootstrap gives asymptotically valid confidence intervals.

\begin{lemma}\label{lemma:TightnessHelperLemma}
    Under Assumption \ref{assump:ZetaBootstrapConsistency}, for any fixed $A \in \mathbb{R}^{d \times 3d}$, $\sqrt{\ntot}(A \hat{\zeta}-A\zeta)=O_p(1)$ and $\sqrt{\ntot}(A \hat{\zeta}^*-A\hat{\zeta})=O_p(1)$, where $O_p(1)$ denotes a sequence of random vectors indexed by $\ntot \in \mathbb{N}$ that is bounded in probability for all $\ntot$ larger than some $m_0 \in \mathbb{N}$.
\end{lemma}

\begin{proof} First we will show that if $Y_{\ntot}=(Y_{\ntot}^{(1)},\dots,Y_{\ntot}^{(d)}) \in \mathbb{R}^d$ for $\ntot \in \mathbb{N}$ is a sequence of random vectors such that such that $Y_{\ntot}^{(j)}=O_p(1)$ for all $j \in \{1,\dots,d\}$, then $Y_{\ntot}=O_p(1)$. To see this, let $Y_{\ntot}=(Y_{\ntot}^{(1)},\dots,Y_{\ntot}^{(d)})$ for $\ntot \in \mathbb{N}$, where $Y_{\ntot}^{(j)}=O_p(1)$ for all $j \in \{1,\dots,d\}$. Fix $\epsilon>0$ and note that since $\epsilon/(2d)>0$ and since $Y_{\ntot}^{(j)}=O_p(1)$ for all $j \in \{1,\dots,d\}$, by definition of bounded in probability, there exists $m_1,\dots,m_d \in \mathbb{N}$, and $M_1,\dots,M_d \in \mathbb{R}$ such that $\sup_{\ntot > m_j } \mathbb{P}( \vert Y_{\ntot}^{(j)} \vert >M_j) < \epsilon/(2d)$ for all $j \in \{1,\dots,d\}$. Letting $m_1, \dots, m_d$ and $M_1, \dots, M_d$ be such numbers and defining $M=\sqrt{\sum_{j=1}^d M_j^2}$, and $m_0=\max_{j \in \{1,\dots,d\}} \{ m_j \}$ observe that for all $\ntot >m_0$, $$ \mathbb{P}( \vert \vert Y_{\ntot} \vert \vert_2 > M)=  \mathbb{P}\Big( \sum_{j=1}^d (Y_{\ntot}^{(j)})^2 > \sum_{j=1}^d M_j^2 \Big) \leq \mathbb{P} \Big( \bigcup\limits_{j=1}^d \{ \vert Y_{\ntot}^{(j)} \vert >M_j \} \Big) \leq \sum_{j=1}^d \mathbb{P}( \vert Y_{\ntot}^{(j)} \vert >M_j) < \frac{\epsilon}{2}.$$ Above the 2nd step follows from monotonicity of probability measure and the third step follows from the union bound and the final step holds because $\sup_{\ntot > m_j } \mathbb{P}( \vert Y_{\ntot}^{(j)} \vert >M_j) < \epsilon/(2d)$ for each $j \in \{1,\dots,d\}$. Since the above argument holds for all $\ntot \in \mathbb{N}$, we have found that $\sup_{\ntot > m_0} \mathbb{P}(  \vert \vert Y_{\ntot} \vert \vert_2 >M) \leq \epsilon/2 < \epsilon$. Hence we have shown that for any $\epsilon>0$ there exists an $m_0$ and $M$ such that $\sup_{\ntot > m_0} \mathbb{P}(  \vert \vert Y_{\ntot} \vert \vert_2 >M) < \epsilon$, implying by definition that $Y_{\ntot}=O_p(1)$.

Now fix $A \in \mathbb{R}^{d \times 3d}$ and let $v_j \in \mathbb{R}^{3d}$ denote the $j$th row of $A$ for $j \in \{1,\dots,d\}$. By Assumption \ref{assump:ZetaBootstrapConsistency}(ii), for each $j \in \{1,\dots,d\}$, $\big[ \sqrt{\ntot}( A \hat{\zeta}-A \zeta) \big]_j=\sqrt{\ntot} v_j^\tran (\hat{\zeta}-\zeta)$ converges in distribution. Thus $\big[ \sqrt{\ntot}( A \hat{\zeta}-A \zeta) \big]_j=O_p(1)$ for each $j \in \{1,\dots,d\}$, so by the result of the previous paragraph, $\sqrt{\ntot}( A \hat{\zeta}-A \zeta)=O_p(1)$.

Fix $j \in \{1,\dots,d\}$, and using $v_j$ defined in the previous paragraph define $R_{\ntot}^* \equiv \sqrt{\ntot} v_j^\tran (\hat{\zeta}^*-\hat{\zeta})$ and define $R_{\ntot} \equiv \sqrt{\ntot} v_j^\tran (\hat{\zeta}-\zeta)$. As argued in the previous paragraph $R_{\ntot}=O_p(1)$. To show that $R_{\ntot}^*=O_p(1)$ fix $\epsilon>0$. By Assumption \ref{assump:ZetaBootstrapConsistency}(i), $$\sup\limits_{x \in \mathbb{R}} \vert \mathbb{P}_*( R_{\ntot}^*  \leq x \giv \hat{\mathbb{P}}_{\ntot})- \mathbb{P}(R_{\ntot} \leq x) \vert \xrightarrow{p} 0.$$ Thus if we let $E_{\ntot,\epsilon/6}$ be the event that $$\sup\limits_{x \in \mathbb{R}} \vert \mathbb{P}_*( R_{\ntot}^*  \leq x \giv \hat{\mathbb{P}}_{\ntot})- \mathbb{P}(R_{\ntot} \leq x) \vert \leq \frac{\epsilon}{6},$$ there exists an $m_0 \in \mathbb{N}$ such that for all $\ntot >m_0$, $\mathbb{P}(E_{\ntot,\epsilon/6}) \geq 1-\epsilon/6$. It is easy to check that under the event $E_{\ntot,\epsilon/6}$, for any $x \in \mathbb{R}$, $$\mathbb{P}_*( \vert R_{\ntot}^* \vert  > x \giv \hat{\mathbb{P}}_{\ntot}) \leq \mathbb{P}( \vert R_{\ntot} \vert > x) + \frac{\epsilon}{3}.$$ Since $R_{\ntot}=O_p(1)$, there also exists an $M$ and $m_1 \geq m_0$ such that for all $\ntot > m_1$, $\mathbb{P}( \vert R_{\ntot} \vert > M ) < \epsilon/3$. Letting $m_1$ and $M$ be such numbers and observe that for all $\ntot >m_1$, $$\begin{aligned} \mathbb{P}( \vert R_{\ntot}^* \vert > M) & = \mathbb{P}( \{ \vert R_{\ntot}^* \vert > M \} \cap E_{\ntot,\epsilon/6})+ \mathbb{P}( \{ \vert R_{\ntot}^* \vert > M \} \cap E_{\ntot,\epsilon/6}^c) 
\\ & \leq \mathbb{P}( \{ \vert R_{\ntot}^* \vert > M \} \cap E_{\ntot,\epsilon/6}) + \frac{\epsilon}{6}
\\ & = \e[ \e_*[ I \{ \vert R_{\ntot}^* \vert > M \} I \{ E_{\ntot,\epsilon/6}\} \giv \hat{\mathbb{P}}_{\ntot}]] + \frac{\epsilon}{6}
\\ & = \e[ I \{ E_{\ntot,\epsilon/6}\} \e_*[ I \{ \vert R_{\ntot}^* \vert > M \}  \giv \hat{\mathbb{P}}_{\ntot}]] + \frac{\epsilon}{6}
\\ & = \e[ I \{ E_{\ntot,\epsilon/6}\} \mathbb{P}_*( \vert R_{\ntot}^* \vert  > M \giv \hat{\mathbb{P}}_{\ntot})] + \frac{\epsilon}{6}
\\ & \leq \e[ \mathbb{P}( \vert R_{\ntot} \vert > M) + \frac{\epsilon}{3}] + \frac{\epsilon}{6}
\\ & = \mathbb{P}( \vert R_{\ntot} \vert > M) + \frac{\epsilon}{2} < \frac{5 \epsilon}{6}.
\end{aligned}$$ Above the penultimate inequality holds because of the aforementioned upper bound on $\mathbb{P}_*( \vert R_{\ntot}^* \vert  > M \giv \hat{\mathbb{P}}_{\ntot})$ when the event $E_{\ntot,\epsilon/6}$ occurs. Taking the supremum of the above inequality over all $\ntot >m_1$, $\sup_{\ntot >m_1} \mathbb{P}( \vert R_{\ntot}^* \vert > M) < \epsilon$. Thus we have shown that for any fixed $\epsilon>0$, there exists numbers $m_1$ and $M$ such that $\sup_{\ntot >m_1} \mathbb{P}( \vert R_{\ntot}^* \vert > M) < \epsilon$, so by definition $R_{\ntot}^*=O_p(1)$. Recalling our definition of $R_{\ntot}^*$ it follows that for each $j \in \{1,\dots,d \}$, $\big[ \sqrt{\ntot}( A \hat{\zeta}^*-A \hat{\zeta}) \big]_j=\sqrt{\ntot} v_j^\tran (\hat{\zeta}^*-\hat{\zeta})=O_p(1)$. Since each component of the sequence of vectors $\sqrt{\ntot}( A \hat{\zeta}^*-A \hat{\zeta})$ is $O_p(1)$, by the result of the first paragraph $\sqrt{\ntot}( A \hat{\zeta}^*-A \hat{\zeta})=O_p(1)$. \end{proof}

In the following lemma, $r_{\ntot}$ is a rescaling constant that depends on $\ntot$ and satisfies $r_{\ntot} \to \infty$ as $\ntot \to \infty$. In many applications of interest $r_{\ntot} =\sqrt{\ntot}$; however, here we choose to state the lemma more generally.

\begin{lemma}\label{lemma:IgnoreOp1TermsForBootCI} Suppose $R_{\ntot}$ is a univariate pivot and that $R_{\ntot}^*$ is the bootstrapped (or approximate bootstrapped) version of $R_{\ntot}$ such that as $\ntot \to \infty$, $$\sup\limits_{x \in \mathbb{R}} \vert \mathbb{P}_*( R_{\ntot}^* \leq x \giv \hat{\mathbb{P}}_{\ntot})- \mathbb{P}(R_{\ntot} \leq x) \vert \xrightarrow{p} 0 \quad \text{and} \quad R_{\ntot} \xrightarrow{d} R_{\infty},$$ where $R_{\infty}$ has symmetric distribution with a continuous and strictly increasing CDF. Further suppose that $\hat{\eta}$ and $\hat{\eta}^*$ are an estimator and a bootstrap (or approximate bootstrap) draw of the estimator such that $r_{\ntot}(\hat{\eta}-\eta)=R_{\ntot}+o_p(1)$ and $r_{\ntot}(\hat{\eta}^*-\hat{\eta})=R_{\ntot}^*+o_p(1)$. Then, letting $\hat{\eta}^{(1)},\dots,  \hat{\eta}^{(B)}$ be IID draws from the bootstrap (or approximate bootstrap) distribution $\hat{\eta}^* \giv \hat{\mathbb{P}}_{\ntot}$, the empirical quantiles of this sequence provide asymptotically valid confidence intervals for $\eta$ in the sense that $$\lim\limits_{\ntot,B \to \infty} \mathbb{P} \Big( \eta \in \big( K_{\ntot,B}^{-1}(\alpha/2),K_{\ntot,B}^{-1}(1-\alpha/2)  \big) \Big) =1-\alpha \quad \text{where } K_{\ntot,B}(x)=\frac{1}{B} \sum_{b=1}^B I \{ \hat{\eta}^{(b)} \leq x \}.$$ 
\end{lemma}

\begin{proof}
Let $\hat{R}_{\ntot} \equiv r_{\ntot}(\hat{\eta}-\eta)$ and $\hat{R}_{\ntot}^* \equiv r_{\ntot}(\hat{\eta}^*-\hat{\eta})$ and note that by assumption $\hat{R}_{\ntot} =R_{\ntot}+o_p(1)$ and $ \hat{R}_{\ntot}^*=R_{\ntot}^*+o_p(1)$. Also let $D_{\ntot}^*=R_{\ntot}^*-\hat{R}_{\ntot}^*$ and note that $D_{\ntot}^*=o_p(1)$. Now define the following CDFs $$H_{\text{Boot}}(x)=  \mathbb{P}_*(R_{\ntot}^* \leq x \giv \hat{\mathbb{P}}_{\ntot}),  \  \hat{H}_{\text{Boot}}(x) = \mathbb{P}_*(\hat{R}_{\ntot}^* \leq x \giv \hat{\mathbb{P}}_{\ntot}), \quad \text{and} \quad H_{\ntot}(x) =\mathbb{P}(R_{\ntot} \leq x),$$ where the first two CDFs are random CDFs that depend on the empirical data distribution $\hat{\mathbb{P}}_{\ntot}$. Letting $\rho_{\infty}(\cdot,\cdot)$ denote the metric defined at \eqref{eq:rhoSupNorm}, by assumption in the lemma statement $\rho_{\infty}(H_{\text{Boot}},H_{\ntot}) \xrightarrow{p} 0$. Also by assumption $R_{\ntot} \xrightarrow{d} R_{\infty}$ where $R_{\ntot}$ has CDF $H_{\ntot}$ and where, by assumption, $R_{\infty}$ has a continuous and strictly increasing CDF as well as a symmetric distribution. Letting $H$ be the CDF of $R_{\infty}$, by Polya's Theorem (Theorem 11.2.9 in \cite{TSH}) $H_{\ntot}(x)$ converges to $H(x)$ uniformly in $x$ as $\ntot \to \infty$ and hence $\rho_{\infty}(H_{\ntot},H)=o(1)$.

Now we will show $\hat{H}_{\text{Boot}}(x) \xrightarrow{p} H(x)$ for all $x \in \mathbb{R}$. To do this fix $x \in \mathbb{R}$ and note $$\begin{aligned} \vert \hat{H}_{\text{Boot}}(x)-H(x) \vert & \leq  \vert\hat{H}_{\text{Boot}}(x)-H_{\text{Boot}}(x) \vert + \vert H_{\text{Boot}}(x)-H_{\ntot}(x) \vert + \vert H_{\ntot}(x)-H(x) \vert 
\\ & \leq \vert\hat{H}_{\text{Boot}}(x)-H_{\text{Boot}}(x) \vert + \rho_{\infty} (H_{\text{Boot}},H_{\ntot} ) + \rho_{\infty} (H_{\ntot},H )
\\ & = \vert \mathbb{P}_*(\hat{R}_{\ntot}^* \leq x \giv \hat{\mathbb{P}}_{\ntot}) -  \mathbb{P}_*(R_{\ntot}^* \leq x \giv \hat{\mathbb{P}}_{\ntot} ) \vert +o_p(1) +o(1)
\\ & = \vert \mathbb{P}_*(R_{\ntot} \leq x+ D_{\ntot}^* \giv \hat{\mathbb{P}}_{\ntot}) -  \mathbb{P}_*(R_{\ntot} \leq x \giv \hat{\mathbb{P}}_{\ntot} ) \vert +o_p(1)
\\ & =  \vert H_{\text{Boot}}(x+D_{\ntot}^*)- H_{\text{Boot}}(x) \vert +o_p(1)
\\ & \leq  \vert H_{\text{Boot}}(x+D_{\ntot}^*)- H_{\ntot}(x+D_{\ntot}^*) \vert +\vert H_{\ntot}(x+D_{\ntot}^*)- H_{\ntot}(x) \vert 
\\ & \quad +\vert H_{\ntot}(x)- H_{\text{Boot}}(x)  \vert+o_p(1)
\\ & \leq 2 \rho_{\infty}(H_{\text{Boot}},H_{\ntot}) +\vert H_{\ntot}(x+D_{\ntot}^*)- H_{\ntot}(x) \vert  +o_p(1)
\\ & = \vert H_{\ntot}(x+D_{\ntot}^*)- H_{\ntot}(x) \vert +o_p(1)
\\ & \leq \vert H_{\ntot}(x+D_{\ntot}^*)-H(x+D_{\ntot}^*) \vert +\vert H(x+D_{\ntot}^*)-H(x) \vert  +\vert H(x)-H_{\ntot}(x) \vert +o_p(1)
\\ & \leq 2 \rho_{\infty}(H_{\ntot},H) +\vert H(x+D_{\ntot}^*)-H(x) \vert +o_p(1)
\\ & = o_p(1).
\end{aligned}$$ Above the last step holds because $H$ is continuous and $D_{\ntot}^*=o_p(1)$ and because as mentioned earlier $\rho_{\infty}(H_{\ntot},H)=o(1)$. Thus we have shown that $\hat{H}_{\text{Boot}}(x) \xrightarrow{p} H(x)$ for all $x \in \mathbb{R}$. 

Now define $\hat{R}_{\ntot}^{(b)}=r_{\ntot}(\hat{\eta}^{(b)}-\hat{\eta})$ for each $b=1,\dots,B$. Since $\hat{\eta}^{(1)},\dots,  \hat{\eta}^{(B)}$ are IID draws from the bootstrap distribution $\hat{\eta}^* \giv \hat{\mathbb{P}}_{\ntot}$, it is clear that $\hat{R}_{\ntot}^{(1)},\dots, \hat{R}_{\ntot}^{(B)}$ are IID draws from the bootstrap distribution  $\hat{R}_{\ntot}^* \giv \hat{\mathbb{P}}_{\ntot}$, which has CDF $\hat{H}_{\text{Boot}}$. Thus defining $$\hat{H}_{\text{Boot}}^{(B)}(x) \equiv\frac{1}{B} \sum_{b=1}^B I \{ \hat{R}_{\ntot}^{(b)} \leq x \},$$ it is clear that by the strong law of large numbers that as $B \to \infty$, $\hat{H}_{\text{Boot}}^{(B)}(x) \xrightarrow{a.s} \hat{H}_{\text{Boot}}(x)$ for all $x \in \mathbb{R}$. Letting $o_B(1)$ denote a term that converges to zero almost surely as $B \to \infty$, combining this with the previous result we get that for all $x \in \mathbb{R}$, $$\vert \hat{H}_{\text{Boot}}^{(B)}(x) 
 -H(x) \vert \leq \vert \hat{H}_{\text{Boot}}^{(B)}(x) 
 -\hat{H}_{\text{Boot}}(x) \vert + \vert \hat{H}_{\text{Boot}}(x) - H(x) \vert =o_B(1)+o_p(1).$$ Thus $\hat{H}_{\text{Boot}}^{(B)}(x) \xrightarrow{p} H(x)$ as $\ntot,B \to \infty$ for all $x$. Since $H$ is continuous and strictly increasing, by Lemma 11.2.1 in \cite{TSH}, as $\ntot,B \to \infty$, $[\hat{H}_{\text{Boot}}^{(B)}]^{-1}(\alpha) \xrightarrow{p} H^{-1}(\alpha)$ for all $\alpha \in (0,1)$.

To complete the proof, recall that $K_{\ntot,B}(x)=\frac{1}{B} \sum_{b=1}^B I \{ \hat{\eta}^{(b)} \leq x \}$, $\hat{R}_{\ntot} \equiv r_{\ntot}(\hat{\eta}-\eta)$ and $\hat{R}_{\ntot}^{(b)} \equiv r_{\ntot}(\hat{\eta}^{(b)}-\hat{\eta})$. Thus if we let $o_{p(\ntot,B)}(1)$ denote a terms that converge to $0$ in probability as $\ntot,B \to \infty$, $$\begin{aligned} \lim \limits_{\ntot,B \to \infty} \mathbb{P}(\eta \geq K_{\ntot,B}^{-1}(1-\alpha/2)) & = \lim \limits_{\ntot,B \to \infty} \mathbb{P} \big( - \hat{R}_{\ntot} \geq r_{\ntot} ( K_{\ntot,B}^{-1}(1-\alpha/2))-\hat{\eta}) \big)
\\ & =  \lim \limits_{\ntot,B \to \infty}\mathbb{P}\big( -\hat{R}_{\ntot} \geq  [\hat{H}_{\text{Boot}}^{(B)}]^{-1}(1-\alpha/2) \big)
\\ & = \lim \limits_{\ntot,B \to \infty}\mathbb{P} \big( -\hat{R}_{\ntot}  \geq  H^{-1}(1-\alpha/2) +o_{p(\ntot,B)}(1) \big)
\\ & = \lim \limits_{\ntot,B \to \infty}\mathbb{P} \big( \hat{R}_{\ntot} +o_{p(\ntot,B)}(1) \leq  H^{-1}(\alpha/2) \big)
\\ & =  H(H^{-1}(\alpha/2))=\alpha/2.
\end{aligned}$$ Above the third step follows from the previous result and the penultimate step follows from the assumption that $H$ is the CDF of a symmetric random variable. The final step above holds because by Slutsky's lemma and since $\hat{R}_{\ntot} +o_{p(\ntot,B)}=R_{\ntot}+o_p(1)+o_{p(\ntot,B)}$ converges to a random variable whose CDF is $H$ as $\ntot,B \to \infty$. A similar argument shows that $\lim_{\ntot,B \to \infty} \mathbb{P}(\eta \leq K_{\ntot,B}^{-1}(\alpha/2))=\alpha/2$. Combining this with the previous result $$\begin{aligned} \lim\limits_{\ntot,B \to \infty} \mathbb{P}\Big( \eta \in \big( K_{\ntot,B}^{-1}(\alpha/2),K_{\ntot,B}^{-1}(1-\alpha/2)  \big) \Big) & =\lim\limits_{\ntot,B \to \infty} \Big( 1-\mathbb{P}(\eta \leq K_{\ntot,B}^{-1}(\alpha/2))-\mathbb{P}(\eta \geq K_{\ntot,B}^{-1}(1-\alpha/2)) \Big)
\\ & = 1-\alpha/2 -\alpha/2=1-\alpha. \end{aligned}$$ 

\end{proof}

\subsection{Proofs of Theorems \ref{theorem:FullPercentileBootstrapCIsValid} and \ref{theorem:GaussianConvBootCIsValid}}

\subsubsection{Proof of Theorem \ref{theorem:FullPercentileBootstrapCIsValid}}\label{sec:ProofOfFullBootPercentileCorollary}

    Fix $j \in \{1,\dots,d\}$. Let $v_j \in \mathbb{R}^{3d}$ be the vector such that $v_j^\tran =e_j^\tran \begin{bmatrix}
        I_{d \times d} & -\Omega & \Omega
    \end{bmatrix}$. Observe that $v_j^\tran \zeta =\theta_j$, $v_j^\tran \hat{\zeta} =\htPPom_j$ and $v_j^\tran \hat{\zeta}^* =\htPPomStar_j$, with the latter being a draw from the bootstrap distribution conditional on $\hat{\mathbb{P}}_{\ntot}$. Let $R_{\ntot}=\sqrt{\ntot} v_j^\tran(\hat{\zeta}-\zeta)=\sqrt{\ntot}(\htPPom_j-\theta_j)$ and $R_{\ntot}^*=\sqrt{\ntot} v_j^\tran(\hat{\zeta}^*-\hat{\zeta})=\sqrt{\ntot}(\htPPomStar_j-\htPPom_j)$. Thus by Assumption \ref{assump:ZetaBootstrapConsistency}, $$\sup\limits_{x \in \mathbb{R}} \vert \mathbb{P}_*( R_{\ntot}^* \leq x \giv \hat{\mathbb{P}}_{\ntot})- \mathbb{P}(R_{\ntot} \leq x) \vert \xrightarrow{p} 0,$$ and moreover, $R_{\ntot}$ converges in distribution to a symmetrically distributed random variable with a continuous, strictly increasing CDF.

    Now define $\hat{R}_{\ntot} \equiv \sqrt{\ntot}(\htPPhom_j-\theta_j)$ and $A= \begin{bmatrix} 0_{d \times d} & -I_{d \times d} & I_{d \times d} \end{bmatrix}$ and note that
    $$\begin{aligned} \hat{R}_{\ntot} & = \sqrt{\ntot} e_j^\tran (\htPPhom-\theta)
    \\ & = R_{\ntot} +   \sqrt{\ntot} e_j^\tran (\hat{\Omega}-\Omega) (\hgm-\hgc)
    \\ & = R_{\ntot} +    e_j^\tran (\hat{\Omega}-\Omega)  \big(\sqrt{\ntot} (A \hat{\zeta}-A \zeta) \big)
    \\ & = R_{\ntot} +e_j^\tran o_p(1) O_p(1)
    \\ & = R_{\ntot} +o_p(1).
    \end{aligned}$$ Above the penultimate step holds by Lemma \ref{lemma:TightnessHelperLemma} and because $\hat{\Omega}=\Omega+o_p(1)$. Also define $\hat{R}_{\ntot}^* \equiv \sqrt{\ntot}(\htPPhomStar_j-\htPPhom_j)$, and note that
    $$\begin{aligned} \hat{R}_{\ntot}^* & = \sqrt{\ntot} e_j^\tran (\htPPhomStar-\htPPhom)
    \\ & = R_{\ntot}^* +   \sqrt{\ntot} e_j^\tran (\hat{\Omega}-\Omega) \big((\hgmarg{,*}-\hgcarg{,*}) -(\hgm-\hgc) \big)
    \\ & = R_{\ntot}^* +    e_j^\tran (\hat{\Omega}-\Omega)  \big(\sqrt{\ntot} (A \hat{\zeta}^*-A \hat{\zeta}) \big)
    \\ & = R_{\ntot}^* +e_j^\tran o_p(1) O_p(1)
    \\ & = R_{\ntot}^* +o_p(1).
    \end{aligned}$$ Above the penultimate step holds by Lemma \ref{lemma:TightnessHelperLemma} and because $\hat{\Omega}=\Omega+o_p(1)$.

    Now let $\htPPhomArg{(1)},\dots, \htPPhomArg{(B)}$ be IID draws from the distribution of $\htPPhomStar \giv \hat{\mathbb{P}}_{\ntot}, \hat{\Omega}$ that are computed in Algorithm \ref{alg:FullPercentileBootstrap}. Further let  
     $K_{\ntot,B}(x)=\frac{1}{B} \sum_{b=1}^B I \{ \htPPhomB_j \leq x \}$ be the empirical CDF of the $j$th component of the $B$ draws of $\htPPhomStar$ from the bootstrap distribution. Since $\hat{R}_{\ntot}^*=\sqrt{\ntot}(\htPPhomStar_j-\htPPhom_j)=R_{\ntot}^*+o_p(1)$, since $\hat{R}_{\ntot}=R_{\ntot}+o_p(1)$ with $R_{\ntot}$ converging in distribution to a symmetrically distributed random variable with a continuous, strictly increasing CDF, and since as mentioned earlier, $$\sup\limits_{x \in \mathbb{R}} \vert \mathbb{P}_*( R_{\ntot}^* \leq x \giv \hat{\mathbb{P}}_{\ntot})- \mathbb{P}(R_{\ntot} \leq x) \vert \xrightarrow{p} 0,$$ the conditions of Lemma \ref{lemma:IgnoreOp1TermsForBootCI} are met (with $r_{\ntot}=\sqrt{\ntot}$ and $\hat{\eta}^*= \htPPhomStar_j$). In particular, by Lemma \ref{lemma:IgnoreOp1TermsForBootCI}, $$\lim\limits_{\ntot,B \to \infty} \mathbb{P} \Big( \theta_j \in \big( K_{\ntot,B}^{-1}(\alpha/2),K_{\ntot,B}^{-1}(1-\alpha/2)  \big) \Big) =1-\alpha.$$ Observing that Algorithm \ref{alg:FullPercentileBootstrap} returns the confidence interval $$\mathcal{C}_j^{1-\alpha}=\big( K_{\ntot,B}^{-1}(\alpha/2),K_{\ntot,B}^{-1}(1-\alpha/2)  \big),$$ completes the proof.

\subsubsection{Proof of Theorem \ref{theorem:GaussianConvBootCIsValid}}\label{sec:ProofOfGaussianConvCIValidity}

 Fix $j \in \{1,\dots,d\}$. Using the same notation as Theorem \ref{theorem:ConsistencyOfGaussianConvBoot} let $$R_{\ntot}=\sqrt{\ntot} (\htPPom -\theta) \quad \text{and} \quad R_{\ntot}^* =\Omega L_{\gamma} Z + \sqrt{\ntot} \big( \htcarg{,*} -\Omega \hgcarg{,*} - (\htc -\Omega \hgc) \big)$$ where $Z \sim \mathcal{N}(0,I_{d \times d})$ is independent of all data and $L_{\gamma} L_{\gamma}^\tran=\SigGM$ gives a lower triangular Cholesky decomposition of $\SigGM$. By Theorem \ref{theorem:ConsistencyOfGaussianConvBoot},  $$\sup\limits_{x \in \mathbb{R}} \vert \mathbb{P}_*( [R_{\ntot}^*]_j \leq x \giv \hat{\mathbb{P}}_{\ntot})- \mathbb{P}([R_{\ntot}]_j \leq x) \vert \xrightarrow{p} 0.$$ Applying Proposition \ref{prop:PPEstimatorCLT} with a constant tuning matrix, $R_{\ntot} \xrightarrow{d} \mathcal{N}\big(0,\SigPTD(\Omega) \big)$.
    
    Now let $\mathcal{S}$ denote the set of symmetric matrices in $\mathbb{R}^{d \times d}$, let $\mathcal{S}_{++}^d$ denote the set of symmetric positive definite matrices in $\mathbb{R}^{d \times d}$, let $\mathcal{L}_{++}^d$ denote the set of lower triangular matrices in $\mathbb{R}^{d \times d}$ with positive diagonal entries. Since matrices in $\mathcal{S}_{++}^d$ have a unique Cholesky decomposition, define $f_{\text{Chol}}: \mathcal{S}_{++}^d \to \mathcal{L}_{++}^d$ to be the function such that for all $A \in \mathcal{S}_{++}^d$, $f_{\text{Chol}}(A)$ is the unique matrix in $\mathcal{L}_{++}^d$ satisfying $A=f_{\text{Chol}}(A)[f_{\text{Chol}}(A)]^\tran$. By Lemma 12.1.6 in \cite{schatzman2002NumericalAnalysisTextbook},  
    $f_{\text{Chol}}$ is a continuous function. Further, since $\SigGM \in \mathcal{S}_{++}^d$ (by assumption $\SigGM \succ 0$, and clearly $\SigGM$ is symmetric) and $L_{\gamma} L_{\gamma}^\tran = \SigGM$, $L_{\gamma}=f_{\text{Chol}}(\SigGM)$ is the unique Cholesky decomposition of $\SigGM$. 

    Now we will show that when $\hat{L}_{\gamma}$ is a lower triangular matrix returned in Algorithm \ref{alg:QuickConvolutionBootstrap} satisfying $\hat{L}_{\gamma}\hat{L}_{\gamma}^\tran =\hSMatGM$, $\sqrt{\ntot} \hat{L}_{\gamma} \xrightarrow{p} L_{\gamma}$. To do this fix $\epsilon>0$ and observe that since $f_{\text{Chol}}(\cdot)$ is continuous at $\SigGM$, there exists a $\delta > 0$, such that for any $M \in \mathcal{S}_{++}^d$, $$\vert \vert f_{\text{Chol}}(M) - f_{\text{Chol}}(\SigGM) \vert \vert > \epsilon \Rightarrow \vert \vert M - \SigGM \vert \vert \geq \delta.  $$ Thus letting $\lambda_d: \mathcal{S} \to \mathbb{R}$ be a function that gives the minimum eigenvalue of the input matrix, $$\begin{aligned}
        \mathbb{P} \big( \vert \vert \sqrt{\ntot} \hat{L}_{\gamma}-L_{\gamma} \vert \vert > \epsilon \big) & = \mathbb{P} \big( \vert \vert \sqrt{\ntot} \hat{L}_{\gamma}-L_{\gamma} \vert \vert > \epsilon,  \hSMatGM \in \mathcal{S}_{++}^d \big) + \mathbb{P} \big( \vert \vert \sqrt{\ntot} \hat{L}_{\gamma}-L_{\gamma} \vert \vert > \epsilon,  \hSMatGM \notin \mathcal{S}_{++}^d \big)
        \\ & \leq \mathbb{P} \big( \vert \vert f_{\text{Chol}}(\ntot \hSMatGM)- f_{\text{Chol}}(\SigGM) \vert \vert > \epsilon,  \hSMatGM \in \mathcal{S}_{++}^d \big) + \mathbb{P}(  \hSMatGM \notin \mathcal{S}_{++}^d)
        \\ & \leq \mathbb{P} \big( \vert \vert \ntot \hSMatGM- \SigGM \vert \vert \geq \delta,  \hSMatGM \in \mathcal{S}_{++}^d \big) + \mathbb{P}\big(  \lambda_d(\ntot \hSMatGM) =0  \big)
        \\ & \leq \mathbb{P} \big( \vert \vert \ntot \hSMatGM- \SigGM \vert \vert \geq \delta\big) + \mathbb{P} \big( \vert \lambda_d(\ntot \hSMatGM)-\lambda_d(\SigGM) \vert > \lambda_d(\SigGM)/2 \big).
    \end{aligned}$$ Since by assumption, $\hSMatGM$ returned in Algorithm \ref{alg:QuickConvolutionBootstrap} satisfies $\ntot \hSMatGM \xrightarrow{p} \SigGM$, and since $\lambda_d(\cdot)$ is also a continuous function $\lambda_d(\ntot \hSMatGM) \xrightarrow{p} \lambda_d(\SigGM)$. By the definition of convergence in probability and since $\delta>0$ and $\lambda_d(\SigGM)/2 >0$, taking the limit as $\ntot \to \infty$ of each side of the above inequality implies that $\lim_{\ntot \to \infty} \mathbb{P} \big( \vert \vert \sqrt{\ntot} \hat{L}_{\gamma}-L_{\gamma} \vert \vert > \epsilon \big)=0$. Since this argument holds for any fixed $\epsilon>0$, $\sqrt{\ntot} \hat{L}_{\gamma} \xrightarrow{p} L_{\gamma}$.
    
    
    Now define $\htPPhomCB{*} \equiv \hat{\Omega} \hgm + \hat{\Omega} \hat{L}_{\gamma} Z + \htcarg{,*} -\hat{\Omega} \hgcarg{,*}$ to be the estimator of $\theta$ in one bootstrap iteration of Algorithm \ref{alg:QuickConvolutionBootstrap} (this is calculated for each bootstrap iteration in Line \ref{line:CalcThetaPTDBootConvApprox} of Algorithm \ref{alg:QuickConvolutionBootstrap}), $$\hat{R}_{\ntot} \equiv \sqrt{\ntot}(\htPPhom-\theta) \quad \text{and} \quad \hat{R}_{\ntot}^* \equiv \sqrt{\ntot} (\htPPhomCB{*} -\htPPhom ).$$ Now we will show that $\hat{R}_{\ntot}=R_{\ntot}+o_p(1)$ and $\hat{R}_{\ntot}^*=R_{\ntot}^*+o_p(1)$. To do this note that $$\hat{R}_{\ntot}  = \sqrt{\ntot}(\htPPhom-\theta)  = R_{\ntot} + (\hat{\Omega}-\Omega) \big( \sqrt{\ntot}(\hgm-\hgc) \big) =R_{\ntot} +o_p(1) O_p(1) =R_{\ntot} +o_p(1),$$ where the last step holds because by Lemma \ref{lemma:TightnessHelperLemma}, $\sqrt{\ntot}(\hgm-\hgc)=\sqrt{\ntot}(A_1 \hat{\zeta}-A_1 \zeta)=O_p(1)$ for $A_1= \begin{bmatrix} 0_{d \times d} & -I_{d \times d} & I_{d \times d} \end{bmatrix}$. Meanwhile, letting $A_2= \begin{bmatrix} 0_{d \times d} & I_{d \times d} & 0_{d \times d} \end{bmatrix}$,
    $$\begin{aligned} \hat{R}_{\ntot}^* & = \sqrt{\ntot} \big( (\hat{\Omega}\hgm + \hat{\Omega} \hat{L}_{\gamma} Z + \htcarg{,*} -\hat{\Omega} \hgcarg{,*} ) -(\hat{\Omega} \hgm+ \htc-\hat{\Omega} \hgc) \big)
    \\ & = \sqrt{\ntot} \hat{\Omega} \hat{L}_{\gamma} Z + \sqrt{\ntot} (\htcarg{,*}-\htc)  -\sqrt{\ntot} \hat{\Omega}( \hgcarg{,*} - \hgc) 
    \\ & = R_{\ntot}^* +(\sqrt{\ntot} \hat{\Omega} \hat{L}_{\gamma} - \Omega L_{\gamma})Z-\sqrt{\ntot} (\hat{\Omega}-\Omega)( \hgcarg{,*} - \hgc) 
    \\ & = R_{\ntot}^* +\big( (\Omega +o_p(1))(\sqrt{\ntot} \hat{L}_{\gamma}) -\Omega L_{\gamma} \big) O_p(1)-(\hat{\Omega}-\Omega) \big( \sqrt{\ntot}(A_2 \hat{\zeta}^* - A_2\hat{\zeta}) \big) 
    \\ & = R_{\ntot}^* +\big( (\Omega +o_p(1))(\sqrt{\ntot} \hat{L}_{\gamma}) -\Omega L_{\gamma} \big) O_p(1) - o_p(1) O_p(1) 
     \\ & = R_{\ntot}^* +\big( (\Omega +o_p(1))(L_{\gamma}+o_p(1)) -\Omega L_{\gamma} \big) O_p(1) - o_p(1)
     \\ & = R_{\ntot}^* +o_p(1).\end{aligned}$$ Above, the penultimate equality holds because $\sqrt{\ntot} \hat{L}_{\gamma} \xrightarrow{p} L_{\gamma}$ and the third last equality holds by Lemma \ref{lemma:TightnessHelperLemma}.
    
    Now let $\htPPhomCB{(1)},\dots, \htPPhomCB{(B)}$ be IID draws from the distribution of $\htPPhomCB{*} \giv \hat{\mathbb{P}}_{\ntot},\hat{\Omega}$ that are computed in Line \ref{line:CalcThetaPTDBootConvApprox} of Algorithm \ref{alg:QuickConvolutionBootstrap}. Further, let  
     $$K_{\ntot,B}(x)=\frac{1}{B} \sum_{b=1}^B I \{ \htPPhomCB{(b)}_j \leq x \}$$ be the empirical CDF of the $j$th component of the $B$ draws of $\htPPhomCB{*}$ from the approximate bootstrap distribution. Since $[\hat{R}_{\ntot}^*]_j=\sqrt{\ntot}(\htPPhomCB{*}_j-\htPPhom_j)=[R_{\ntot}^*]_j+o_p(1)$, since $[\hat{R}_{\ntot}]_j= \sqrt{\ntot}(\htPPhom_j-\theta_j)=[R_{\ntot}]_j+o_p(1)$ with $[R_{\ntot}]_j$ converging in distribution to a symmetrically distributed random variable with a continuous, strictly increasing CDF, and since as mentioned earlier, $$\sup\limits_{x \in \mathbb{R}} \vert \mathbb{P}_*( [R_{\ntot}^*]_j \leq x \giv \hat{\mathbb{P}}_{\ntot})- \mathbb{P}([R_{\ntot}]_j \leq x) \vert \xrightarrow{p} 0,$$ the conditions of Lemma \ref{lemma:IgnoreOp1TermsForBootCI} are met (with $r_{\ntot}=\sqrt{\ntot}$ and $\hat{\eta}^*= \htPPhomCB{*}_j$). In particular, by Lemma \ref{lemma:IgnoreOp1TermsForBootCI}, $$\lim\limits_{\ntot,B \to \infty} \mathbb{P} \Big( \theta_j \in \big( K_{\ntot,B}^{-1}(\alpha/2),K_{\ntot,B}^{-1}(1-\alpha/2)  \big) \Big) =1-\alpha.$$ Observing that Algorithm \ref{alg:QuickConvolutionBootstrap} returns the confidence interval $$\mathcal{C}_j^{1-\alpha}=\big( K_{\ntot,B}^{-1}(\alpha/2),K_{\ntot,B}^{-1}(1-\alpha/2)  \big),$$ completes the proof.
   
\section{Elaboration on when assumptions are met}

\subsection{Regularity conditions for Assumption \ref{assump:AsymptoticLinearity} to hold for M-estimators}\label{sec:RegularityMestimation}

 In this section, we provide sufficient conditions under which Assumption \ref{assump:AsymptoticLinearity} holds in M-estimation settings. Moreover we provide an explicit formula for $\Psi(\cdot)$ and $\tilde{\Psi}(\cdot)$ in such settings (see Proposition \ref{prop:MestimatorsAsymptoticallyLinear}).

Throughout this section, we suppose $l_{\vartheta}(\cdot)$ is a loss function parameterized by $\vartheta \in \Theta$, and suppose the goal is to estimate $\theta= \argmin_{\vartheta \in \Theta} \e[l_{\vartheta}(X)]$. For convenience we define $L(\vartheta) = \e[l_{\vartheta}(X)]$, $\tilde{L}(\vartheta) = \e[l_{\vartheta}(\tilde{X})]$ and $\gamma=\argmin_{\vartheta \in \Theta} \tilde{L}(\vartheta)$. We find sufficient conditions under which the estimators given by
\begin{equation}\label{eq:WeightedMestimatorDef}
    \begin{bmatrix}
        \htc \\ \hgc \\ \hgm
    \end{bmatrix} =
    \begin{bmatrix}
        \argmin_{\vartheta \in \Theta} \bigl\{ \frac{1}{\ntot} \sum_{i=1}^\ntot W_i l_{\vartheta}(X_i) \bigr\}
        \\  \argmin_{\vartheta \in \Theta} \bigl\{ \frac{1}{\ntot} \sum_{i=1}^\ntot W_i l_{\vartheta}(\tilde{X}_i) \bigr\}
        \\  \argmin_{\vartheta \in \Theta} \bigl\{ \frac{1}{\ntot} \sum_{i=1}^\ntot \bar{W}_i l_{\vartheta}(\tilde{X}_i) \bigr\}
    \end{bmatrix}
\end{equation} satisfy the weighted asymptotic linearity assumption (Assumption \ref{assump:AsymptoticLinearity}). 

Some sufficient conditions include the following, slightly stronger modification of the smoothness assumption from \cite{PPI++} \begin{assumption}[Smooth Enough Loss]\label{assump:SmoothEnoughForAsymptoticLineariaty} The loss function $l_{\vartheta}(\cdot)$ satisfies

\begin{enumerate}[(i)] 
    \item $\vartheta \mapsto l_{\vartheta}(X)$ is differentiable at  $\vartheta=\theta$ for almost every $X$ while $\vartheta \mapsto l_{\vartheta}(\tilde{X})$ is differentiable at $\vartheta=\theta$ and at $\vartheta=\gamma$ for almost every $\tilde{X}$.
    \item $\vartheta \mapsto l_{\vartheta}(X)$ is locally Lipschitz around $\vartheta=\theta$ and $\vartheta \mapsto l_{\vartheta}(\tilde{X})$ is locally Lipschitz around both $\vartheta=\theta$ and $\vartheta=\gamma$. In particular, there is a neighborhood of $\theta$ where $l_{\vartheta}(x)$ is $M(x)$-Lipschitz in $\vartheta$ and $l_{\vartheta}(\tilde{x})$ is $\tilde{M}_1(\tilde{x})$-Lipschitz in $\vartheta$, and there is a neighborhood of $\gamma$ where $l_{\vartheta}(\tilde{x})$ is $\tilde{M}_2(\tilde{x})$-Lipschitz in $\vartheta$, such that $\e[M(X)^2+\tilde{M}_1(\tilde{X})^2+\tilde{M}_2(\tilde{X})^2]< \infty$.
    \item The population losses given by $L(\vartheta)$ and 
     $\tilde{L}(\vartheta)$ both admit a 2nd-order Taylor expansions about $\theta$, $\tilde{L}(\vartheta)$ also admits a 2nd-order Taylor expansion about $\gamma$, and the Hessians $\nabla^2 L(\theta)$ and $\nabla^2 \tilde{L}(\gamma)$ are nonsingular.
\end{enumerate}
\end{assumption} We remark that above, the assumptions on the smoothness of $\vartheta \mapsto l_{\vartheta}(\tilde{X})$ in a neighborhood of $\vartheta=\theta$ are unnecessary for proving the following proposition whereas the assumptions on the smoothness of $\vartheta \mapsto l_{\vartheta}(\tilde{X})$ in a neighborhood of $\vartheta=\gamma$ are unnecessary for the method in \cite{PPI++}, which we study in Appendix \ref{sec:PPI++ComparisonSection}. Given the similarity of these smoothness assumptions, and because they both need to be invoked at least once in this paper, they are included in the same assumption statement for ease of exposition.

\begin{proposition}\label{prop:MestimatorsAsymptoticallyLinear}
        Under Assumptions \ref{assump:SamplingLabelling} and \ref{assump:SmoothEnoughForAsymptoticLineariaty}, if $\htc$, $\hgc$, and $\hgm$ are given by Equation \eqref{eq:WeightedMestimatorDef} and $\htc \xrightarrow{p} \theta$, $\hgc \xrightarrow{p} \gamma$, and $\hgm \xrightarrow{p} \gamma$, then Assumption \ref{assump:AsymptoticLinearity} holds with \begin{equation}\label{eq:PsiForMestimator}
\Psi(x) = -[\nabla^2 L(\theta)]^{-1} \dot{l}_{\theta}(x) \quad \text{and} \quad \tilde{\Psi}(\tilde{x}) = -[\nabla^2 \tilde{L}(\gamma)]^{-1} \dot{l}_{\gamma}(\tilde{x}).
    \end{equation}
\end{proposition}

\begin{proof}
    We show that $\htc= \theta + \frac{1}{\ntot} \sum_{i=1}^\ntot W_i \Psi(X_i) +o_p(1/\sqrt{\ntot})$, and for brevity omit the proofs that $\hgc= \gamma + \frac{1}{\ntot} \sum_{i=1}^\ntot W_i \Psi(\tilde{X}_i) +o_p(1/\sqrt{\ntot})$ and $\hgm= \gamma + \frac{1}{\ntot} \sum_{i=1}^\ntot \bar{W}_i \Psi(\tilde{X}_i) +o_p(1/\sqrt{\ntot})$, which will hold by nearly identical arguments. To do this, let $Y=(W,X)$ and $Y_i=(W_i,X_i)$ for each $i$, and let $m_{\vartheta}(Y)=-Wl_{\vartheta}(X)$. Observe $\theta=\argmax_{\vartheta \in \Theta}\e[ m_{\vartheta}(Y)]$ and $\htc = \argmax_{\vartheta \in \Theta} \frac{1}{\ntot} \sum_{i=1}^\ntot m_{\vartheta}(Y_i)$. We now check that the following conditions from Theorem 5.23 of \cite{VanderVaartTextbook} hold: \begin{enumerate}[(i)]
        \item For almost every $Y$, $\vartheta \mapsto m_{\vartheta}(Y)$ is differentiable at $\vartheta=\theta$.
        \item There exists a function $\dot{m}$ of $Y$ with $\e[(\dot{m}(Y))^2]< \infty$ , such that for every $\vartheta_1$ and $\vartheta_2$ in a neighborhood of $\theta$, $\vert m_{\vartheta_1}(Y) -m_{\vartheta_2}(Y) \vert \leq \dot{m}(Y) \vert \vert \vartheta_1- \vartheta_2 \vert \vert$
        \item The map $\vartheta \mapsto \e[m_{\vartheta}(Y)]$ admits a 2nd-order Taylor expansion at its maximum $\theta$ and has nonsingular 2nd derivative matrix given by $V_{\theta}$.
        \item  $ \frac{1}{\ntot} \sum_{i=1}^\ntot m_{\htc}(Y_i) \geq \argmax_{\vartheta \in \Theta} \frac{1}{\ntot} \sum_{i=1}^\ntot m_{\vartheta}(Y_i)-o_p(1/\ntot)$
    \end{enumerate}
    Above (i) is a direct result of Assumption \ref{assump:SmoothEnoughForAsymptoticLineariaty}(i). Letting $M$ be the function given in Assumption \ref{assump:SmoothEnoughForAsymptoticLineariaty}(ii), (ii) holds by letting $\dot{m}(Y)= W M(X)$ because of Assumption \ref{assump:SmoothEnoughForAsymptoticLineariaty}(ii) and because by Assumption \ref{assump:SamplingLabelling} (iv),  $W \leq 1/\pi(\tilde{X}) \leq 1/a < \infty$ almost surely, so $\e[(\dot{m}(Y))^2] \leq \e[ M(X)^2]/a^2 < \infty$. Since by Assumption \ref{assump:SamplingLabelling} and the tower property $\e[m_{\vartheta}(Y)]=-\e[l_{\vartheta}(X)]=-L(\vartheta)$, (iii) above follows directly from the definition of $\theta$ and Assumption \ref{assump:SmoothEnoughForAsymptoticLineariaty}(iii), with $V_{\theta}=-\nabla^2 L(\theta)$. Finally (iv) holds trivially because $\htc = \argmax_{\vartheta \in \Theta} \frac{1}{\ntot} \sum_{i=1}^\ntot m_{\vartheta}(Y_i)$. Thus since all conditions are met and since $\htc \xrightarrow{p} \theta$, by Theorem 5.23 in \cite{VanderVaartTextbook}, $$\sqrt{\ntot} (\htc -\theta) = [\nabla^2 L(\theta)]^{-1} \frac{1}{\sqrt{\ntot}} \sum_{i=1}^\ntot -W_i \dot{l}_{\theta}(X_i)+o_p(1).$$ Rearranging terms $\htc= \theta + \frac{1}{\ntot} \sum_{i=1}^\ntot W_i \Psi(X_i) +o_p(1/\sqrt{\ntot})$. Similar arguments give the desired asymptotic linear expansions for $\hgc$ and $\hgm$.
\end{proof}

The above proposition guarantees that for M-estimators, Assumption \ref{assump:AsymptoticLinearity} holds under standard regularity conditions, provided that $\htc \xrightarrow{p} \theta$, $\hgc \xrightarrow{p} \gamma$, and $\hgm \xrightarrow{p} \gamma$. Consistency of $\htc, \hgc$, and $\hgm$ often holds. For example, in combination with Assumptions \ref{assump:SamplingLabelling} and \ref{assump:SmoothEnoughForAsymptoticLineariaty}, the following gives a sufficient (but not necessary) condition under which $\htc \xrightarrow{p} \theta$, $\hgc \xrightarrow{p} \gamma$, and $\hgm \xrightarrow{p} \gamma$ holds.

\begin{assumption}[Sufficient conditions for consistency]\label{assump:SufficientConditionsForConsistency} $\theta$ is the unique minimizer of $L(\vartheta)$ and $\gamma$ is the unique minimizer of $\tilde{L}(\vartheta)$. Further, one of the two following conditions holds \begin{enumerate}[(i)]
    \item $\vartheta \mapsto l_{\vartheta}(X)$ and $\vartheta \mapsto l_{\vartheta}(\tilde{X})$ are convex for almost every $(X,\tilde{X})$, or
    \item the parameter space $\Theta \subset \mathbb{R}^d$ is a compact set and both $L(\vartheta)$ and $\tilde{L}(\vartheta)$ are continuous functions.
\end{enumerate}

\end{assumption}

We note that consistency holds under Assumptions \ref{assump:SamplingLabelling}, \ref{assump:SmoothEnoughForAsymptoticLineariaty}, and \ref{assump:SufficientConditionsForConsistency}(i) based on a direct application of Proposition 1 in \cite{PPI++} for the special case where $\lambda=0$. Meanwhile under Assumptions \ref{assump:SamplingLabelling} and \ref{assump:SufficientConditionsForConsistency}(ii), consistency holds based on directly applying the results of Problem 5.27 in \cite{VanderVaartTextbook} and Theorem 5.7 in \cite{VanderVaartTextbook}.


\subsection{Sufficient conditions for Assumption \ref{assump:ZetaBootstrapConsistency}}\label{sec:WhenDoWeHaveZetaBootstrapConsistency}

\subsubsection{Z-estimators}\label{sec:ZestimatorsBootstrapConsistencyConditions}

Consider the Z-estimation setting where $\psi_{\vartheta}: \mathbb{R}^p \to \mathbb{R}^d$ is a class of functions parameterized by $\vartheta$ in the parameter space $\Theta \subseteq \mathbb{R}^d$ and that our estimand is the unique $\vartheta \in \Theta$ solving $\e[\psi_{\vartheta}(X)]=0$, while $\gamma$ is the unique $\vartheta \in \Theta$ solving $\e[\psi_{\vartheta}(\tilde{X})]=0$. Further suppose that  $\htc$, $\hgc$ and $\hgm$ are (weighted) Z-estimators estimators solving the estimating equations $$\frac{1}{\ntot}  \sum_{i=1}^\ntot W_i \psi_{\vartheta}(X_i)=0, \frac{1}{\ntot} \sum_{i=1}^\ntot W_i \psi_{\vartheta}(\tilde{X}_i)=0, \ \text{and } \frac{1}{\ntot}  \sum_{i=1}^\ntot \bar{W}_i \psi_{\vartheta}(\tilde{X}_i)=0,$$ respectively. In such settings $\hat{\zeta}=(\htc,\hgc,\hgm)$ is also a Z-estimator giving a solution to $\frac{1}{\ntot} \sum_{i=1}^\ntot \psi_{(\vartheta_1,\vartheta_2,\vartheta_3)}^{\text{(stack)}} (V_i)=0$, where $\psi_{(\vartheta_1,\vartheta_2,\vartheta_3)}^{\text{(stack)}} (V_i) = \big(W_i \psi_{\vartheta_1}(X_i), W_i \psi_{\vartheta_2}(\tilde{X}_i), \bar{W}_i \psi_{\vartheta_3}(\tilde{X}_i)  \big)$. Let $z=(\vartheta_1,\vartheta_2,\vartheta_3)$ be shorthand notation and define $\bar{\psi}(z)=\e[\psi_{z}^{\text{(stack)}} (V_i)]$ and observe that under Assumption \ref{assump:SamplingLabelling}(ii)-(iv), $\bar{\psi}(\zeta)=\e[(\psi_{\theta}(X),\psi_{\gamma}(\tilde{X}),\psi_{\gamma}(\tilde{X}))]=0$. 

Since $\hat{\zeta}$ is a Z-estimator, Assumption \ref{assump:ZetaBootstrapConsistency} can be shown using existing theory on bootstrap consistency of Z-estimators. Below we restate a simplified version of the sufficient conditions from Theorem 10.16 of \cite{KosorokEmpiricalProcessTextbook} with notation convenient for our setting. The simplifications occur because we only present a setting where we presume $\hat{\zeta}$ and $\hat{\zeta}^*$ are exact (as opposed to approximate) solutions to the estimating equations $\frac{1}{\ntot}\sum_{i=1}^\ntot \psi_{z}^{\text{(stack)}} (V_i)=0$ and $\frac{1}{\ntot}\sum_{i=1}^\ntot \psi_{z}^{\text{(stack)}} (V_i^*)=0$.

\begin{assumption}[Z-estimation sufficient conditions for bootstrap consistency] \label{assump:ZestimatorSufficientConds} Suppose that

\begin{enumerate}[(i)]
       \item For any sequence $\{z_n\} \in \Theta^3$, $\lim\limits_{n \to \infty} \bar{\psi}(z_n) = 0$ implies $\lim\limits_{n \to \infty} z_n   = \zeta$.
       \item  The function class $\mathcal{F}=\{ \psi_{z}^{\text{(stack)}} \ : \ z \in \Theta^3 \}$ is strong $\mathbb{P}$-Glivekno-Cantelli.
       \item  For some $\eta>0$, the function class $\mathcal{F}_{\eta}=\{ \psi_{z}^{\text{(stack)}} \ : \ z \in \Theta^3, \vert \vert z - \zeta\vert \vert \leq \eta  \}$ is $\mathbb{P}$-Donsker.
       \item $\lim\limits_{z \to \zeta} \e \Big[ \big| \big| \psi_{z}^{\text{(stack)}}(V_i) -\psi_{\zeta}^{\text{(stack)}} (V_i) \big| \big|_2 ^2 \Big]=0$ and  $\e \Big[ \big| \big| \psi_{\zeta}^{\text{(stack)}}(V_i)  \big| \big|_2 ^2 \Big] <\infty$.
       \item The function from $\mathbb{R}^{3d} \to \mathbb{R}^{3d}$ given by $z \mapsto \bar{\psi}(z)$ is differentiable at $\zeta$ with nonsingular derivative matrix $D_{\bar{\psi}}$.
   \end{enumerate}
\end{assumption} 
Above, (i) is an indefinability assumption for the unique zero $\zeta$ of the function $z \mapsto \bar{\psi}(z)$. The second assumption is a standard assumption that would typically be used to establish that $\hat{\zeta} \xrightarrow{p} \zeta$. The third assumption is one that would typically be used to establish asymptotic normality of $\hat{\zeta}$, and that holds under fairly mild regularity conditions such as when $z \mapsto \psi_{z}^{\text{(stack)}}(V_i)$ is locally Lipschitz in a neighborhood of $\zeta$ (e.g., see example 19.7 in \cite{VanderVaartTextbook}). The fourth assumption is a smoothness regularity condition and a bounded second moment regularity condition. For the fifth assumption
note that the derivative of $z \mapsto \bar{\psi}(z)$ will always be block diagonal with blocks of size $d \times d$. Therefore, assumptions of nonsingularity of $D_{\bar{\psi}}$ reduce to standard assumptions used when studying the asymptotic behavior of $\htc$, $\hgc$, and $\hgm$ separately.

The following proposition is simply a direct application of Theorem 10.16 in \cite{KosorokEmpiricalProcessTextbook}. We omit the proof of this proposition for brevity, because it merely involves translating notation between the texts and simplifying their assumptions to the setting where $\hat{\zeta}$ and $\hat{\zeta}^*$ are exact zeros rather than approximate zeros to the estimating equation. The proof also involves noting that the conclusion of their theorem is convergence of $\sqrt{\ntot}(\hat{\zeta}-\zeta )$ and $\sqrt{\ntot}(\hat{\zeta}^*-\hat{\zeta})$ to multivariate Gaussians, which is stronger than Assumption \ref{assump:ZetaBootstrapConsistency}. 
\begin{proposition}
    
[\cite{KosorokEmpiricalProcessTextbook}] Under Assumptions \ref{assump:SamplingLabelling} and \ref{assump:ZestimatorSufficientConds},
   Assumption \ref{assump:ZetaBootstrapConsistency} will hold.
\end{proposition}

We note that Assumptions  \ref{assump:SamplingLabelling} and \ref{assump:ZestimatorSufficientConds} are not the most general assumptions under which Assumption \ref{assump:ZetaBootstrapConsistency} will hold for Z-estimators. To establish more general conditions under which Assumption \ref{assump:ZetaBootstrapConsistency} holds, including non-IID settings, we refer the reader to Chapter 13 of \cite{KosorokEmpiricalProcessTextbook}.

\subsubsection{When \texorpdfstring{$\htc$, $\hgc$, and $\hgm$ are}{} Hadamard differentiable functions}\label{sec:UsingHadamardDifferentiablityToShowConsistency}

Another way to show that Assumption \ref{assump:ZetaBootstrapConsistency} holds is by showing that each component estimator of $\hat{\zeta}$, namely $\htc$, $\hgc$, and $\hgm$, is a Hadamard differentiable function of an empirical distribution. In this section we make that condition more precise and prove that under such a condition (and other regularity conditions) Assumption \ref{assump:ZetaBootstrapConsistency} holds. For simplicity, we focus on cases where $\htc$, $\hgc$, and $\hgm$ are functions of weighted empirical cumulative distribution functions, although the approach can be generalized to settings where $\htc$, $\hgc$, and $\hgm$ are functions of distributions defined by empirical averages of each function in a function class. 

Define $\bar{\mathbb{R}}=\mathbb{R} \cup \{\infty\} \cup \{-\infty \}$ and let $\mathbb{D}$ be the space of bounded functions from $\bar{\mathbb{R}}^p \to \mathbb{R}$ equipped with the sup norm $\vert \vert \cdot \vert \vert_{\mathbb{D}}$ such that for $g \in \mathbb{D},$ $\vert \vert g \vert \vert_{\mathbb{D}} =\sup_{z \in \bar{\mathbb{R}}^{p}} \vert g(z) \vert$. Let  $\hat{F}_1,\hat{F}_2,\hat{F}_3 \in \mathbb{D}$ be the empirical weighted CDFs given by $$\hat{F}_1(a)=\frac{1}{\ntot} \sum_{i=1}^{\ntot} W_i I \{X_i \leq a \}, \quad \hat{F}_2(a)=\frac{1}{\ntot} \sum_{i=1}^{\ntot} W_i I \{\tilde{X}_i \leq a \}, \quad \hat{F}_3(a)=\frac{1}{\ntot} \sum_{i=1}^{\ntot} \bar{W}_i I \{\tilde{X}_i \leq a \}.$$ Note that throughout this section, including in the above formulas, vector inequalities are said to hold if and only if they hold pointwise. We similarly let $\hat{F}_1^*,\hat{F}_2^*,\hat{F}_3^* \in \mathbb{D}$ denote the weighted CDFs of the bootstrap draws which are given by $$\hat{F}^*_1(a)=\frac{1}{\ntot} \sum_{i=1}^{\ntot} W_i^* I \{X_i^* \leq a \}, \ \ \hat{F}^*_2(a)=\frac{1}{\ntot} \sum_{i=1}^{\ntot} W_i^* I \{\tilde{X}_i^* \leq a \}, \ \ \hat{F}_3^*(a)=\frac{1}{\ntot} \sum_{i=1}^{\ntot} \bar{W}_i^* I \{\tilde{X}_i^* \leq a \}.$$ It is also convenient to define $F_X,F_{\tilde{X}} \in \mathbb{D}$ to be the cumulative distribution functions of $X$ and $\tilde{X}$ given by $F_X(a)=\mathbb{P}(X \leq x)$ and $F_{\tilde{X}}(a)=\mathbb{P}(\tilde{X} \leq a)$.

The next assumption states that $\htc$, $\hgc$, $\hgm$, and their bootstrap counterparts can be written as a sufficiently smooth function of the above weighted empirical distributions. 
\begin{assumption}(Hadamard differentiable component estimators)\label{assump:HadamardDifferentiablityAssumption}
    There exists a space $\mathbb{D}_{\phi} \subseteq \mathbb{D}$ and function $\phi : \mathbb{D}_{\phi} \to \mathbb{R}^d$, 
     such that $\phi(F_X)=\theta$, $\phi(F_{\tilde{X}})=\gamma$, such that almost surely, $$\begin{bmatrix} \htc \\ \hgc \\ \hgm  \end{bmatrix} = \begin{bmatrix} \phi(\hat{F}_1) \\ \phi(\hat{F}_2) \\ \phi(\hat{F}_3)  \end{bmatrix} \quad \text{and} \quad \begin{bmatrix} \htcarg{,*} \\ \hgcarg{,*} \\ \hgmarg{,*} \end{bmatrix} = \begin{bmatrix} \phi(\hat{F}_1^*) \\ \phi(\hat{F}_2^*) \\ \phi(\hat{F}_3^*)  \end{bmatrix}.$$ Moreover $\phi$ is Hadamard differentiable at $F_X$ and $F_{\tilde{X}}$ tangentially to the subspace $\mathbb{C}_0 \subseteq \mathbb{D}$, where $\mathbb{C}_0$ denotes the set of continuous functions from $\bar{\mathbb{R}}^p \to \mathbb{R}$. We denote the Hadamard derivates of $\phi$ at $F_X$ and $F_{\tilde{X}}$ with $\phi_{F_X}'(\cdot)$ and $\phi_{F_{\tilde{X}}}'(\cdot)$.
\end{assumption} To make the above assumption more concrete, note that many estimators of interest are Hadamard differentiable functions of weighted empirical CDFs. For example, when $X$ and $\tilde{X}$ are univariate (i.e., $p=1$) and both have continuous CDFs, and when $\phi(\cdot)$ gives the $q$th empirical quantile (in the sense that for any nondecreasing $F \in \mathbb{D}$ we have $\lim \limits_{x \uparrow \phi(F)} F(x) \leq q \leq F(\phi(F))$), then $\phi$ will be Hadamard differentiabile at $F_X$ and $F_{\tilde{X}}$ with respect to the subspace $\mathbb{C}_0$. This claim about Hadamard differentiability of quantiles is a direct result of Lemma 21.3 from \cite{VanderVaartTextbook}. We refer readers to other chapters of \cite{VanderVaartTextbook} as well as \cite{VDVAndWellnerTextbook} and \cite{KosorokEmpiricalProcessTextbook} for Hadamard differentiability statements for other estimators such as trimmed means, among other estimators.

We now introduce an assumption that $X$ and $\tilde{X}$ have continuous CDFs (e.g., this will hold when $\mathbb{P}_X$ and $\mathbb{P}_{\tilde{X}}$ have no point masses and finite probability density everywhere). This assumption can be loosened when necessary (e.g., to settings where $F_X$ and $F_{\tilde{X}}$ are continuous at relevant points or neighborhoods), but the assumption stated below allows for easier exposition and a cleaner characterization of the space in which the limiting empirical processes of interest belong.

\begin{assumption}[$X$ and $\tilde{X}$ have continuous CDFs] \label{assump:XandTildeXDensity} $X$ and $\tilde{X}$ have cummulative distribution functions $F_X(\cdot)$ and $F_{\tilde{X}}(\cdot)$, which are continuous everywhere.
\end{assumption}

\begin{theorem}\label{theorem:HadamardDiffImpliesBootstrapConsitency}
    Under Assumptions \ref{assump:SamplingLabelling}, \ref{assump:HadamardDifferentiablityAssumption}, and \ref{assump:XandTildeXDensity}, Assumption \ref{assump:ZetaBootstrapConsistency} holds.
\end{theorem}

\begin{proof}

Throughout the proof, for any set $S$, we define $l^{\infty}(S)$ to be the space of functions $g: S \to \mathbb{R}$ that have finite sup norm. Note for example that $\mathbb{D}=l^{\infty}(\bar{\mathbb{R}}^p)$. The proof is lengthy, so we break it into sequences of steps as follows: 
\begin{enumerate}
    \setlength{\itemsep}{1pt}
    \setlength{\parskip}{1pt}
    \item defining a function class $\mathcal{F}$ and showing it is $\mathbb{P}$-Donsker,
    \item applying existing functional CLTs from \cite{VanderVaartTextbook} for empirical processes and empirical bootstrap processes,
    \item showing that the limiting distribution of such processes almost surely corresponds to a triplet of elements in $\mathbb{C}_0$ (using Assumption \ref{assump:XandTildeXDensity}),
    \item establishing Hadamard differentiability of relevant functionals (using Assumption \ref{assump:HadamardDifferentiablityAssumption}),
    \item and applying the functional delta method to establish Assumption \ref{assump:ZetaBootstrapConsistency}.
\end{enumerate}

 \paragraph{Defining function class $\mathcal{F}$ and showing it is $\mathbb{P}$-Donsker:} As in the main text, let $V_i=(W_i,\bar{W}_i,X_i,\tilde{X}_i)$ and $V_i^*=(W_i^*,\bar{W}_i^*,X_i^*,\tilde{X}_i^*)$ for each $i \in \{1,\dots, \ntot \}$. Now for each $a \in \bar{\mathbb{R}}^p$ define $f_{a,1}$, $f_{a,2}$,$f_{a,3} : \mathbb{R}^{2p+2} \to \mathbb{R}$ to be the functions given by $$f_{a,1}(V)=W I \{ X \leq a \}, \quad  f_{a,2}(V)=W I \{ \tilde{X} \leq a \}, \quad \text{and} \quad f_{a,3}(V)=\bar{W} I \{ \tilde{X} \leq a \}$$ where $V=(W,\bar{W},X,\tilde{X})$. Now for each $j \in \{1,2,3\}$ define $\mathcal{F}_j \equiv \{ f_{a,j} \ : \ a \in \bar{\mathbb{R}}^p \}$ and $\mathcal{F} \equiv \mathcal{F}_1 \cup \mathcal{F}_2 \cup \mathcal{F}_3$. 
 
 Now define $c=\max\{\vert \vert W \vert \vert_{\infty}, \vert \vert \bar{W} \vert \vert_{\infty} \}$. Note that by Assumption \ref{assump:SamplingLabelling}, $c< \infty$. Hence, $\mathcal{F}$ will have a finite envelope function that takes on the constant $c$. 
    
    Now fix $\epsilon>0$. Consider the case where $j=1$, and we will show that $\mathcal{F}_j$ can be covered by $(2c^2/\epsilon^2)^p$ $\epsilon$-brackets. To do this, note that for some $m< 2c^2/\epsilon^2$,  we can pick for each $k \in \{1,\dots,p\}$, $t_0^{(k)},t_1^{(k)},\dots,t_m^{(k)}$ such that $-\infty=t_0^{(k)} < t_1^{(k)}<\dots < t_m^{(k)}=\infty$ and such that for each $r \in \{1,\dots,m\}$, $\mathbb{P}( t_{r-1}^{(k)} < X^{(k)} \leq t_r^{(k)}, ) < \epsilon^2/c^2$. We can thus consider $m^d$ brackets indexed by $\bm{r}=(r_1,\dots,r_p) \in \{1,\dots,m\}^p$ given by $[f_{a_l(\bm{r}),j},f_{a_u(\bm{r}),j}]$, where $a_u(\bm{r})=(t_{r_1}^{(1)},\dots,t_{r_p}^{(p)})$ and $a_l(\bm{r})=(t_{r_1-1}^{(1)},\dots,t_{r_p-1}^{(p)})$. It is easy to check that such brackets cover $\mathcal{F}_j$. Moreover, for any $\bm{r} \in \{1,\dots,m\}^p$, $$\e [(f_{a_u(\bm{r}),j}-f_{a_l(\bm{r}),j})^2] \leq c^2 \mathbb{P} ( a_l(\bm{r}) <X  \leq a_u(\bm{r}) ) \leq c^2 \mathbb{P} ( t_{r_1-1}^{(1)} <X^{(1)} \leq t_{r_1}^{(1)} ) < \epsilon^2.$$ Thus for $j=1$, $\mathcal{F}_j$ can be covered by at most $(2c^2/\epsilon^2)^p$ $\epsilon$-brackets. For $j=2$ and $j=3$, an identical argument that replaces $X$ with $\tilde{X}$ shows that $\mathcal{F}_j$ can be covered by at most $(2c^2/\epsilon^2)^p$ $\epsilon$-brackets. Hence $\mathcal{F}=\mathcal{F}_1 \cup \mathcal{F}_2 \cup \mathcal{F}_3$ can be covered by at most $3 (2c^2/\epsilon^2)^p$ $\epsilon$-brackets.

    For any $\epsilon >0$ the $\epsilon$-bracketing number of $\mathcal{F}$, thus satisfies $N_{[]}(\epsilon, \mathcal{F},L_2(\mathbb{P}))< 3 (2c^2)^p \epsilon^{-2p}$. The bracketing integral therefore satisfies $$J_{[]}(1,\mathcal{F},L_2(\mathbb{P})) = \int_0^1 \sqrt{\log\big(N_{[]}(\epsilon, \mathcal{F},L_2(\mathbb{P})) \big)} \rd \epsilon < \int_0^1 \sqrt{\log(3(2c^2)^p) + 2p \log(1/\epsilon)} \rd \epsilon < \infty.$$ Hence by Theorem 19.5 in \cite{VanderVaartTextbook}, $\mathcal{F}$ is $\mathbb{P}$-Donsker.

 \paragraph{Applying function central limit theorems:}   Now for each $f \in \mathcal{F}$, define $$\mathbb{G}_{\ntot} f = \sqrt{\ntot}\big( \frac{1}{\ntot} \sum_{i=1}^\ntot f(V_i)-\e[f(V_i)]\big) \quad \text{and} \quad \mathbb{G}_{\ntot}^* f = \sqrt{\ntot}\big( \frac{1}{\ntot} \sum_{i=1}^\ntot f(V_i^*)-\frac{1}{\ntot}  \sum_{i=1}^\ntot f(V_i)\big),$$ which gives the empirical process and the bootstrap empirical processes indexed by $f \in \mathcal{F}$. Since $\mathcal{F}$ is $\mathbb{P}$-Donsker, the empirical processes $\{\mathbb{G}_{\ntot} f \ : \ f \in \mathcal{F} \}$ converges in distribution to a tight, mean $0$ Gaussian Process, call it $T$, taking values in $l^{\infty}(\mathcal{F})$. We choose to contruct $T$ to have continuous sample paths almost surely (in a sense defined in the next part of the proof). Moreover since $\mathcal{F}$ is $\mathbb{P}$-Donsker and has finite envelope function, if we view $\mathbb{G}_{\ntot}^*$ as a random element of $l^{\infty}(\mathcal{F})$, we can apply Theorem 23.7 in \cite{VanderVaartTextbook} to get that $\mathbb{G}_{\ntot}^*$ converges to $T$ in the following sense:
    \begin{equation}\label{eq:EmpiricalProcessBootstrapConvergence} \sup\limits_{h \in \text{BL}_1(l^{\infty}(\mathcal{F}))} \Big| \mathbb{E} \big[h \big( \mathbb{G}_{\ntot}^* \big) \giv \hat{\mathbb{P}}_{\ntot} \big] -\mathbb{E}[h(T)] \Big| \xrightarrow{p} 0. \end{equation}
Above $\text{BL}_1(l^{\infty}(\mathcal{F}))$ is the set of all $h: l^{\infty}(\mathcal{F}) \to [-1,1]$ that are uniformly Lipschitz.

\paragraph{Showing that $T$ almost surely corresponds to a triplet of elements of $\mathbb{C}_0$:} Before applying the functional delta method to the above convergences for $\mathbb{G}_{\ntot}$ and $\mathbb{G}_{\ntot}^*$ we must first characterize a subset of $l^{\infty}(\mathcal{F})$ to which $T$ is constructed to belong. To do this for each $j \in \{1,2,3\}$, define $\Phi_j: l^{\infty}(\mathcal{F}) \to \mathbb{D}$ such that for all $a \in \bar{\mathbb{R}}^p$ and $G \in l^{\infty}(\mathcal{F})$, $[\Phi_j(G)](a)=G(f_{a,j})$. Now let $\Phi : l^{\infty}(\mathcal{F}) \to \mathbb{D} \times \mathbb{D} \times \mathbb{D}$ be a one-to-one correspondence given by $\Phi(G)=\big(\Phi_1(G),\Phi_2(G),\Phi_3(G) \big)$ and define $(T_1,T_2,T_3) \equiv \Phi(T)$, so that for each $j \in \{1,2,3\}$ and $a \in \bar{\mathbb{R}}^p$, $T_j(a)=T(f_{a,j})$. We will now show that $T_1$ is almost surely a continuous function. To do this let $\Omega_T$ be the probability space on which $T$ is defined and let $T_1(\cdot)(\omega)$ and $T(\cdot)(\omega)$ be the specific realizations of $T_1$ and $T$ at a fixed $\omega$. Note that by Lemma 18.15 in \cite{VanderVaartTextbook}, $T$ can be constructed so that for almost every $\omega \in \Omega_T$, $f \mapsto T(f)(\omega)$ is uniformly continuous with respect to the semimetric given by $\rho(f,f')=\sqrt{\var(T(f)-T(f'))}$. We choose such a construction of $T$. Now fix $\omega \in \Omega_T$ for which $f \mapsto T(f)(\omega)$ is uniformly $\rho$-continuous (which can be done for almost every $\omega$). Next fix $\epsilon>0$ and $a_0 \in \bar{\mathbb{R}}^p$. Since $f \mapsto T(f)(\omega)$ is uniformly $\rho$-continuous there exists a $\delta_1$ such that whenever $\var(T(f)-T(f')) < \delta_1^2$, $\vert T(f)(\omega)-T(f')(\omega) \vert < \epsilon$. Fixing such a $\delta_1$, and noting that since by Assumption \ref{assump:XandTildeXDensity}, $F_X$ is continuous at $a_0$, there exists a $\delta_2$ such that for all $a \in \bar{\mathbb{R}}^p$ whenever $\vert \vert a-a_0 \vert \vert_2 < \delta_2$, $\vert F_X(a)-F_X(a_0) \vert < \delta_1^2/c^2$. With such a choice of $\delta_2$, observe that if $\vert \vert a-a_0 \vert \vert_2 < \delta_2$, $$\var(T(f_{a_0,1})-T(f_{a,1}))=\var\big(W (I\{X \leq a \} - I\{X \leq a_0 \} ) \big) \leq c^2 \vert F_X(a)-F_{X}(a_0) \vert < \delta_1^2,$$ where the first step holds because $\big( T(f_{a_0,1}), T(f_{a,1}) \big)$ has the same distribution as the limiting distribution of $$\sqrt{\ntot} \Big( \frac{1}{\ntot} \sum_{i=1}^{\ntot} \begin{bmatrix} W_i I \{X_i \leq a_0 \} \\ W_i I \{X_i \leq a \} \end{bmatrix} - \begin{bmatrix}
    F_X(a_0) \\ F_X(a)
\end{bmatrix}\Big).$$ It follows by an earlier result that when $\vert \vert a-a_0 \vert \vert_2 < \delta_2$, $\vert T_1(a_0)(\omega)-T_1(a)(\omega)\vert=\vert T(f_{a_0,1})(\omega)-T(f_{a,1})(\omega) \vert < \epsilon$. Since this argument holds for any fixed $a_0$ and $\epsilon>0$, $T_1(\cdot)(\omega)$ is continuous. Further the argument holds for almost every $\omega$, so $T_1$ is almost surely a continuous function on $\bar{\mathbb{R}}^p$ (i.e., $\mathbb{P}(T_1 \in \mathbb{C}_0)=1$). Nearly identical arguments involving $\tilde{X}$ instead of $X$ (and $W$ instead of $\bar{W}$) show that $T_2$ (and $T_3$) are almost surely in $\mathbb{C}_0$. Thus $T$ is a tight, mean zero Gaussian Process such that the natural correspondence map $\Phi(T)$ is almost surely a triplet of 3 functions in $\mathbb{C}_0$.

\paragraph{Establishing Hadamard differentiability of relevant functionals:} Now fix $v \in \mathbb{R}^{3d}$ and let $v_1,v_2,v_3 \in \mathbb{R}^d$ be its components such that $v=(v_1,v_2,v_3)$. Let $\phi: \mathbb{D}_{\phi} \to \mathbb{R}^d$ be the function satisfying Assumption \ref{assump:HadamardDifferentiablityAssumption}. Further let $\phi_v : \mathbb{D}_{\phi}^3   \to \mathbb{R}$ be the function given by $$\phi_v(G_1,G_2,G_3)=v_1^\tran \phi  (G_1) +v_2^\tran \phi (  G_2 )+v_3^\tran \phi (G_3),$$ and we will develop a Hadamard differentiability of $\phi_v \circ \Phi : l^{\infty}(\mathcal{F}) \to \mathbb{R}$ result. To do this let $\bar{G} \in l^{\infty}(\mathcal{F})$ be the function satisfying $\bar{G}(f)=\e[f(V_i)]$ for each $f \in \mathcal{F}$. Now note that $\Phi : l^{\infty}(\mathcal{F}) \to \mathbb{D}^3$ is Hadamard differentiable at $\bar{G}$, with respect to any subspace of $l^{\infty}(\mathcal{F})$. This is because $\Phi (\cdot)$ is a continuous, linear map so for any sequence $(h_t)_{t >0 }$ of elements of $l^{\infty}(\mathcal{F})$ such that $h_t \to h$ as $t \downarrow 0$, $$\lim\limits_{t \downarrow 0} \frac{\Phi(\bar{G} +t h_t) -\Phi(\bar{G})}{t}=\lim\limits_{t \downarrow 0} \Phi(h_t)=\Phi(h).$$ Further the Hadamard derivative of $\Phi$ at $\bar{G}$ is given by $\Phi_{\bar{G}}'(h)=\Phi(h)$ for any $h \in l^{\infty}(\mathcal{F})$. Also observe that by Assumption \ref{assump:SamplingLabelling} and definition of $\Phi_j$ for $j \in \{1,2,3\}$, $\Phi_1(\bar{G})=F_X$ and $\Phi_2(\bar{G})=\Phi_3(\bar{G})=F_{\tilde{X}}$ and for convenience define $\bar{F}=(F_X,F_{\tilde{X}},F_{\tilde{X}})$. Now let $(h_t)_{t > 0} =\big( h_{1t},h_{2t}, h_{3t} \big)_{t >0}$ be any sequence in $\mathbb{D}^3$ that converges to some $h= \big(h_1,h_2,h_3 \big) \in \mathbb{C}_0 \times \mathbb{C}_0 \times \mathbb{C}_0$ as $t \downarrow 0$, and observe that $$\begin{aligned} \lim\limits_{t \downarrow 0} \frac{\phi_v(\bar{F}+th_t)-\phi_v(\bar{F})}{t}  & = v_1^\tran \lim\limits_{t \downarrow 0} \frac{ \phi(F_X+th_{1t})-\phi(F_X)}{t}
    \\ & \quad  +v_2^\tran \lim\limits_{t \downarrow 0} \frac{\phi(F_{\tilde{X}}+th_{2t})-\phi(F_{\tilde{X}})}{t} +v_3^\tran \lim\limits_{t \downarrow 0} \frac{\phi(F_{\tilde{X}}+th_{3t})-\phi(F_{\tilde{X}})}{t}  
    \\ & = v_1^\tran \phi_{F_X}'(h_1)+v_2^\tran \phi_{F_{\tilde{X}}}'(h_2)+v_3^\tran \phi_{F_{\tilde{X}}}'(h_3), \end{aligned}$$ where the last step holds by Assumption \ref{assump:HadamardDifferentiablityAssumption}. Thus for $h=(h_1,h_2,h_3)$, letting $\phi_v' (h)= v_1^\tran \phi_{F_X}'(h_1)+v_2^\tran \phi_{F_{\tilde{X}}}'(h_2)+v_3^\tran \phi_{F_{\tilde{X}}}'(h_3)$, the above result shows that $\phi_v$ is Hadamard differentiable at $\bar{F}=(F_X,F_{\tilde{X}},F_{\tilde{X}})=\Phi(\bar{G})$ tangentially to the subspace $\mathbb{C}_0^3$ with Hadamard derivative $\phi_v'$. Now define $\Phi_{(v)}=\phi_v \circ \Phi : l^{\infty}(\mathcal{F}) \to \mathbb{R}$. By the chain rule for Hadamard differentiability (e.g., Theorem 20.9 in \cite{VanderVaartTextbook}), $\Phi_{(v)}$ is Hadamard differentiable at $\bar{G}$ tangentially to the subspace $\{h \in l^{\infty}(\mathcal{F}) : \Phi_{\bar{G}}'(h) \in \mathbb{C}_0^3 \}=\{h \in l^{\infty}(\mathcal{F}) : \Phi(h) \in \mathbb{C}_0^3 \}$ with Hadamard derivative at $\bar{G}$ denoted by $\Phi_{(v)}'$.

    \paragraph{Applying the functional delta method to establish Assumption \ref{assump:ZetaBootstrapConsistency}:} Now we can apply the functional delta method to obtain the desired result. In particular, for $f \in \mathcal{F}$ define $\mathbb{P}_{\ntot} f=\frac{1}{\ntot} \sum_{i=1}^{\ntot} f(V_i)$, $\mathbb{P}_{\ntot}^* f=\frac{1}{\ntot} \sum_{i=1}^{\ntot} f(V_i^*)$ and recall that $\bar{G}(f)=\e[f(V_i)]$. $\mathbb{P}_{\ntot}$, $\mathbb{P}_{\ntot}^*$, $\mathbb{G}_{\ntot}$, and $\mathbb{G}_{\ntot}^*$ can all be viewed as random elements of $l^{\infty}(\mathcal{F})$ where $\mathbb{G}_{\ntot}=\sqrt{\ntot}(\mathbb{P}_{\ntot}-\bar{G})$ and $\mathbb{G}_{\ntot}^*=\sqrt{\ntot}(\mathbb{P}_{\ntot}^*-\mathbb{P}_{\ntot})$. Now recall that the empirical process $\{\mathbb{G}_{\ntot} f \ : \ f \in \mathcal{F} \}$ converges to a tight mean $0$ Gaussian process $T$, which we showed earlier almost surely satisfies $\Phi(T) \in \mathbb{C}_0^3$. Since $\Phi_{(v)}$ is Hadamard differentiable at $\bar{G}$ tangentially to the subspace $\{h \in l^{\infty}(\mathcal{F}) : \Phi(h) \in \mathbb{C}_0^3 \}$, and $T$ is almost surely in this subspace, by the functional delta method (e.g., Theorem 20.8 in \cite{VanderVaartTextbook}) $$\sqrt{\ntot} \big( \Phi_{(v)}(\mathbb{P}_{\ntot})-\Phi_{(v)}(\bar{G}) \big) \xrightarrow{d} \Phi_{(v)}'(T).$$ Recalling that the empirical process $\mathbb{G}_{\ntot}^*=\sqrt{\ntot}(\mathbb{P}_{\ntot}^*-\mathbb{P}_{\ntot})$ satisfies the convergence in Equation \eqref{eq:EmpiricalProcessBootstrapConvergence} and $T$ is almost surely in the subspace $\{h \in l^{\infty}(\mathcal{F}) : \Phi(h) \in \mathbb{C}_0^3 \}$, we can apply the functional delta method for the bootstrap (Theorem 23.9 in \cite{VanderVaartTextbook}) to get $$\mathbb{P}_* \big[ \sqrt{\ntot}\big( \Phi_{(v)}(\mathbb{P}_{\ntot}^*)- \Phi_{(v)}(\mathbb{P}_{\ntot}) \big) \leq x \giv \hat{\mathbb{P}}_{\ntot} \big] \xrightarrow{p} \mathbb{P}(\Phi_{(v)}'(T) \leq x ) \quad \text{for all} \quad x \in \mathbb{R}.$$ 
    
    To complete the proof, we simplify the above convergence results. Note that since $\Phi_{(v)}' : l^{\infty}(\mathcal{F}) \to \mathbb{R}$ is a continuous linear map and since $T$ is a tight Gaussian process taking values in $l^{\infty}(\mathcal{F})$, by Lemma 3.10.8 in \cite{VDVAndWellnerTextbook}, $\Phi_{(v)}'(T)$ is normally distributed. Since $T$ is mean zero and $\Phi_{(v)}'$ is linear, $\e[\Phi_{(v)}'(T)]=0$ (see Page 523 in \cite{VDVAndWellnerTextbook} for a more formal justification). Thus $\Phi_{(v)}'(T) \stackrel{\text{dist}}{=} \mathcal{N}(0,\sigma_v^2)$ for some $\sigma_v^2 \geq 0$. Also observe that since $\Phi(\bar{G})=(F_X,F_{\tilde{X}},F_{\tilde{X}})$, $\Phi(\mathbb{P}_{\ntot})=(\hat{F}_1,\hat{F}_2,\hat{F}_3)$, and $\Phi(\mathbb{P}_{\ntot}^*)=(\hat{F}_1^*,\hat{F}_2^*,\hat{F}_3^*)$, and since $\Phi_{(v)}=\phi_v \circ \Phi$, we can apply the formulas from Assumption \ref{assump:HadamardDifferentiablityAssumption} and the definition of $\phi_v$, $\zeta$, $\hat{\zeta}$, and $\hat{\zeta}^*$ to get that $\Phi_{(v)}(\bar{G})=v^\tran \zeta$, $\Phi_{(v)}(\mathbb{P}_{\ntot})=v^\tran \hat{\zeta}$, and $\Phi_{(v)}(\mathbb{P}_{\ntot}^*)=v^\tran \hat{\zeta}^*$. Hence combining these simplifications with earlier convergence results, $$\sqrt{\ntot} v^\tran (\hat{\zeta}-\zeta)  \xrightarrow{d} \mathcal{N}(0,\sigma_v^2) \quad \text{and} \quad \mathbb{P}_* \big( \sqrt{\ntot} v^\tran  (\hat{\zeta}^*-\hat{\zeta} ) \leq x \giv \hat{\mathbb{P}}_{\ntot} \big) \xrightarrow{p} \mathbb{P} \big( \mathcal{N}(0,\sigma_v^2) \leq x \big) \ \ \ \forall_{x \in \mathbb{R}}.$$ Next let $H(\cdot)$ be the CDF of $\mathcal{N}(0,\sigma_v^2)$, and $$H_{\ntot}(x) \equiv \mathbb{P}(\sqrt{\ntot} v^\tran (\hat{\zeta}-\zeta) \leq x) \quad \text{and} \quad \hat{H}_{\ntot,\text{Boot}}(x) =\mathbb{P}_* \big( \sqrt{\ntot} v^\tran  (\hat{\zeta}^*-\hat{\zeta} ) \leq x \giv \hat{\mathbb{P}}_{\ntot} \big).$$ Note that by Polya's Theorem (Theorem 11.2.9 in \cite{TSH}), since $\sqrt{\ntot} v^\tran (\hat{\zeta}-\zeta)  \xrightarrow{d} \mathcal{N}(0,\sigma_v^2)$ with $\mathcal{N}(0,\sigma_v^2)$ having a continuous CDF $H$, $H_{\ntot}(x)$ converges to $H(x)$ uniformly in $x$. Further since $\hat{H}_{\ntot,\text{Boot}}(x) \xrightarrow{p} H(x)$ for all $x \in \mathbb{R}$ where $H$ is continuous, as stated on page 339 in \cite{VanderVaartTextbook}, $\sup_{x \in \mathbb{R}} \vert \hat{H}_{\ntot,\text{Boot}}(x)-H(x) \vert=o_p(1)$. To complete the proof observe that $$\begin{aligned}
     \sup\limits_{x \in \mathbb{R}} \vert \hat{H}_{\ntot,\text{Boot}}(x)-H_{\ntot}(x) \vert \leq \sup\limits_{x \in \mathbb{R}} \vert \hat{H}_{\ntot,\text{Boot}}(x)-H(x) \vert +\sup\limits_{x \in \mathbb{R}} \vert H(x)-H_{\ntot}(x) \vert 
     = o_p(1)+o(1).
\end{aligned}$$ Thus we have shown $\sqrt{\ntot} v^\tran (\hat{\zeta}-\zeta)  \xrightarrow{d} \mathcal{N}(0,\sigma_v^2)$ and
$$\sup_{x \in \mathbb{R}} | \mathbb{P}_* ( \sqrt{\ntot} v^\tran  (\hat{\zeta}^*-\hat{\zeta} ) \leq x \giv \hat{\mathbb{P}}_{\ntot} ) -\mathbb{P}(\sqrt{\ntot} v^\tran (\hat{\zeta}-\zeta) \leq x ) |  \xrightarrow{p} 0.$$ Since the above argument holds for any fixed $v \in \mathbb{R}^{3d}$, Assumption \ref{assump:ZetaBootstrapConsistency} holds.  
\end{proof}

\section{Efficiency comparisons with alternative estimators}\label{sec:AsymptoticVarianceComparisonsAppendix}

In this appendix, we prove the claims in Section \ref{sec:EfficiencyCalculations} about how the asymptotic variance of $\htPPhomOpt$ relates to alternatives. For ease of comparison, we only consider settings where the complete sample is a uniform random sample (i.e., where the labelling weights always satisfy $\pi(\tilde{X})=\pi_L$ for some constant $\pi_L \in (0,1)$). In Section \ref{sec:ComparisonWithClassicalEstimator}, we derive a formula for the ratio of the asymptotic variance of $\htPPhomOpt$ with that of $\htc$ for the case where $d=1$. In Section \ref{sec:PPI++ComparisonSection}, we formally introduce the PPI++ estimator from \cite{PPI++}, and derive a formula for the ratio of the asymptotic variance of $\htPPhomOpt$ with that of the PPI++ estimator. In that section we also show that there are some settings where the PPI++ estimator is less efficient than $\htPPhomOpt$ and other settings where the opposite is true (see Remark \ref{remark:PPI++VersusPTDQuantile} and Equation \eqref{eq:PPI++VersusTPPARE}). Finally, in Section \ref{sec:PTDusingGammaHatAll}, we show that for any $d \geq 1$, the optimally tuned PTD estimator has the same asymptotic variance as an optimally tuned variant of the PTD estimator which uses all data (as opposed to just the incomplete data) to estimate $\gamma$.

\subsection{Comparison with classical estimator}\label{sec:ComparisonWithClassicalEstimator}

\begin{proposition}\label{prop:AREUnivariateTPPvsClassical}
    Suppose $d=1$ and each point has an equal probability $\pi_L$ of being labelled (i.e., $\mathbb{P}(I=1 \giv \tilde{X})=\pi_L$ for all $\tilde{X}$). Under Assumptions \ref{assump:SamplingLabelling} and \ref{assump:AsymptoticLinearity}, if $\hat{\Omega}_{\textnormal{opt}} \xrightarrow{p} \Omega_{\textnormal{opt}}$, then $\sqrt{\ntot}(\htPPhomOpt-\theta) \xrightarrow{d} \mathcal{N}(0,\sigma_{\textnormal{TPTD}}^2)$ and $\sqrt{\ntot}(\htc-\theta) \xrightarrow{d} \mathcal{N}(0,\sigma_{\textnormal{classical}}^2)$, where $$\frac{\sigma_{\textnormal{TPTD}}^2}{\sigma_{\textnormal{classical}}^2}=1-(1-\pi_L) \cdot \corr^2(\Psi(X),\tilde{\Psi}(\tilde{X})).$$
\end{proposition} \begin{proof}
    By Proposition \ref{prop:MoreEfficientThanclassical} and the formulas for the submatrices of $\Sigma_{\zeta}$, $\sqrt{\ntot}(\htPPhomOpt-\theta) \xrightarrow{d} \mathcal{N}(0,\sigma_{\textnormal{TPTD}}^2)$ and $\sqrt{\ntot}(\htc-\theta) \xrightarrow{d} \mathcal{N}(0,\sigma_{\textnormal{classical}}^2)$ where $\sigma_{\textnormal{classical}}^2=\var(W \Psi(X))$ and $$\sigma_{\text{TPTD}}^2=\var(W \Psi(X))-\frac{\cov^2(W \Psi(X), W \tilde{\Psi}(\tilde{X}))}{\var(W \tilde{\Psi}(\tilde{X}))+\var(\bar{W} \tilde{\Psi}(\tilde{X}))}.$$ By assumption, in this setting $W=I/\pi_L$ and $\bar{W}=(1-I)/(1-\pi_L)$, and $I \sim \text{Bernoulli}(\pi_L)$ and is independent of $X$ and $\tilde{X}$, and recall $\e[\Psi(X)]=\e[\tilde{\Psi}(\tilde{X})]=0$ by Assumption \ref{assump:AsymptoticLinearity}. Thus $\sigma_{\text{classical}}^2=\var(\Psi(X))/\pi_L$ and $$\begin{aligned} \sigma_{\text{TPTD}}^2 & = \frac{1}{\pi_L } \var( \Psi (X)) -\frac{\cov^2( \Psi (X), \tilde{\Psi}(\tilde{X}))/\pi_L^2}{\var( \tilde{\Psi} (\tilde{X}))/\pi_L+\var( \tilde{\Psi} (\tilde{X}))/(1-\pi_L)} 
    \\ & = \frac{\var(\Psi(X))}{\pi_L} \Big(1- \frac{\cov^2( \Psi (X), \tilde{\Psi} (\tilde{X}))}{ \var( \Psi (X)) \var(\tilde{\Psi} (\tilde{X})) (1+\pi_L/(1-\pi_L)) } \Big).
    \\ & =  \sigma_{\text{classical}}^2 \Big(1- (1-\pi_L) \corr^2( \Psi (X), \tilde{\Psi} (\tilde{X}))  \Big).
\end{aligned}$$ 
\end{proof}


\subsection{Comparison with PPI++}\label{sec:PPI++ComparisonSection}

In this section, we study the asymptotic efficiency of $\htPPhomOpt$ relative to the the estimator from \cite{PPI++} with an optimally chosen tuning parameter $\lambda$, which we denote by $\htPPIplus$. We find that $\htPPIplus$ is not asymptotically equivalent to $\htPPhomOpt$, even when both $\htPPIplus$ and $\htPPhomOpt$ are well defined and asymptotically normal. In Remark \ref{remark:PPI++VersusPTDQuantile} we present an example where $\htPPhomOpt$ is more efficient than $\htPPIplus$ and an example where $\htPPhomOpt$ is less efficient than $\htPPIplus$. We start by presenting the definition of $\htPPIplus$ and show (in Proposition \ref{prop:PPI++Efficiency}) that in the setting of Proposition \ref{prop:AREUnivariateTPPvsClassical} and under similar regularity conditions, $\sqrt{\ntot}(\htPPIplus-\theta) \xrightarrow{d} \mathcal{N}(0,\sigma_{\text{PPI++}}^2)$ where $$\frac{\sigma_{\text{PPI++}}^2}{\sigma_{\text{classical}}^2}=1-(1-\pi_L) \cdot \corr^2(\Psi(X),\Psi(\tilde{X})).$$ 

Let $l_{\vartheta}(\cdot)$ be a loss function satisfying the smoothness Assumption \ref{assump:SmoothEnoughForAsymptoticLineariaty}, which is slightly stronger than the ``Smooth Enough Loss" assumption from \cite{PPI++}. \cite{PPI++} considers solving the $M$-estimation tasks of estimating $\theta=\argmin_{\vartheta \in \mathbb{R}^d} \e[l_{\vartheta}(X)]$ by using estimatiors of the form $$\hat{\theta}_{\lambda}^{\text{PP}} = \argmin_{\vartheta \in \mathbb{R}^d} \Big( \frac{1}{\nc} \sum_{i=1}^\ntot I_i l_{\vartheta}(X_i) + \lambda \cdot \big( \frac{1}{\nm} \sum_{i=1}^\ntot (1-I_i) l_{\vartheta}(\tilde{X}_i) -\frac{1}{\nc} \sum_{i=1}^\ntot I_i l_{\vartheta}(\tilde{X}_i) \big)  \Big).$$

\cite{PPI++} shows that under certain conditions, if $\hat{\lambda}=\lambda+o_p(1)$, $\hat{\theta}_{\hat{\lambda}}^{\text{PP}}$ follows a CLT with asymptotic variance that depends on $\lambda$, and they define $\lambda_*$ to be the $\lambda$ minimizing the trace of the asymptotic variance of $\hat{\theta}_{\hat{\lambda}}^{\text{PP}}$. We define $\htPPIplus$ to be an optimally tuned version of $\hat{\theta}_{\hat{\lambda}}^{\text{PP}}$ in the sense that  $\htPPIplus=\hat{\theta}_{\hat{\lambda}_*}^{\text{PP}}$, where $\hat{\lambda}_* \xrightarrow{p} \lambda_*$.

\begin{proposition}\label{prop:PPI++Efficiency}
    Suppose $d=1$ and each point has an equal probability $\pi_L$ of being labelled (i.e., $\mathbb{P}(I=1 \giv \tilde{X})=\pi_L$ for all $\tilde{X}$). Under Assumptions \ref{assump:SamplingLabelling}, \ref{assump:AsymptoticLinearity}, and \ref{assump:SmoothEnoughForAsymptoticLineariaty}, and under the assumptions that $\hat{\lambda}_* \xrightarrow{p} \lambda_*$, $\htPPIplus \xrightarrow{p} \theta$, and $\htc \xrightarrow{p} \theta$, $\sqrt{\ntot}(\htPPIplus-\theta) \xrightarrow{d} \mathcal{N}(0,\sigma_{\textnormal{PPI++}}^2)$ and $\sqrt{\ntot}(\htc-\theta) \xrightarrow{d} \mathcal{N}(0,\sigma_{\textnormal{classical}}^2)$, where $$\frac{\sigma_{\textnormal{PPI++}}^2}{\sigma_{\textnormal{classical}}^2}=1-(1-\pi_L) \cdot \corr^2(\Psi(X),\Psi(\tilde{X})).$$

\end{proposition}

\begin{proof}
    Note that by assumption, all of the necessary conditions in Theorem 1 of \cite{PPI++} are met with the exception of the condition that $\nc/(\nm) \to r$ for some constant $r \geq 0$. Note that by the strong law of large numbers, $$\big( \frac{\nc}{\ntot}, \frac{\nm}{\ntot}\big) = \big(\frac{1}{\ntot} \sum_{i=1}^\ntot I_i, \frac{1}{\ntot} \sum_{i=1}^\ntot (1-I_i) \big) \xrightarrow{a.s.} \big(\pi_L,1-\pi_L \big).$$  Hence by the continuous mapping theorem, $\nc/(\nm) \xrightarrow{a.s.} \pi_L/(1-\pi_L)= r$, so with probability 1, Theorem 1 of \cite{PPI++} is applicable. Hence by  Theorem 1 of \cite{PPI++}, $$\sqrt{\nc} (\htPPIplus -\theta) \xrightarrow{d} \mathcal{N} \Bigg(0, \frac{r \lambda_*^2 \cov \big( \nabla_{\vartheta} l_{\vartheta}( \tilde{X}) \big|_{\vartheta=\theta} \big) +\cov \big( \nabla_{\vartheta} l_{\vartheta}(X) \big|_{\vartheta=\theta} - \lambda_* \nabla_{\vartheta} l_{\vartheta}( \tilde{X}) \big|_{\vartheta=\theta} \big) }{[\nabla_{\vartheta}^2  \e[l_{\vartheta}(X)] \big|_{\vartheta=\theta} ]^2} \Bigg),$$ where $\nc=\pi_L \ntot +o_p(1)$ is the size of the complete sample. Since $\lambda_*$ minimizes the asymptotic variance across all possible $\lambda$, and the asymptotic variance is quadratic in $\lambda$, $$\lambda_* = \frac{\cov \big( \nabla_{\vartheta} l_{\vartheta}(X) \big|_{\vartheta=\theta}, \nabla_{\vartheta} l_{\vartheta}( \tilde{X}) \big|_{\vartheta=\theta} \big)}{(1+r) \cov \big( \nabla_{\vartheta} l_{\vartheta}( \tilde{X}) \big|_{\vartheta=\theta} \big) }.$$ Plugging in this value for $\lambda_*$ and noting $\sqrt{\ntot}=\sqrt{\nc/\pi_L} +o_p(1)$,  $\sqrt{\ntot} ( \htPPIplus -\theta) \xrightarrow{d} \mathcal{N}(0,\sigma_{\text{PPI++}}^2)$, where $$\sigma_{\text{PPI++}}^2 = \frac{1}{\pi_L [\nabla_{\vartheta}^2  \e[l_{\vartheta}(X)] \big|_{\vartheta=\theta} ]^2} \Big(\cov \big( \nabla_{\vartheta} l_{\vartheta}(X) \big|_{\vartheta=\theta})- \frac{\cov^2 \big( \nabla_{\vartheta} l_{\vartheta}(X) \big|_{\vartheta=\theta}, \nabla_{\vartheta} l_{\vartheta}( \tilde{X}) \big|_{\vartheta=\theta} \big)}{(1+r) \cov \big( \nabla_{\vartheta} l_{\vartheta}( \tilde{X}) \big|_{\vartheta=\theta} \big) } \Big).$$

    Now by Proposition \ref{prop:MestimatorsAsymptoticallyLinear}, $$\Psi(X)= -\frac{ \nabla_{\vartheta} l_{\vartheta}(X) \big|_{\vartheta=\theta}}{\nabla_{\vartheta}^2  \e[l_{\vartheta}(X)] \big|_{\vartheta=\theta}} \quad \text{and} \quad \Psi(\tilde{X}) = -\frac{ \nabla_{\vartheta} l_{\vartheta}(\tilde{X}) \big|_{\vartheta=\theta}}{\nabla_{\vartheta}^2  \e[l_{\vartheta}(X)] \big|_{\vartheta=\theta}}.$$ Hence distributing the factor of $1/[\nabla_{\vartheta}^2  \e[l_{\vartheta}(X)] \big|_{\vartheta=\theta}]^2 = [\nabla_{\vartheta}^2  \e[l_{\vartheta}(X)] \big|_{\vartheta=\theta}]^2/[\nabla_{\vartheta}^2  \e[l_{\vartheta}(X)] \big|_{\vartheta=\theta}]^4$ and using scaling properties of covariances, $$\begin{aligned}
    \sigma_{\text{PPI++}}^2 & = \frac{1}{\pi_L} \Big( \cov(\Psi(X)) - \frac{\cov^2(\Psi(X), \Psi(\tilde{X}))}{(1+r) \cov(\Psi(\tilde{X}) )} \Big)
    \\ & = \frac{\var( \Psi(X))}{\pi_L}  \Big(1- \frac{\corr^2(\Psi(X), \Psi(\tilde{X}))}{(1+r) }  \Big).
\end{aligned}$$ Finally by Lemma \ref{lemma:JointCLTzeta}, $\sqrt{\ntot}(\htc-\theta) \xrightarrow{d} \mathcal{N}(0,\sigma_{\text{classical}}^2)$, where $\sigma_{\text{classical}}^2=\var(W\Psi(X))= \var(\Psi(X))/\pi_L$. Combining this with the previous result and recalling $r=\pi_L/(1-\pi_L) \Rightarrow (1+r)^{-1} = 1-\pi_L$, $$\sigma_{\textnormal{PPI++}}^2= \sigma_{\textnormal{classical}}^2 \Big(1-(1-\pi_L) \cdot \corr^2(\Psi(X),\Psi(\tilde{X})) \Big).$$
\end{proof}

One consequence of Proposition \ref{prop:PPI++Efficiency} is that the asymptotic variance $\htPPIplus$ therefore does not depend on $\tilde{\Psi}(\cdot)$ whereas the asymptotic variance of $\htPPhomOpt$ does. Also note that when the conditions of Propositions \ref{prop:AREUnivariateTPPvsClassical} and \ref{prop:PPI++Efficiency} are both met, the asymptotic relative efficiency between $\htPPIplus$ and $\htPPhomOpt$ is given by \begin{equation}\label{eq:PPI++VersusTPPARE}
    \frac{\sigma_{\text{PPI++}}^2}{\sigma_{\text{TPTD}}^2} = \frac{1-(1-\pi_L) \cdot \corr^2(\Psi(X),\Psi(\tilde{X}))}{1-(1-\pi_L) \cdot \corr^2(\Psi(X),\tilde{\Psi}(\tilde{X}))}.
\end{equation}

The following remark shows that the asymptotic efficiency ratio in Equation \eqref{eq:PPI++VersusTPPARE} is smaller than $1$ in some settings and larger than $1$ in other settings. This finding contradicts the claims in Proposition 1 of \cite{MiaoLuNeurIPS} that $\htPPhomOpt$ is always more efficient than the PPI++ estimator from \cite{PPI++}. At the time of this writing, we think that  Proposition 1 in \cite{MiaoLuNeurIPS} is incorrect.

\begin{remark}[Comparison of $\htPPIplus$ and $\htPPhomOpt$ for quantile estimation]\label{remark:PPI++VersusPTDQuantile} Let $F$ and $\tilde{F}$ denote the CDFs of $X$ and $\tilde{X}$, respectively, and for a fixed $q \in (0,1)$, suppose our estimand $\theta=F^{-1}(q)$ is the $q$th quantile of $X$. Further suppose that $F$ is differentiable at $\theta$ with positive derivative and that similarly $\tilde{F}$ is differentiable at $\tilde{F}^{-1}(q)$ with positive derivative. By Corollary 21.5 in \cite{VanderVaartTextbook} it is clear that for some finite constants $\tilde{c}_0$, $\tilde{c}_1 \neq 0$, $c_0$, and $c_1 \neq 0$ that just depend on $F$ and $\tilde{F}$ but not the observed data, the functions $\Psi,\tilde{\Psi}: \mathbb{R} \to \mathbb{R}$ given by $\Psi(x)=c_0+c_1 I\{ x \leq F^{-1}(q) \}$ and $\tilde{\Psi}(x)=\tilde{c}_0+\tilde{c}_1 I\{ x \leq \tilde{F}^{-1}(q) \}$ satisfy Assumption \ref{assump:AsymptoticLinearity}. Therefore in this setting Equation \eqref{eq:PPI++VersusTPPARE}, simplifies to $$\frac{\sigma_{\text{PPI++}}^2}{\sigma_{\text{TPTD}}^2} = \frac{1-(1-\pi_L) \cdot \corr^2 \big(I\{X \leq F^{-1}(q)\},I\{ \tilde{X} \leq F^{-1}(q)\} \big)}{1-(1-\pi_L) \cdot \corr^2\big(I\{X \leq F^{-1}(q)\} , I\{ \tilde{X} \leq \tilde{F}^{-1}(q)\} \big)}.$$
Note that if $\tilde{X}=g(X)$ for a fixed monotone strictly increasing function $g$, then $X \leq F^{-1}(q) \Leftrightarrow \tilde{X} \leq \tilde{F}^{-1}(q)$ in which case $\corr^2\big(I\{X \leq F^{-1}(q)\} , I\{ \tilde{X} \leq \tilde{F}^{-1}(q)\}\big)=1$, but $\corr^2\big(I\{X \leq F^{-1}(q)\} , I\{ \tilde{X} \leq F^{-1}(q)\} \big) \leq 1$ and hence $\sigma_{\text{TPTD}}^2 \leq \sigma_{\text{PPI++}}^2$ (with a strict a inequality for most nontrivial cases of strictly increasing functions $g$). Meanwhile if $X \sim \mathcal{N}(0,0.25^2)$, $\tilde{X}=X+E$ where $E \indep X$ with $E \sim \mathcal{N}(0,1)$ and $q=0.99$, we present very strong numerical evidence in Figure \ref{fig:PPI++MoreEfficientExample} that $\sigma_{\text{PPI++}}^2< \sigma_{\text{TPTD}}^2$.
\end{remark}

\begin{figure}[t]
    \centering 
    \includegraphics[width=0.95 \hsize]{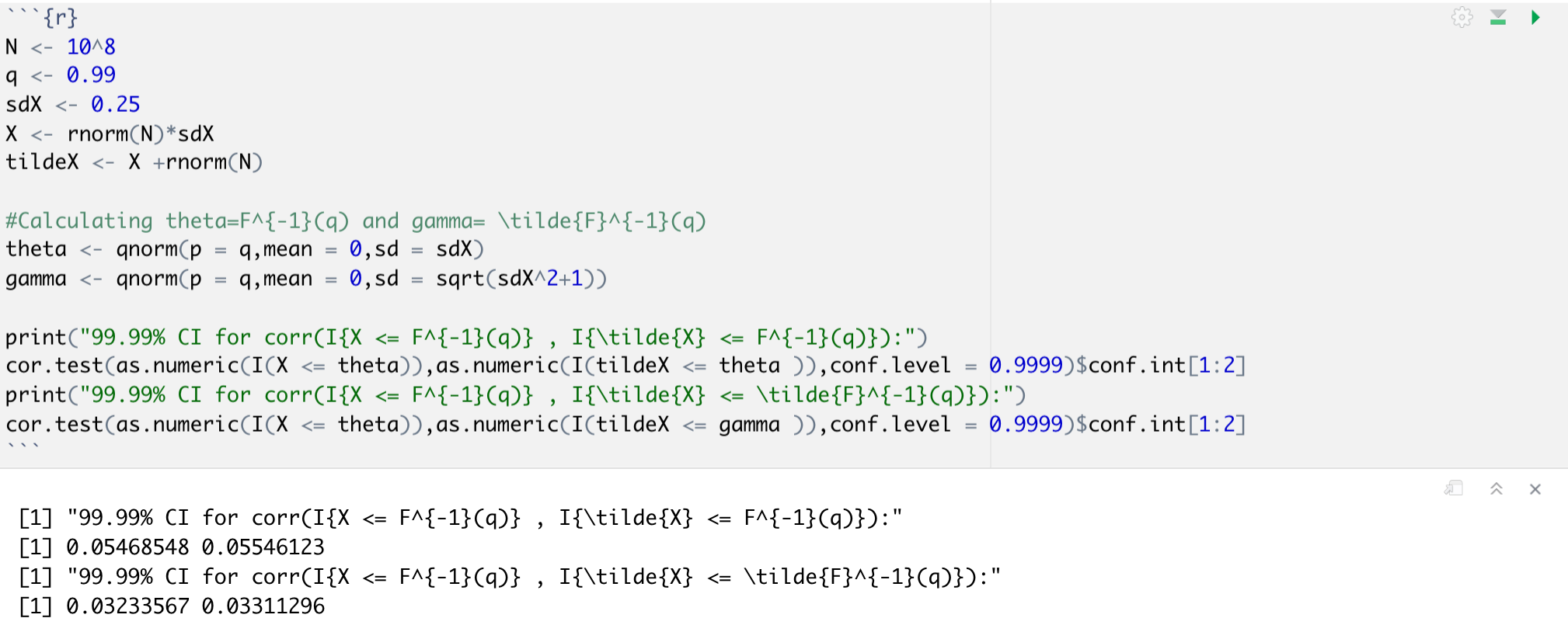}
    \caption{Letting $X \sim \mathcal{N}(0,0.25^2)$, $\tilde{X}=X+E$ where $E \indep X$ with $E \sim \mathcal{N}(0,1)$ and $q=0.99$, the following R code generates $10^8$ samples of $(X,\tilde{X})$ and uses these $10^8$ samples to construct 99.99\% confidence intervals for $\corr(I\{X \leq F^{-1}(q) \}, I\{\tilde{X} \leq F^{-1}(q) \})$ and $\corr(I\{X \leq F^{-1}(q) \}, I\{ \tilde{X} \leq \tilde{F}^{-1}(q) \})$. These confidence intervals are far apart giving strong numerical evidence that $\corr(I\{X \leq F^{-1}(q) \}, I\{\tilde{X} \leq F^{-1}(q) \})>\corr(I\{X \leq F^{-1}(q) \}, I\{ \tilde{X} \leq \tilde{F}^{-1}(q) \})$, which by Equation \eqref{eq:PPI++VersusTPPARE} further implies that $\sigma_{\text{PPI++}}^2 < \sigma_{\text{TPTD}}$.}
    \label{fig:PPI++MoreEfficientExample}
\end{figure}

\subsection{Comparison with PTD using \texorpdfstring{$\hga$}{}}\label{sec:PTDusingGammaHatAll}

In this section, we consider the efficiency of PTD estimators that use all samples rather than just the incomplete samples to construct the first term of the PTD estimator. This variant of the PTD estimator was used in \cite{ChenAndChen2000}, \cite{YangAndDing2020}, \cite{kremers2021generalsimplerobustmethod},  and \cite{gronsbell2024PromotesCC}, among others. In particular, let $\hga=\calA(\tilde{\cx}; \bm{1})$ (i.e., let $\hga$ be the unweighted estimator of $\gamma$ that uses all samples $(\tilde{X}_i)_{i=1}^{\ntot}$) and consider PTD estimators of the form
\begin{equation}\label{eq:PTDwithGammaHatAll}
    \htPPhom_a= \hat{\Omega} \hga + (\htc - \hat{\Omega} \hgc ).
\end{equation} To study the asymptotic efficiency of this estimator we introduce an asymptotic linearity assumption for $\hga$ that will generally hold in settings where Assumption \ref{assump:AsymptoticLinearity} holds.

\begin{assumption}[Asymptotic linearity of $\hga$]\label{assump:AsymptoticLinearityGammaHatAll} Suppose $\sqrt{\ntot} \Big( \hga -\gamma - \frac{1}{\ntot} \sum_{i=1}^\ntot  \tilde{\Psi}(\tilde{X}_i) \Big) \xrightarrow{p} 0$, where $\tilde{\Psi}:\mathbb{R}^p \to \mathbb{R}^d$ is the same function as the one introduced in Assumption \ref{assump:AsymptoticLinearity}.
\end{assumption}

The following proposition shows that in the special case where the labelling weights are constant (i.e., the when the complete sample is a uniform random subsample), $\htPPhom_a$ with an optimally chosen tuning matrix $\hat{\Omega}$ will have the same asymptotic variance as that of optimally tuned PTD estimator studied in the main text. In particular, the asymptotic variance given in the following proposition matches that of $\htPPhomOpt$, given in Proposition \ref{prop:MoreEfficientThanclassical}.

\begin{proposition}
    Suppose that Assumption \ref{assump:SamplingLabelling} holds with $\pi(\tilde{X})=\pi_L \in (0,1)$ for all $\tilde{X}$ and that $\htc$, $\hgc$, and $\hga$ are all asymptotically linear in the sense of Assumptions \ref{assump:AsymptoticLinearity} and \ref{assump:AsymptoticLinearityGammaHatAll}. If $\hat{\Omega} \xrightarrow{p} \Omega$, then $\sqrt{\ntot}( \htPPhom_a-\theta) \xrightarrow{d} \mathcal{N} \big(0,\Sigma_{\textnormal{PTD},a}(\Omega) \big)$, where $\Sigma_{\textnormal{PTD},a}(\Omega)$ is quadratic in $\Omega$. Moreover, letting $\Omega_{\textnormal{opt},a}$ be the matrix minimizing the diagonal entries of $\Sigma_{\textnormal{PTD},a}(\Omega)$, $$\Sigma_{\textnormal{PTD},a}(\Omega_{\textnormal{opt},a})=\SigTC - \SigTGC (\SigGC+\SigGM)^{-1} [ \SigTGC]^\tran.$$
\end{proposition}

\begin{proof}
    Let $\hat{\zeta}^a=(\htc,\hgc,\hga)$ and note that by an identical argument to that seen in the proof of Lemma \ref{lemma:JointCLTzeta}, $\sqrt{\ntot}(\hat{\zeta}^a-\zeta) \xrightarrow{d} \mathcal{N}(0,\Sigma_a)$ where $$\Sigma_a=\begin{bmatrix}
    \SigTC & \SigTGC & A_1 \\
    [ \SigTGC]^\tran & \SigGC & A_2 \\
    A_1^\tran & A_2 & A_2
    \end{bmatrix},$$ where $A_1=\cov \big( W \Psi(X), \tilde{\Psi}(\tilde{X}) \big)$ and $A_2=\cov \big( W \tilde{\Psi}(\tilde{X}), \tilde{\Psi}(\tilde{X}) \big)=\var(\tilde{\Psi}(\tilde{X}))$ (the simplification for $A_2$ follows from the tower property, Assumption \ref{assump:SamplingLabelling}, and the fact that $\e[\tilde{\Psi}(\tilde{X})]=0$). Letting $B_{\hat{\Omega}}=\begin{bmatrix}
        I_{d \times d} & -\hat{\Omega} & \hat{\Omega}
    \end{bmatrix}$ and $B_{\Omega}=\begin{bmatrix}
        I_{d \times d} & -\Omega & \Omega
    \end{bmatrix}$, the previous result and Slutsky's lemma give $$\sqrt{\ntot}(\htPPhom_a -\theta) =B_{\hat{\Omega}} \big(\sqrt{\ntot}(\hat{\zeta}^a -\zeta) \big)=B_{\Omega} \big(\sqrt{\ntot}(\hat{\zeta}^a -\zeta) \big)+o_p(1) O_p(1) \xrightarrow{d} \mathcal{N}\big(0, \Sigma_{\textnormal{PTD},a}(\Omega) \big),$$ where $\Sigma_{\textnormal{PTD},a}(\Omega)=B_{\Omega} \Sigma_a B_{\Omega}^\tran$. Therefore, the asymptotic variance of $\htPPhom_a$ is a function of $\Omega$ given by $$\Sigma_{\textnormal{PTD},a}(\Omega)= \SigTC +(A_1-\SigTGC) \Omega^\tran + \Omega (A_1-\SigTGC)^\tran + \Omega (\SigGC-A_2) \Omega^\tran.$$ For each $j \in \{1,\dots,d\}$, note that $[\Sigma_{\textnormal{PTD},a}(\Omega)]_{jj}$ only depends on the $j$th row of $\Omega$, which we denote by $\Omega_{j \cdot} \in \mathbb{R}^d$. Further, $[\Sigma_{\textnormal{PTD},a}(\Omega)]_{jj}$ is a quadratic function in $\Omega_{j \cdot}$ that is minimized when $\Omega_{j \cdot}=(\SigGC -A_2)^{-1} (\SigTGC-A_1)^\tran e_j$. Hence, letting $$
    \Omega_{\textnormal{opt},a} \equiv (\SigTGC-A_1) (\SigGC -A_2)^{-1},$$ and setting $\Omega=\Omega_{\text{opt},a}$ simultaneously minimizes each diagonal entry of $\Sigma_{\textnormal{PTD},a}(\Omega)$. Now observe that \begin{equation}\label{eq:VoptGammaALLGeneral}
        \Sigma_{\textnormal{PTD},a}(\Omega_{\text{opt},a})  = \SigTC - (\SigTGC-A_1) (\SigGC -A_2)^{-1} (\SigTGC-A_1)^\tran.
    \end{equation} This formula further simplifies since we assume $\pi(\tilde{X})=\pi_L \in (0,1)$ for all $\tilde{X}$. In particular, by the tower property, Assumption \ref{assump:SamplingLabelling}, and the fact that $\tilde{\Psi}(\tilde{X})$ is mean $0$, $$\SigGC=\var(W \tilde{\Psi}(\tilde{X}))=\e[(I/\pi_L^2) \tilde{\Psi}(\tilde{X})  \tilde{\Psi}(\tilde{X})^\tran]=\var( \tilde{\Psi}(\tilde{X}))/\pi_L =A_2/\pi_L.$$ Identical arguments show that $\SigGM=A_2/(1-\pi_L)$ and that $\SigTGC=A_1/\pi_L$. Thus, $(\SigTGC-A_1)=(1-\pi_L)\SigTGC$ and $(\SigGC-A_2)=(1-\pi_L)^2 (\SigGC + \SigGM)$. Plugging these into formula \eqref{eq:VoptGammaALLGeneral} gives the desired result. 
\end{proof}

    \section{Further details on data-based experiments}\label{sec:FurtherDetailsOnExperiments}

    \subsection{AlphaFold error-in-response regressions}\label{sec:AlphaFoldExpsDetails}

We test our method for a logistic regression model that assesses whether there is an interaction between two types of posttranslational modifications (PTMs) as predictors of whether a protein region is an internally disordered region (IDR).  To do this, we used a dataset of $M=10{,}802$ samples that originated from \cite{bludau2022structural}, was used in previous prediction-powered inference studies \citep{OriginalPPI,PPI++}, and was downloaded from Zenodo \citep{PPIZenodo}. Each sample in the dataset had 3 PTMs which we denote by the vector $Z$, an indicator of IDR which we call $Y_{\text{IDR}} \in \{0,1\}$, and a prediction of $Y_{\text{IDR}}$ based on AlphaFold \citep{AlphaFoldPaper}, which we call $\tilde{Y}_{\text{IDR}} \in \{0,1\}$.

We suppose the investigator is interested in estimating the parameters $\theta=(\beta_0,\beta_1,\beta_2,\beta_3)$ that give the best fit to the following logistic regression model $$\mathbb{P}\big(Y_{\text{IDR}}=1 \giv 
 Z \big) = \frac{\exp \big( \beta_0+\beta_1 Z_{\text{Acet}}+\beta_2 Z_{\text{Ubiq}}+\beta_3 Z_{\text{Acet}} Z_{\text{Ubiq}}  \big)}{1+\exp \big( \beta_0+\beta_1 Z_{\text{Acet}}+\beta_2 Z_{\text{Ubiq}}+\beta_3 Z_{\text{Acet}} Z_{\text{Ubiq}}  \big)},$$ where $Z_{\text{Ubiq}} \in \{0,1\}$ and $Z_{\text{Acet}} \in  \{0,1\}$ are indicators of whether the sample was subject to ubiquitination and acetylation, respectively. We further suppose that the investigator only has the budget to obtain $\ntot=7{,}500$ samples and $\sim$ $1{,}000$ measurements of $Y_{\text{IDR}}$ and is most interested in estimating the interaction coefficient $\beta_3$. Since only about 11\% of samples had $Z_{\text{Acet}}=1$ selecting the complete sample via uniform random sampling would lead to a stark imbalance in the number of complete samples with each of the four possible values of $(Z_{\text{Ubiq}},Z_{\text{Acet}}) \in \{0,1\}^2$ and would therefore be suboptimal for $\beta_3$ estimation. Instead, in our simulations, we consider the case where the investigator selects which points belong to the complete sample using independent Bernoulli draws with probabilities that depend only on $Z_{\text{Ubiq}}$ and $Z_{\text{Acet}})$ such that for each of the four combinations of $Z_{\text{Ubiq}}$ and $Z_{\text{Acet}}$, the expected number of points in the complete sample with that combination is $250$ (and the expected complete sample size is $1{,}000$). Various methods for estimating $\theta$ and its confidence intervals are then deployed, and this process of randomly selecting the complete sample and implementing various inference approaches for $\theta$ are repeated across $500$ simulations. The results are displayed in the top panels of Figures \ref{fig:ViolinPlotFigure}, \ref{fig:ExperimentsVaryingCIMethod}, and \ref{fig:ExperimentsVaryingTuningMatrix}.


\subsection{Housing price error-in-covariate regressions}\label{sec:HousingPriceExample}

We downloaded a dataset of ground truth observations and remote sensing-based estimates for housing price, income, nightlights, road length, tree cover, elevation, and population from \cite{MOSAIKSSourceCode}, which was studied in \cite{ProctorPaper} and \cite{KerriRSEPaper} to evaluate methods for leveraging error-prone predictions that arise in remote sensing settings. This dataset was synthesized in \cite{MOSAIKSPaper}, where a method called MOSAIKS was developed to produce satellite-based estimates of various quantities of interest. In summary, \cite{MOSAIKSPaper} produced remote sensing-based estimates of these 7 variables by using unsupervised learning approaches to covert daytime satellite imagery for each $\sim 1 \text{km} \times 1 \text{km}$ grid cell to $8{,}912$ features and subsequently used ridge regression models trained on a small labelled dataset to predict these 7 variables from the $8{,}912$ remotely sensed features. We remark that the ``ground truth" observations for 2 of these 7 variables were also based on remote sensing data and may have been somewhat error-prone; however, these ground truth observations are still thought to be more accurate than the MOSAIKS-based predictions. For example, the ground truth nightlight estimates were derived from nighttime satellite imagery whereas the MOSAIKS-based predictions were derived from daytime satellite imagery. Nonetheless, because our focus is on testing PTD-based methods on a real dataset rather than reporting scientifically meaningful regression coefficients, we ignore possible errors in these ground truth observations.

We test the PTD estimator and our methods for constructing confidence intervals on linear regression and quantile regression models regressing housing price on income, nightlights, and road length. Each of the $M=46{,}418$ samples corresponded to a distinct $\sim 1 \text{km} \times 1 \text{km}$ grid cell. Keeping the variable transformations and units from \cite{MOSAIKSPaper}, the response variable $Y_{\text{Housing\$}}$ was the log of the housing price (in units of \$/$\text{ft}^2$) averaged over the grid cell, and the covariates $Z_{\text{Income}}$, $Z_{\text{Nightlights}}$, and $Z_{\text{RoadLength}}$ were the averages across the grid cell of income (measured in units of \$/household), nightlights (measured in units of $\log(1+\text{nW}/\text{cm}^2 \cdot \text{sr})$), and road length (measured in units of meters per grid cell). We supposed the investigator was interested in estimating the best fit $\theta=(\beta_0,\beta_1,\beta_2,\beta_3)$ in the following linear model $$Y_{\text{Housing\$}}= \beta_0+ \beta_1 Z_{\text{Income}} + \beta_2 Z_{\text{Nightlights}} +\beta_3 Z_{\text{RoadLength}} + \varepsilon.$$ We further suppose that the investigator only has the budget to obtain $\ntot=5{,}000$ samples with $\sim$ 500 of them having measurements of $Z_{\text{Nightlights}}$ and $Z_{\text{RoadLength}}$, but with all of the $\ntot$ samples having MOSAIKS-based proxies for $\tilde{Z}_{\text{Nightlights}}$ and $\tilde{Z}_{\text{RoadLength}}$. We simulate the scenario where in expectation the investigator obtains 500 measurements of $Z_{\text{Nightlights}}$ and $Z_{\text{RoadLength}}$ and use a variety of methods for constructing confidence intervals for the components of $\theta$. The results across 500 such simulations are displayed the rows of Figures \ref{fig:ViolinPlotFigure}, \ref{fig:ExperimentsVaryingCIMethod}, and \ref{fig:ExperimentsVaryingTuningMatrix} with the signifier `(2)'. These experiments were repeated when modifying the above linear model to be quantile regression (with $q=0.5$) and modifying the budget to allow for $\sim 1{,}000$ complete samples (see the rows with the signifier `(3)' from the same figures). The quantile regressions were implemented using the \texttt{quantreg} R package \citep{QuantRegPackage}. For the classical approach in the quantile regression experiments, we used the normal approximations and standard errors based on the Powell kernel estimate of the sandwich covariance formula for $\var(\htc)$.

\subsection{Tree cover error-in-both regressions}\label{sec:TreecoverExpsDetails}

We test our methods on a number of regressions relating tree cover to elevation and population where all data and proxies were taken from \cite{MOSAIKSSourceCode} described in the previous subsection. Each of the $M=67{,}968$ samples corresponded to a distinct $\sim 1 \text{km} \times 1 \text{km}$ grid cell. Keeping the variable transformations and units from \cite{MOSAIKSPaper}, the response variable $Y_{\text{TreeCover}}$ was the percent forest cover in the grid cell, and the covariates $Z_{\text{Elevation}}$ and $Z_{\text{Population}}$ was the average elevation (measured in meters) and population (measured in $\log(1+\text{people}/\text{km}^2)$) across the grid cell. We supposed the investigator was interested in estimating the best fit $\theta=(\beta_0,\beta_1,\beta_2)$ in the following linear model $$Y_{\text{TreeCover}}= \beta_0+ \beta_1 Z_{\text{Elevation}} + \beta_2 Z_{\text{Population}} + \varepsilon.$$ We suppose the investigator has access to $Z_{\text{Elevation}}$ and the MOSAIKS-based proxies $\tilde{Z}_{\text{TreeCover}}$ and $\tilde{Z}_{\text{Population}}$ on $\ntot=5{,}000$ samples but is only able to measure $Z_{\text{Population}}$ and $Y_{\text{TreeCover}}$ on a limited number of samples. The rows in Figures \ref{fig:ViolinPlotFigure}, \ref{fig:ExperimentsVaryingCIMethod}, and \ref{fig:ExperimentsVaryingTuningMatrix} with the signifier `(4)' show the results from $500$ simulations in which the investigator selects a complete sample of size $500$ in expectation using IID Bernoulli trials. 

We then conduct simulations for a setup with the same variables and model specification, with the caveat being that a cluster sampling scheme was used to select the $\ntot$ samples and another cluster sampling scheme was used to select which of the $\ntot$ samples belong to the complete sample. In particular, we treated each 0.5° $\times$ 0.5° grid cell as a cluster and conducted $500$ simulations where in each simulation, all samples from approximately $502.4$ clusters were used (such that in expectation the investigator had access to $\ntot=10{,}000$ samples) and IID Bernoulli trials determined which clusters were assigned to the complete sample (in expectation, precisely $1{,}000$ samples and approximately $50.24$ clusters were assigned to the complete sample in each simulation). The rows in Figures \ref{fig:ViolinPlotFigure}, \ref{fig:ExperimentsVaryingCIMethod}, and \ref{fig:ExperimentsVaryingTuningMatrix} with the signifier `(5)' show the results across these simulations when using Algorithm \ref{alg:ClusterBootstrap}, which constructs confidence intervals that account for the cluster labelling and sampling scheme. For these rows the ``classical approach" uses the normal approximation and \texttt{vcovCL} in R to account for the clusters when estimating $\var(\htc)$, which led to slight undercoverage that we suspect is due to the small number of labelled clusters. In Figure \ref{fig:ExperimentsVaryingCIMethod}, we also consider the convolution-based speed-up to Algorithm \ref{alg:ClusterBootstrap} that is described at the end of Section \ref{sec:ClusterBootstrapDescription} (we refer to this speed-up as Algorithm \ref{alg:QuickConvolutionBootstrap}/\ref{alg:ClusterBootstrap} in Table \ref{table:ExperimentSummary}).

We then consider the case where the investigator is interested in estimating the best fit $\theta=(\beta_0,\beta_1,\beta_2)$ to the following logistic regression model $$\mathbb{P}(Y_{\text{TreeCover}} > 10 \% \giv Z) = \frac{\exp\big( \beta_0+ \beta_1 Z_{\text{Elevation}} + \beta_2 Z_{\text{Population}} \big)}{1+\exp\big( \beta_0+ \beta_1 Z_{\text{Elevation}} + \beta_2 Z_{\text{Population}} \big)}.$$ Note that the threshold of 10\% is meaningful from a forestry perspective, as both the Food and Agriculture Organization of the United Nations and the 2017 Forest Resources Report of the United States defines forested land as land that both meets certain criteria and has greater than 10\% tree cover \citep{UDSAForestServiceReport}. We again suppose the investigator has access to $Z_{\text{Elevation}}$ and the MOSAIKS-based proxies $\tilde{Z}_{\text{TreeCover}}$ and $\tilde{Z}_{\text{Population}}$ on all $\ntot$ samples but is only able to measure $Z_{\text{Population}}$ and $Y_{\text{Tree cover}}$ on a limited number of samples. The rows in Figures \ref{fig:ViolinPlotFigure}, \ref{fig:ExperimentsVaryingCIMethod}, and \ref{fig:ExperimentsVaryingTuningMatrix} with the signifier `(6)' show the results from $500$ simulations in which the investigator has $\ntot=8{,}000$ samples and selects a complete subsample of size $1{,}000$, in expectation, using IID Bernoulli trials.

\subsection{Census error-in-covariate stratified sampling}\label{sec:CensusExpsDetails}

We test the stratified bootstrap method (Algorithm \ref{alg:StratifiedBootstrap}) for a linear regression model assessing the association between disability and income while controlling for age in a setting where stratified sampling and labelling is assumed. To do this we used data from a 2019 US Census survey (downloaded via the Folktables interface \citep{FolktablesPaper}) of $M=200{,}227$ individuals aged 25--64 who were living in California and had reported their income $Y_{\text{Income}}$ (measured in \$). The covariates were age $Z_{\text{Age}}$ and an indicator $Z_{\text{Disability}} \in \{0,1\}$ of whether the subject had a disability. In addition, we trained a Random Forest model on the 2018 California Census (also downloaded via the Folktables interface \citep{FolktablesPaper}) to predict $Z_{\text{Disability}}$ using 16 features from the census including age, sex, race, income, marital status, military service history, and indicators of hearing, vision, and cognitive difficulties. We applied the Random Forest model to these same features to produce predictions $\tilde{Z}_{\text{Disability}}$ of disability status of the $M=200{,}227$ subjects from the 2019 dataset, and these predictions had an overall accuracy of about $0.9535$.

We supposed that the investigator was interested in estimating the best fit $\theta=(\beta_0,\beta_1,\beta_2)$ in the following linear model $$Y_{\text{Income}}=\beta_0+\beta_1 Z_{\text{Age}}+\beta_2 Z_{\text{Disability}} +\varepsilon.$$ We suppose the investigator has the budget to conduct a survey obtaining precisely $\nm=5{,}000$ samples of $(Y_{\text{Income}}, Z_{\text{Age}},\tilde{Z}_{\text{Disability}})$ and another survey obtaining precisely $\nc=1{,}000$ samples of $(Y_{\text{Income}}, Z_{\text{Age}}, Z_{\text{Disability}},\tilde{Z}_{\text{Disability}})$. We further suppose that the investigator conducts these two surveys from the population of $M=200{,}227$ using a stratified sampling scheme with $K=4$ strata and the same number of subjects surveyed for each strata (precisely $1{,}250$ and $250$ subjects were sampled per strata for the incomplete sample survey and complete sample survey, respectively). The four strata were determined by binning the subjects ages into the following 4 categories 25--34, 35--44, 45--54, and 55--64. We simulate the scenario where the investigator obtains these two stratified samples 500 times. The results for the classical approach versus Algorithm \ref{alg:StratifiedBootstrap} can be found in the bottom panels of Figures \ref{fig:ViolinPlotFigure}, \ref{fig:ExperimentsVaryingCIMethod}, and \ref{fig:ExperimentsVaryingTuningMatrix}. For these rows the confidence intervals for the classical approach used a bootstrap approach that was equivalent to running Algorithm \ref{alg:StratifiedBootstrap} with $\hat{\Omega}=0$.

\end{document}